\documentclass[11pt,a4paper, pagebackref]{article}

\usepackage{fullpage}

\usepackage{authblk}

\usepackage{amsmath,amssymb,amsfonts,amsthm,dsfont}
\usepackage{hyperref}
\usepackage{mathtools}
\usepackage[capitalize]{cleveref}
\newenvironment{claimproof}{\paragraph{Proof of Claim:}}{\hfill$\blacksquare$}

\newtheorem{theorem}{Theorem}[section]
\newtheorem{observation}[theorem]{Observation}
\newtheorem{proposition}[theorem]{Proposition}
\newtheorem{claim}{Claim}
\newtheorem*{claim*}{Claim}
\crefname{claim}{Claim}{Claims}

\newtheorem{lemma}[theorem]{Lemma}
\theoremstyle{definition}
\newtheorem{example}{Example}
\newtheorem{definition}{Definition}
\newtheorem{remark}{Remark}

\usepackage[utf8]{inputenc}
\usepackage{tikz}
\usepackage{bm}
\usepackage{algorithm,algorithmic}
\usetikzlibrary{calc,backgrounds,fit, shapes}

\tikzstyle{vertex}=[draw, circle, fill, inner sep = 2pt]

\usepackage{framed}
\usepackage{tabularx}

\newlength{\RoundedBoxWidth}
\newsavebox{\GrayRoundedBox}
\newenvironment{GrayBox}[1]%
   {\setlength{\RoundedBoxWidth}{.93\columnwidth}
    \def\boxheading{#1}
    \begin{lrbox}{\GrayRoundedBox}
       \begin{minipage}{\RoundedBoxWidth}}%
   {   \end{minipage}
    \end{lrbox}
    \begin{center}
    \begin{tikzpicture}%
       \node(Text)[draw=black!20,fill=white,rounded corners,inner sep=2ex,text width=\RoundedBoxWidth]
             {\usebox{\GrayRoundedBox}};
        \coordinate(x) at (current bounding box.north west);
        \node [draw=white,rectangle,inner sep=3pt,anchor=north west,fill=white]
        at ($(x)+(6pt,.75em)$) {\boxheading};
    \end{tikzpicture}
    \end{center}}

\newenvironment{defproblemx}[1]{\noindent\ignorespaces%
                                \FrameSep=6pt%
                                \parindent=0pt%
                \begin{GrayBox}{#1}%
                \begin{tabular*}{\columnwidth}{!{\extracolsep{\fill}}@{\hspace{.1em}} >{\itshape} p{1.1cm} p{0.89\columnwidth} @{}}%
            }{
                 \vspace*{-1em}
                 \end{tabular*}%
                \end{GrayBox}%
                \ignorespacesafterend
            }

\newcommand{\defProblemTask}[3]{%
  \begin{defproblemx}{#1}
    Input: & #2 \\
    Task: & #3
  \end{defproblemx}
}

\newcommand{\mdsmtuples}{\textsc{MDSR-ML-Sets}}
\newcommand{\tdsmtuples}{\textsc{3DSR-ML-Sets}}
\newcommand{\mdsr}{\textsc{MDSR}}
\newcommand{\tdsrp}{\textsc{3DSR-Poset}}

\newcommand{\disunion}{\mathbin{\dot{\cup}}}
\DeclareMathOperator{\ML}{{ML}}
\newcommand{\LogicTRUE}{\mathsf{true}}
\newcommand{\LogicFALSE}{\mathsf{false}}
\newcommand{\NP}{\ensuremath{\mathsf{NP}}}
\newcommand{\Wone}{\ensuremath{\mathsf{W[1]}}}
\newcommand{\FPT}{\ensuremath{\mathsf{FPT}}}
\newcommand{\XP}{\ensuremath{\mathsf{XP}}}

\newcommand{\upperBoundDistance}{2\kappa d^2 +  4\kappa + 3d +1}

\DeclareMathOperator{\instable}{instable}

\DeclareMathOperator{\pend}{
    \succ 
    \overset{\raise0.3em\hbox{\text{\scriptsize{(rest)}}}}{\ldots}}

\newcommand{\barc}{\bar{c}}
\newcommand{\bard}{\bar{d}}

\newcommand{\barf}{\bar{f}}
\newcommand{\barg}{\bar{g}}

\newcommand{\dsmpo}{\textsc{MDSR-Poset}\xspace}
\newcommand{\CO}{\text{CO}}
\newcommand{\betterEdgese}{E^{\succeq e}}

\usepackage{xspace}

\usepackage[misc]{ifsym}
\usepackage{ stmaryrd }

\newcommand{\succml}{\succ_{\ML}}

\allowdisplaybreaks

\newcommand{\convexpath}[2]{
[
    create hullnodes/.code={
        \global\edef\namelist{#1}
        \foreach [count=\counter] \nodename in \namelist {
            \global\edef\numberofnodes{\counter}
            \node at (\nodename) [draw=none,name=hullnode\counter] {};
        }
        \node at (hullnode\numberofnodes) [name=hullnode0,draw=none] {};
        \pgfmathtruncatemacro\lastnumber{\numberofnodes+1}
        \node at (hullnode1) [name=hullnode\lastnumber,draw=none] {};
    },
    create hullnodes
]
($(hullnode1)!#2!-90:(hullnode0)$)
\foreach [
    evaluate=\currentnode as \previousnode using \currentnode-1,
    evaluate=\currentnode as \nextnode using \currentnode+1
    ] \currentnode in {1,...,\numberofnodes} {
  let
    \p1 = ($(hullnode\currentnode)!#2!-90:(hullnode\previousnode)$),
    \p2 = ($(hullnode\currentnode)!#2!90:(hullnode\nextnode)$),
    \p3 = ($(\p1) - (hullnode\currentnode)$),
    \n1 = {atan2(\y3,\x3)},
    \p4 = ($(\p2) - (hullnode\currentnode)$),
    \n2 = {atan2(\y4,\x4)},
    \n{delta} = {-Mod(\n1-\n2,360)}
  in
    {-- (\p1) arc[start angle=\n1, delta angle=\n{delta}, radius=#2] -- (\p2)}
}
-- cycle
}

\newcommand{\myone}{\textsf{true}}
\newcommand{\myzero}{\textsf{false}}

\begin{document}

\title{Multidimensional Stable Roommates\\ with Master List\thanks{Main work done while all authors were affiliated with TU~Berlin.}\thanks{An extended abstract of this work appears in
the \emph{Proceedings of the 16th International Conference on Web and Internet Economics}~\cite{DBLP:conf/wine/BredereckHKN20}. This full version now contains full proofs of all results.}}

\newcommand{\acktext}{
KH was supported by DFG Research Training Group 2434 ``Facets of Complexity''.
DK was partially supported by DFG project NI 369/19 while at TU~Berlin.
We thank the anonymous reviewers of the conference version of this paper.
}

\author[1]{Robert~Bredereck}
\author[2]{Klaus~Heeger}
\author[3]{Du\v{s}an~Knop}
\author[2]{Rolf~Niedermeier}

\affil[1]{
  Humboldt-Universit\"at zu Berlin, Institut f\"ur Informatik, Algorithm Engineering, Germany\protect\\
  \texttt{robert.bredereck@hu-berlin.de}
}
\affil[2]{
	Technische Universit\"at Berlin, Fakult\"{a}t~IV, Algorithmics and Computational Complexity, Germany\protect\\
  \texttt{\{heeger,rolf.niedermeier\}@tu-berlin.de}
  }
\affil[3]{
  Czech Technical University in Prague, Prague, Czech Republic\protect\\
  \texttt{dusan.knop@fit.cvut.cz}}
\date{}

\maketitle

\begin{abstract}
Since the early days of research in algorithms and complexity, the
computation of stable matchings is a core topic. While in the classic
setting the goal is to match up two agents (either from different
``gender'' (this is \textsc{Stable Marriage}) or ``unrestricted'' (this
is \textsc{Stable Roommates})), Knuth~[1976] triggered the study
of three- or multidimensional cases. Here, we focus on the study
of \textsc{Multidimensional Stable Roommates}, known to be
\NP-complete since the early~1990's. Many \NP-completeness results, however,
rely on general input instances that do not occur in at least
some of the specific application scenarios. With the quest for identifying
islands of tractability for \textsc{Multidimensional Stable Roommates}, 
we study the case of master
lists. Here, as natural in applications where agents express their
preferences based on ``objective'' scores, one roughly speaking assumes
that all agent preferences are ``derived from'' a central master list,
implying that the individual agent preferences shall be similar.
Master lists have been frequently studied in the two-dimensional (classic)
stable matching case, but seemingly almost never for the multidimensional
case. This work, also relying on methods from parameterized algorithm
design and complexity analysis, performs a first systematic study
of \textsc{Multidimensional Stable Roommates} under the assumption of
master lists.

\medskip

\end{abstract}

\section{Introduction}
Computing stable matchings is a core topic in the intersection of algorithm design, algorithmic game theory, and computational social choice.
It has numerous applications such as higher education admission in several countries~\cite{BALINSKI99,BraunDK10},
kidney exchange~\cite{RothSU04}, assignment of dormitories~\cite{PerachPR08}, P2P-networks~\cite{GaiLMMRV07}, wireless three-sided networks~\cite{CuiJ13}, and spatial crowdsourcing~\cite{LiCYWC19}.
The research started in the 1960's with the seminal work of Gale and Shapley~\cite{GaleShapley1962}, introducing the \textsc{Stable Marriage} problem:
given two different types of agents, called ``men'' and ``women'',
each agent of one gender has preferences (i.e., strict orders aka rankings) over the agents of the opposite gender.
Then, the task is to find a matching which is stable.
Informally, a matching is \emph{stable} if no pair of agents can improve by
breaking up with their currently assigned partners and
instead matching to each other.

Many variations of this problem have been studied; \textsc{Stable Roommates}, with only one type of agents, is among the most
prominent ones.
Knuth~\cite{Knuth76} asked for generalizing \textsc{Stable Marriage} to dimension three, i.e., having three types of agents and having to match the agents to groups of size three, where any such group contains exactly one agent of each type.
Here, a matching is called \emph{stable} if there is no group of three agents which would improve by being matched together.
We focus on the \textsc{Multidimensional Stable Roommates} problem.
Here, there is only one type of agents, now having preferences over 
\emph{$(d-1)$-sets} (that is, sets of size $d-1$) of (the other) agents.

\begin{example}
 \label{example:intro}
 Consider the following instance of \textsc{3-Dimensional Stable Roommates} with six agents $a$, $b$, $c$, $d$, $e$, and~$f$.
 \begin{align*}
  a &: \{b, d\} \succ \{b, c\} \succ \{b, e\} \succ \{b, f\} \succ \{c, d\} \succ \{c, e\} \succ \{c, f\} \succ \{d, e\} \succ \{d, f\} \succ \{e, f\}\\
  b &: \{a, d\} \succ \{a, c\} \succ \{a, e\} \succ \{a, f\} \succ \{c, d\} \succ \{c, e\} \succ \{c, f\} \succ \{d, e\} \succ \{d, f\} \succ \{e, f\}\\
  c &: \{a, b\} \succ \{a, d\} \succ \{a, e\} \succ \{b, d\} \succ \{a, f\} \succ \{b, e\} \succ \{b, f\} \succ \{d, e\} \succ \{d, f\} \succ \{e, f\}\\
  d &: \{a, b\} \succ \{a, c\} \succ \{a, e\} \succ \{a, f\} \succ \{b, c\} \succ \{b, e\} \succ \{b, f\} \succ \{c, e\} \succ \{c, f\} \succ \{e, f\}\\
  e &: \{a, b\} \succ \{a, c\} \succ \{a, d\} \succ \{a, f\} \succ \{b, c\} \succ \{b, d\} \succ \{b, f\} \succ \{c, d\} \succ \{c, f\} \succ \{d, f\}\\
  f &: \{a, b\} \succ \{a, c\} \succ \{a, d\} \succ \{a, e\} \succ \{b, c\} \succ \{b, d\} \succ \{b, e\} \succ \{c, d\} \succ \{c, e\} \succ \{d, e\}
 \end{align*}
 Matching~$M_1: = \{\{a, b, c\}, \{d, e, f\}\}$ is not stable, as $\{a, b, d\}$ is blocking because $a $ prefers~$\{b, d\}$ to $\{b, c\}$, agent~$b$ prefers $\{a, d\}$ to $\{a, c\}$, and $d$ prefers $\{a, b\}$ to $\{e, f\}$.
 However, matching~$M_2 := \{ \{a, b, d\}, \{c, e, f\}\}$ is stable.
\end{example}

As this problem is \NP-complete in general~\cite{NH91}, we focus on
the case where the preferences of all agents are derived from an ordered master list.
For instance, master lists naturally arise when the agent preferences are based on scores, e.g., when assigning junior doctors to medical posts in the UK~\cite{IMS08} or when allocating students to dormitories~\cite{PerachPR08}. Master lists have been frequently used in the context
of (two-dimensional) stable matchings~\cite{BiroIS11,IMS08,OMalley07,PerachPR08} or the related \textsc{Popular Matching} problem~\cite{KavithaNN14}.
We generalize master lists to the multidimensional setting in two natural ways.
First, following the above spirit of preference orders, we assume that
the master list consists of sets of size~$d-1$.
Each agent then derives its preferences from the master list by just deleting all $(d-1)$-sets containing the agent itself.
For example, in \Cref{example:intro} agents~$d$, $e$, and $f$ derive their preferences from the list $ \{a, b\} \succ \{a, c\} \succ \{a, d\} \succ \{a, e\} \succ \{a, f\} \succ \{b, c\} \succ \{b, d\} \succ \{b, e\} \succ \{b, f\} \succ \{c, d\} \succ \{c, e\} \succ \{c, f\} \succ \{d, e\} \succ \{d, f\} \succ \{e, f\}$, while $a$, $b$, and $c$ do not.
In the second way we study, the master list is a poset over the set of agents.
In this case, any agent~$a$ shall prefer a $(d-1)$-set~$t$ to a $(d-1)$-set~$t'$ if $t$ is ``better'' than $t'$ according to the master list, where ``better'' means that $a$ does not prefer the $k$-th best agent of $t'$ to the $k$-th best agent from $t$ (according to the master list).
For any tuples~$t$,~$t'$ for which neither $t$ is ``better'' than $t'$ nor $t'$ is ``better'' than $t$, an agent may prefer~$t$ to $t'$ or $t'$ to $t$ independently of the other agents.
More formally, we require that any agent prefers a set of $d-1$ agents $t$ to any set of $d-1$ agents $t'$ dominated by $t$, where we say that $t= \{a_1, \dots, a_{d-1}\}$ dominates $t' = \{b_1, \dots, b_{d-1}\}$ if $a_i = b_i$ or the master list prefers~$a_i$ to $b_i$ for all $i\in [d-1]$ and $a_i \neq b_i$ for some $i\in [d-1]$.
Then for any two sets~$\{a_1, \dots, a_{d-1}\}$ and $\{b_1, \dots, b_{d-1}\}$ of $d-1$ agents with~$\{a_1, \dots, a_{d-1}\}$ dominating~$\{b_1, \dots, b_{d -1}\}$, the preferences of any agent must fulfill that the set $\{a_1, \dots, a_{d-1}\}$ is before $\{b_1, \dots, b_{d-1}\}$.
In \Cref{example:intro}, the preferences of $c$, $d$, $e$, and $f$ are derived from the list of agents $a \succ b \succ c \succ d \succ e \succ f$.
For such master lists, we also relax the condition that the master list is a strict order (that is, for every two different agents $a$ and $b$, either $a$ is better than $b$ or $b$ is better than $a$ in the master lits) by the condition that the master list is a partially ordered set (poset), and consider the parameterized complexity with respect to parameters measuring the similarity to a strict order.
Preferences where such a parameter is small might arise if there are few similar rankings, and each agent derives its ranking from these orders, or if the objective score consists of several attributes and each agent weights these attributes slightly differently.
From these rankings of each agent, a master poset arises by saying that agent~$a$ is better than agent~$b$ if and only if all agents (except for~$a$ and $b$) agree on this.

\subsection{Related work}
\textsc{Stable Roommates} can be solved in linear time~\cite{Irving85}.
If the preferences are incomplete (that is, two agents may prefer being unmatched to being matched together) and derived from a strict master list, then both \textsc{Stable Marriage} and \textsc{Stable Roommates} admit a unique stable 
matching~\cite{IMS08}.\footnote{Actually, Irving et al.~\cite{IMS08} only state this for \textsc{Stable Marriage}, but the generalization to \textsc{Stable Roommates} is trivial (actually, the whole statement is straightforward).}

If the preferences are complete but contain ties, then there are three different generalizations of stability studied in the literature.
\emph{Weak stability} considers a pair to be blocking if both agents in this pair prefer each other to their assigned partner in the matching.
\emph{Strong stability} considers a pair to be blocking if one agent prefers the pair to the agent assigned to it, and the other agent does not prefer the pair to the agent assigned to it.
\emph{Super-stability} considers a pair to be blocking if both agents in this pair do not prefer their assigned partner to each other.
Finding a weakly stable matching in a \textsc{Stable Roommates} instance is \NP-complete~\cite{Ronn90}.
However, if the preferences are complete, derived from a master list, and contain ties, then one can decide whether a given pair of agents in a \textsc{Stable Marriage} instance is matched together in some weakly stable matching in linear time~\cite{IMS08}
(which is \NP-complete for general complete preferences~\cite{ManloveIIMM02}), and a weakly stable matching in a \textsc{Stable Roommates} instance always exists and can be found in linear time.
For incomplete preferences derived from a master list with ties, an $O(\sqrt{n}m)$-time algorithm for finding a strongly stable matching is known~\cite{OMalley07} (where $n$ is the number of agents and $m$ is the number of acceptable pairs), while for general preferences, only an $O(mn)$-time algorithm is known~\cite{Kunysz16}.
Finding a weakly stable matching in a \textsc{Stable Roommates} instance, however, is \NP-complete if the preferences contain ties, are incomplete, and are derived from a master list~\cite{IMS08}.
Further examples of \textsc{Stable Marriage} problems becoming easier for complete preferences derived from a master list are given by Scott~\cite[Chapter 8]{Scott05}.
There is quite some work for
\textsc{3-Dimensional Stable Marriage}~\cite{Danilov03,OstrovskyR14,Wu16,ZhongB19}, but less so for
\textsc{3-Dimensional Stable Roommates}.

While master lists are a standard setting for finding 2-dimensional stable matchings~\cite{BiroIS11,IMS08,Kamiyama19,OMalley07,PerachPR08}, we are only aware of few works combining multidimensional stable matchings with master lists.
	Escamocher and O'Sullivan~\cite{EO18} gave a recursive formula for the number of 3-dimensional stable matchings for cyclic preferences (i.e., the agents are partitioned into three sets $A_0$, $A_1$, and~$A_2$, and each agent from $A_i$ only cares about the agent from $A_{i+1}$ (modulo~3) it is matched to) derived from master lists.
Cui and Jia~\cite{CuiJ13} showed that if the preferences are cyclic and the preferences of the agents from $A_1$ are derived from a master list, while each agent from~$A_3$ is indifferent between all agents from $A_1$, then a stable matching always exists and can be found in polynomial time, but it is \NP-complete to find a maximum-cardinality stable matching.
There is some work on $d$-dimensional stable matchings and cyclic
preferences (without master lists)~\cite{Hofbauer16,LamP19}.

Deineko and Woeginger~\cite{DeinekoW13} showed that \textsc{3-Dimensional Stable Roommates} is \NP-complete for preferences derived from a metric space.
For the special case of the Euclidean plane, Arkin et al.~\cite{ArkinBEOMP09} showed that a stable matching does not always exist, but
left the complexity of deciding existence open.

Huang~\cite{DBLP:conf/esa/Huang07} showed that \textsc{3-Dimensional Stable Roommates} is \NP-complete even if the preferences are \emph{consistent}, i.e., for each agent~$a$, there exists a strictly ordered preference list~$\succ_a$ over all other agents such that for any two pairs $\{b, c\}$ and $\{d,e\}$ of agents with $b \succ_a d$ and $c \succ_a e$ it holds 
that $a$ prefers $\{b, c\}$ to $\{d, e\}$.
Note that in a \textsc{3-dimensional Stable Roommates} instance, the preferences of all agents are derived from a strict order $\succ$ as master poset if and only if the preferences of every agent are consistent, and this can be witnessed by the strict order $\succ$ for every agent. 

Iwama et al.~\cite{IwamaMO07}
introduced the \NP-complete
\textsc{Stable Roommates with Triple Rooms}, where each agent has preferences over all other agents, and prefers a 2-set $p$ of agents to a 2-set~$p'$ if it prefers the best-ranked agent of $p$ to the best-ranked agent of $p'$, and the second-best agent of $p$ to the second-best agent of $p'$.
They showed that this problem is \NP-complete.

Our scenario of
\textsc{Multidimensional Stable Roommates} can be seen as a special case of finding core-stable outcomes for hedonic games where each agent
prefers size-$d$ coalitions over singleton-coalitions
which are then preferred over all other coalitions~\cite{Rothe15,Woe13}.
Notably, there are fixed-parameter tractability results
for hedonic games (without
fixed ``coalition'' size as we request) with respect to
treewidth (MSO-based)~\cite{HKMO19,Pet16}.
Other research considers hedonic games with fixed coalition size~\cite{CFH19},
but aims for Pareto optimal outcomes instead of core stability which we consider.

To the best of our knowledge, the parameterized complexity of multidimensional stable matching problems has not yet been investigated.

\subsection{Our contributions}
Our results are surveyed in~\Cref{tab:results}.
To our surprise, even if the preferences are derived from a master list of 2-sets of agents (this is the special case of dimension $d=3$), a stable matching is not guaranteed to exist (\Cref{sec:no-stable-matching}).
We
use such an instance not admitting a stable matching to show that \textsc{Three-Dimensional Stable Roommates} is \NP-complete also when restricted to preferences derived from a master list of 2-sets (\Cref{tNPc}).

If the preferences are derived from a strict master list of agents, then a unique stable matching always exists and can be found by a straightforward algorithm (\Cref{tconsistent}).
When relaxing the condition that the master list is strictly ordered to being a poset (i.e., the master list may also declare two agents ``incomparable'' instead of stating that one is better than the other, but if agent~$a$ is better than $b$ and $b$ is better than $c$, then also $a$ is better than $c$), then the problem clearly is \NP-complete, as a master list which ties all agents does not impose any condition on the preferences of the agents, and \textsc{Three-Dimensional Stable Roommates} is \NP-complete.
Consequently, in the spirit of 
``distance from triviality''-parameterization~\cite{GHN04,Nie06}, we investigate the parameterized complexity with respect to several parameters
measuring the distance of the poset to a strict order.
Note that our algorithm can also solve the corresponding search problem.
For the parameter maximum number of agents incomparable to a single agent, we show that \textsc{Multidimensional Stable Roommates} is fixed-parameter tractable~(\FPT) (even
when $d$ is part of the input) (\Cref{thm:FPT-kappa}).
If this parameter is bounded, then this results in a
special case of 3-dimensional stable matching problems which can be solved by an ``efficient'' nontrivial algorithm.
Considering the stronger parameter width of the master poset, we show \textsc{Three-dimensional Stable Roommates} to be \Wone-hard, and this is true also for the orthogonal parameter deletion (of agents) distance to a strictly ordered master poset (\Cref{tWh}).
We also show that \textsc{Three-Dimensional Stable Roommates}
is \NP-complete even with a strict order of the agents as a master poset if each agent is allowed to declare an arbitrary set of $2$-sets unacceptable (\Cref{thm:incomplete-w-h}), contrasting the polynomial-time solvability when every agent accepts every $(d-1)$-set of other agents (\Cref{tconsistent}).
  \begin{table}[t]
\caption{Results overview: six variations of \textsc{Multidimensional Stable Roommates}.
All three studied parameters measure the similarity of the master poset to a strict order.
Note that the parameter ``max.\ number~$\kappa$ of incomparable agents'' is weaker than the parameter ``Width of master poset'', and both parameters are incomparable to the parameter ``Deletion distance to strictly ordered master poset''.}
    \begin{center}
      \begin{tabular}{c | c }
        Setting/Parameter & Complexity\\
        \hline
        \multicolumn{1}{l|}{Master list of 2-sets} & \NP-complete for $d=3$ (\Cref{tNPc}) \\
        \hline\hline
        \multicolumn{1}{l|}{Master poset of agents:} & \\
        \hline
        Linear master poset of agents & linear time (\Cref{tconsistent}) \\
        $max.~number~\kappa$ of incomparable agents & $O(n^2) + (\kappa^2 2^{12\kappa})^{O(\kappa^2 2^{12\kappa})} n$ (\Cref{thm:FPT-kappa}) \\
        Width of master poset & \Wone-hard for $d=3$ (\Cref{thm:width_w-h}) \\
        Incomplete preferences, strictly ordered master poset & \NP-complete for $d \ge 3$ (\Cref{thm:incomplete-w-h})\\
        Deletion distance to strictly ordered master poset & \Wone-hard for $d=3$ (\Cref{tWh})
      \end{tabular}
    \end{center}
\label{tab:results}
  \end{table}

  \subsection{Structure of the paper}
  After introducing basic notation in \Cref{sec:preliminaries}, we consider
  \textsc{3-Dimensional Stable Roommates} with master list of $(d-1)$-sets of agents
  in \Cref{sec:ml-tuples} and show its \NP-completeness.
  Then, we turn to master posets of agents.
  We show in \Cref{sec:strict-order} that \textsc{3-Dimensional Stable Roommates} is easy if the master list is strictly ordered.
  Moreover, we consider the case that preferences are incomplete, or that the master list is a poset and investigate parameters measuring the similarity to a strict order in \Cref{sec:posets,sec:del-dist}.
  In \Cref{sec:incomplete}, we show that \textsc{3-Dimensional Stable Roommates} is \NP-complete if the master list is strict, but every agent may declare an arbitrary subset of $(d-1)$-sets to be not acceptable.
  Finally, we conclude in \Cref{sec:conclusion}.

\section{Preliminaries}
\label{sec:preliminaries}

Let $[n]\coloneqq \{1, 2, 3, \dots, n\}$ and $[n, m]\coloneqq \{n, n+ 1, \dots, m\}$.
A set of cardinality $d$ will also be called \emph{$d$-set}.
For a set $X$ and an integer~$d$, we denote by $\binom{X}{d}$ the set of size-$d$ subsets of $X$.
A \emph{preference list} $\succ$ over a set $X$ is a strict order of $X$.
We call a set of pairwise disjoint $d$-subsets of a set~$A$ of agents a \emph{$d$-dimensional matching}.
Usually, $d$~is clear from the context; if so, then we may 
only write ``matching''.
Given a $d$-dimensional matching $M$ and an agent~$a$, we denote by~$M(a)$ the $(d-1)$ set~$t$ such that $t\cup \{a\} \in M$; if for~$a$ no such $(d-1)$-set exists, then $M(a)\coloneqq \emptyset$.
We say that an agent~$a$ \emph{prefers} a $(d-1)$-set $t$ to a $(d-1)$-set $t'$ if $t \succ_a t$ where $\succ_a$ is the preference list of~$a$. 
Any agent prefers any $(d-1)$-set not containing itself to being unmatched.
A \emph{blocking $d$-set} for a $d$-dimensional matching~$M$ is a set of $d$ agents~$\{a_1, a_2,\dots,  a_d\}$ such that, for all~$i\in [d]$, either $a_i$ is unmatched in $M$ or $\{a_1, a_2, \dots, a_d\} \setminus \{a_i\} \succ_{a_i} \{b_1^i, b_2^i, \dots, b_{d-1}^i\}$, where~$\{b_j^i : j\in [d-1]\}\cup \{a_i\}\in M$.
A matching is called \emph{stable} if it does not admit a blocking $d$-set.
Now, we are ready to define our central problem.

\defProblemTask{\textsc{Multidimensional Stable Roommates} (\textsc{MDSR})}
{An integer $d$, a set $A$ of agents together with a preference list $\succ_a$ over $\binom{A\setminus \{a\}}{d-1}$ for each agent~$a\in A$.}
{Decide whether a stable matching exists.}
Note that we require each agent to list each size-$(d-1)$ set of other agents.
We denote by \textsc{$\ell$-DSR} the restriction of \textsc{MDSR} to instances with $d = \ell$.
It is known that a 3-dimensional stable matching does not always exist, and \textsc{3-DSR} is \NP-complete~\cite{NH91}.

A \emph{master list} $\ML$ is a preference list over $\binom{V}{d-1}$.
A preference list $\succ_v$ for an agent~$a$ is \emph{derived from a master list} $\ML$
by deleting all $(d-1)$-sets containing~$a$.

\begin{example}
 Let $A = \{a_1, a_2, a_3, a_4\}$ be a set of agents, $d=3$, and let $\{a_1, a_2\} \succ \{a_2, a_4\}\succ \{a_1, a_3\} \succ \{a_3, a_4\} \succ \{a_2, a_3\} \succ \{a_1, a_4\}$ be the master list.

 Then the preferences of $a_1$ are $\{a_2, a_4\}\succ_{a_1} \{a_3, a_4\} \succ_{a_1} \{a_2, a_3\} $, the preferences of $a_2$ are $\{a_1, a_3\} \succ_{a_2} \{a_3, a_4\} \succ_{a_2} \{a_1, a_4\}$, the preferences of $a_3$ are $\{a_1, a_2\} \succ_{a_3} \{a_2, a_4\} \succ_{a_3} \{a_1, a_4\}$, and the preferences of $a_4$ are $\{a_1, a_2\} \succ_{a_4} \{a_1, a_3\} \succ_{a_4} \{a_2, a_3\} $.
\end{example}

Next, we define the \textsc{Multidimensional Stable Roommates with Master List of ${(d-1)}$-Sets} problem (\mdsmtuples).
\defProblemTask{\mdsmtuples}
{An integer $d$, a set $A$ of agents, and a master list $\succ_{\ML}$ over $\binom{A}{d-1}$, from which the preference list of each agent is derived.}
{Decide whether a stable matching exists.}
Again, we denote by \textsc{$\ell$-DSR-ML-Sets} the problem \mdsmtuples\xspace restricted to instances with $d = \ell$.

We now turn to the case that the master list orders single agents instead of $(d-1)$-sets of agents.
We first need the definition of a partially ordered set.

  A \emph{partially ordered set (poset)} is a  pair $(V, \succeq)$, where $\succeq$ is a binary relation over the set $V$ such that      (i) $v\succeq v$ for all $v\in V$,
      (ii) $v\succeq w$ and $w\succeq v$ if and only if $ v = w$, and
      (iii) if $u\succeq v$ and $v\succeq w$, then $u\succeq w$.

  If $v\succeq w$ and $v\neq w$, then we write $v \succ w$.
  If neither $v\succeq w$ nor $w\succeq v$, then we say that~$v$ and $w$ are \emph{incomparable}, and write $v \perp w$.
  Instead of $v\succeq w$ or $v\succ w$, we may also write $w\preceq v$ 
  or $w \prec v$.
  A \emph{weak order} is a poset such that for every $a, b$, and $c$ with $a \perp b$ and $b \perp c$ also $a\perp c$ hold.

  A \emph{chain} of~$(V, \succeq)$ is a subset $X=\{x_1, x_2, \dots, x_k\}\subseteq V$ such that $x_i \succ x_{i+1}$ for all~$i\in [k-1]$.
  An \emph{antichain} is a subset $X \subseteq V$ such that for all $v, w\in X$ with $v\neq w$, we have $v\perp w$.
  The \emph{width} of a poset is the size of a maximum antichain.

  For a poset $\succ$ over a set $V$, $\kappa_{\succ} (v) \coloneqq |\{w\in V : v\perp w\}|$ be the number of elements incomparable with~$v$.
  We define $\kappa (\succ)\coloneqq \max_{v\in V} \kappa_{\succ} (v)$.
  As an example, consider the poset $(\{v_1, v_2, v_3, v_4\}, \succ) $ with $v_1 \succ v_2$, $v_2 \succ v_3$, and $v_1 \succ v_4$.
  Here, $v_1$ is comparable to all other agents, $v_2$ and $v_3$ are only incomparable to $v_4$, and $v_4$ is incomparable to $v_2$ and $v_3$, and we have $\kappa (\succ ) = 2$.
Note that if $\bar{G}_\succ$ is the incomparability graph of the poset $(V, \succ)$ (i.e., the graph whose vertex set is~$V$ and there is an edge between $v,w \in V$ if and only if $v\perp w$), then $\Delta (\bar{G}_\succ ) = \kappa (\succ)$, where $\Delta (\bar{G}_\succ)$ is the maximum vertex degree in~$\bar{G}_\succ$.
If $\succ$ is a weak order, then the parameter~$\kappa (\succ)$ equals the maximum size of a tie.

Dilworth's Theorem \cite{Dilworth50} states that the width of a poset is the minimum number of chains such that each element of the poset is contained in one of these chains.

Having defined posets, we now show the connection to \textsc{Multidimensional Stable Roommates} by using preferences derived from a poset of agents.

\begin{definition}
 Given a set~$A$ of agents, a poset $(A, \succ_{\ML})$ (called the \emph{master poset}), and an integer~$d$, a preference list $\succ_a$ on $\binom{A\setminus \{a\}}{d-1}$ \emph{is derived from} $\succ_{\ML}$ if whenever $a_1, \dots, a_{d -1}$ and $b_1,\dots, b_{d -1}$ with $a_i\succeq_{\ML} b_i$ for all $i\in [d-1]$, then $\{a_1, \dots, a_{d -1}\} \succeq_{v} \{b_1, \dots, b_{d -1}\}$.
\end{definition}

\begin{example}
 Let $a_1 \succ a_2 \succ a_3 \succ a_4 \succ a_5$ be a master poset.
 Then $a_1$ has one of the following
 two preferences:
  $\{a_2, a_3\} \succ_{a_1} \{a_2, a_4\} \succ_{a_1}\{a_2, a_5\} \succ_{a_1} \{a_3, a_4\} \succ_{a_1} \{a_3, a_5\} \succ_{a_1} \{a_4, a_5\}$ or
  $\{a_2, a_3\} \succ_{a_1} \{a_2, a_4\} \succ_{a_1} \{a_3, a_4\}\succ_{a_1}\{a_2, a_5\}  \succ_{a_1} \{a_3, a_5\} \succ_{a_1} \{a_4, a_5\}$.
  
 For the master poset~$\succ$ with $a_2 \succ a_3$, $a_2 \succ a_4$, $a_3 \succ a_1$, and $a_4 \succ a_1$ (and $a_3 \perp a_4$), agent $a_1$ has one of the following preferences:
 $\{a_2, a_3\} \succ_{a_1} \{a_2, a_4\} \succ_{a_1} \{a_3, a_4\}$ or
  $\{a_2, a_4\} \succ_{a_1} \{a_2, a_3\} \succ_{a_1} \{a_3, a_4\}$.

\end{example}

The formal definition of \dsmpo reads as follows.
\defProblemTask{\dsmpo}
{An \textsc{MDSR} instance $\mathcal{I} = (A, (\succ_a)_{a\in A}, d)$ and a master poset $\succeq_{\ML}$ such that the preferences $\succ_a$ of each agent $a$ are derived from $\succeq_{\ML}$.}
{Decide whether there exists a stable matching in $\mathcal{I}$.}
Again, we denote by \textsc{$\ell$-DSR-Poset} the problem \dsmpo\ restricted to instances with $d = \ell$.

\paragraph*{Parameterized Complexity.}

A \emph{parameterized language~$L$} over a finite alphabet~$\Sigma$ is a subset~$L\subseteq \Sigma^* \times \mathbb{N}$.
A parameterized language~$L$ is \emph{fixed-parameter tractabile} if there exists an algorithm which correctly decides for every instance $(I, k) \in \mathbb{N}$ whether $(I, k) \in L$ in FPT-time, that is, $f(k) \cdot |(I, k)|^{O(1)}$, where $|(I, k)| := |I|+ k$ with $|I|$ being the size of $I$ and $f$ being a computable function.
The class of all fixed-parameter tractable parameterized languages is denoted by \FPT.

To indicate that a parameterized language is not fixed-parameter tractable, one usually uses the notions of parameterized reductions and \Wone-hardness.
A \emph{parameterized reduction} from a parameterized language $L$ to a parameterized language $L'$ is an algorithm $\mathcal{A}$ which computes, given an instance $(I, k)$ of $L$, an instance $(I', k' ) \in L'$ with $k' \le g(k)$ for some computable function $g$ such that $(I, k) \in L$ if and only if $(I', k') \in L'$ and runs in FPT-time.
\Wone is the class of parameterized problems which are equivalent to \textsc{Clique} parameterized by solution size under parameterized reductions.
To show \Wone-hardness of a parameterized problem~$L$, one usually gives a parameterized reduction to \textsc{Clique} or another \Wone-hard problem.
Under the complexity-theoretic assumption $\FPT \neq \Wone$, \Wone-hardness of a parameterized problem implies that this problem is not in \FPT.
For more details on parameterized complexity, we refer to standard textbooks~\cite{CyganFKLMPPS15,DBLP:series/txcs/DowneyF13,FG06,Nie06}.

\section{Three-dimensional stable roommates with master list of $2$-sets}
\label{sec:ml-tuples}

In this section, we consider the case that the preferences are complete and derived from a master list of $(d-1)$-sets.
In \Cref{sec:no-stable-matching}, we give a small instance with six agents not admitting a stable matching. We use this to show in \Cref{sNPh} that already for dimension $d=3$ and preferences derived from a master list of 2-sets, deciding whether an instance admits a stable matching is \NP-complete.

\subsection{A 3-DSR-ML instance not admitting a stable matching}
\label{sec:no-stable-matching}

Subsequently, we present a \tdsmtuples\xspace instance $\mathcal{I}_{\operatorname{instable}} $ with six agents not admitting a stable matching, showing that stable matchings do not have to exist even in the presence of master lists.
This somewhat surprising observation will be crucial for our 
NP-hardness proof in \Cref{sNPh}.

The instance $\mathcal{I}_{\operatorname{instable}} $ has six agents $a$, $b$, $c$, $d$, $e$, and $f$.
The master list is:
    $\{a, b\}
    \succ \{a, c\} \succ \{a, d\} \succ \{a, f\}  \succ \{b, e\}\succ \{c, d\} \succ \{a, e\} \succ \{b, f\}
    \succ \{c, e\}\succ \{b, d\}\succ \{d, e\} \succ \{b, c\} \succ \{c, f\} \succ \{d, f\} \succ \{e, f\}.$

  \begin{observation}\label{oinstability}
    Instance~$\mathcal{I}_{\operatorname{instable}}$
    does not admit a stable matching.
  \end{observation}

  \begin{proof}
    \Cref{tbt} presents for each of the $\frac{\binom{6}{3}}{2} = 10$ matchings a blocking 3-set.
    \begin{table}[t]

\caption{A blocking 3-set for each matching in instance $\mathcal{I}_{\operatorname{instable}}$ from \Cref{oinstability}.}    \begin{center}
      \begin{tabular}{c | c}
        Matching & Blocking 3-set \\
        \hline
        $\{a, b, c\}, \{d, e, f\}$ & $\{a, d, e\}$ \\
        $\{a, b, d\}$, $\{c, e, f\}$ & $ \{a, c, e\}$\\
        $\{a, b, e\}$, $\{c, d, f\}$ & $\{b, c, d\}$ \\
        $\{a, b, f\}$, $\{c, d, e\}$ & $\{a, c, d\}$\\
        $\{a, c, d\}$, $\{b, e, f\}$ & $\{a, b, e\}$\\
      \end{tabular}~~~~~~~~\begin{tabular}{c | c}
        Matching & Blocking 3-set \\
        \hline
        $\{a, c, e\}$, $\{b, d, f\}$ & $\{a, b, e\}$\\
        $\{a, c, f\}$, $\{b, d, e\}$ & $\{a, b, e\}$\\
        $\{a, d, e\}$, $\{b, c, f\}$ & $\{a, b, e\}$\\
        $\{a, d, f\}$, $\{b, c, e\}$ & $\{a, b, e\}$\\
        $\{a, e, f\}$, $\{b, c, d\}$ & $\{a, c, d\}$
      \end{tabular}

\label{tbt}
    \end{center}
    \end{table}
  \end{proof}

\subsection{NP-completeness of 3-DSR-ML}
\label{sNPh}

Using the instance $\mathcal{I}_{\text{instable}}$ from~\Cref{sec:no-stable-matching}, we show $\NP$-completeness of \tdsmtuples,
reducing from the $\NP$-complete problem
\textsc{1-in-3 Positive 3-Occurrence-SAT}~\cite{Gonzalez85}.
\defProblemTask{\textsc{1-in-3 Positive 3-Occurrence-SAT}}
{A boolean formula in conjunctive normal form, where each clause contains {\emph{ exactly}} three {\emph{pairwise different}} variables, and each variable appears exactly three times and only non-negatedly in the formula.}
{Decide whether there exists a truth assignment satisfying exactly one literal from each clause.}

  The basic idea of the reduction is the following.
  For each clause $C_j$, we have two agents~$c_j$ and~$d_j$. For each variable~$x_i$, we have three agents $x_i^1$, $x_i^2$, and~$x_i^3$, one for each occurrence of the variable.
  Additionally, there are six agents~$z_i^{k, \ell}$, $\ell\in [6]$, for each literal of a clause.
  In any stable matching the agents $c_j$ and $d_j$ are matched to an agent~$x_i^\ell$ corresponding to a variable occurring in clause~$C_j$.
  Now consider a variable~$x_i$.
  If agent~$x_i^k$ (corresponding to the $k$-th occurrence of~$x_i$) is matched to a 2-set $\{c_j, d_j\}$ (corresponding to clause~$C_j$) for all~$k\in [3]$, then a stable matching can match $\{z_i^{k,1}, z_i^{k,2}, z_i^{k,3}\}$ and $\{z_i^{k,4}, z_i^{k,5}, z_i^{k,6}\}$.
  If agent~$x_i^k$ (corresponding to occurrences of variable~$x_i$) is not matched to a 2-set~$\{c_j, d_j\}$ (corresponding to clause~$C_j$) for all~$k\in [3]$, then we can match $\{x_1^i, x_2^i, x_3^i\}$ and then again match $\{z_i^{k,1}, z_i^{k,2}, z_i^{k,3}\}$ and $\{z_i^{k,4}, z_i^{k,5}, z_i^{k,6}\}$.
  If, however, agent~$x_i^k$ is matched to a 2-set of the form~$\{c_j, d_j\}$ for one or two values of~$k$, then an agent $x_i^j$ which is not matched to a 2-set~$\{c_j, d_j\}$ will together with~$z_i^{k,1}$,~$z_i^{k,2}$,~$z_i^{k,3}$,~$z_i^{k,4}$, and~$z_i^{k,5}$ form a subinstance of six agents which does not admit a stable matching, and, thus, the resulting matching will not be stable.

  Summarizing, any stable matching matches $c_j$ and $d_j$ to exactly one variable occurring in clause~$C_j$, and for each variable $x_i$, either all or none of the agents $x_i^k$ are matched to such 2-sets~$\{c_j, d_j\}$.
  In other words, setting those variables~$x_i$ such that $x_i^k$ is matched to a 2-set~$\{c_j, d_j\}$ to $\LogicTRUE$ and all other variables to $\LogicFALSE$, we get a solution of the \textsc{1-in-3 Positive 3-Occurrence-SAT} instance from each stable matching.

  ``Inverting'' this process (i.e., matching each clause to its true variable, and then matching the variable gadgets as described above), allows to construct a stable matching from a solution to the \textsc{1-in-3 Positive 3-Occurrence-SAT} instance.

  Note that \textsc{MDSR} is in $\NP$ as the size of the input is $\Omega(\binom{n}{d-1})$, where $n$ is the number of agents, as the master preference list contains $\binom{n}{d-1}$ sets of size~$d-1$, and, thus, stability can be verified in time polynomial in the input size by just checking for each $d$-set whether it is blocking.
  
  We now formally describe the reduction and prove its correctness.

  \subsubsection{The reduction}
\label{sec:NPh-reduction}

  Let $x_1, \dots, x_n$ be the variables and let $C_1, \dots, C_m$ be the clauses of a \textsc{1-in-3 Positive 3-Occurrence-SAT} instance $\mathcal{I}$.
  We construct a \textsc{3-DSR-ML} instance $\mathcal{I}' = (A, (\succ_a)_{a\in A})$ as follows.

  For each clause $C_j$, we add two agents $c_j$ and $d_j$ to $A$.
  For the $k$-th occurrence ($k \in [3]$) of a variable $x_i$ in a clause, we add an agent $x_i^k$.
  We refer to the agent corresponding to the~$\ell$-th literal of clause $C_j$ as $y_j^\ell$.
  The $k$-th occurrence of $x_i$ is also the $\ell$-th literal of some clause~$C_j$, and we will denote the agent $x_i^k$ also by $y^\ell_j$, i.e., $x^k_i = y^\ell_i$.
  For each agent $x_i^k$, we add six agents~$z_i^{k,1}, \dots, z_i^{k, 6}$ to $A$.

  For each $j\in [m]$, we define $\mathcal{A}_j$ to be the following part of the master list:
  \[
    \{c_j, d_j\} \succ \{y_j^1, d_j\} \succ \{y_j^3, c_j\} \succ \{y_j^2, d_j\} \succ \{y_j^2, c_j\} \succ \{y_j^3, d_j\} \succ \{y_j^1, c_j\}.
  \]

  For each agent $x^k_i$, we define $\mathcal{B}^k_i$ to be the following part of the master list (note that (by renaming $x_i^k$ to $a$, agent $z_i^{k, 1}$ to $b$, agent $z_i^{k, 2}$ to $c$, \dots, and $z_i^{k, 5}$ to $f$) this contains the instance~$\mathcal{I}_{\operatorname{instable}}$ from \Cref{oinstability};
  this ensures that $x_k^i$ has to be matched to a 2-set which is before $\mathcal{B}^k_i$ in the master list):
  \begin{align*}
    \{x^k_i, z_i^{k, 1}\} & \succ \{x^k_i, z_i^{k,2}\} \succ \{x^k_i, z_i^{k,3}\} \succ \{x^k_i, z_i^{k,5}\}  \succ \{z_i^{k,1}, z_i^{k,4}\}\succ \{z_i^{k,2}, z_i^{k,3}\} \\
    &\succ \{x^k_i, z_i^{k,4}\} \succ \{z_i^{k,1}, z_i^{k,5}\} \succ \{z_i^{k,2}, z_i^{k,4}\}\succ \{z_i^{k,1}, z_i^{k,3}\}\succ \{z_i^{k,3}, z_i^{k,4}\}\\
    & \succ \{z_i^{k,1}, z_i^{k,2}\} \succ \{z_i^{k,2}, z_i^{k,5}\} \succ \{z_i^{k,3}, z_i^{k,5}\} \succ \{z_i^{k,4}, z_i^{k,5}\} \succ \{x_j^i, z_i^{k,6}\} \\
    & \succ \{z_i^{k,1}, z_i^{k,6}\}\succ \{z_i^{k,2}, z_i^{k,6}\} \succ \{z_i^{k,3}, z_i^{k,6}\} \succ \{z_i^{k,4}, z_i^{k,6}\} \succ \{z_i^{k,5}, z_i^{k,6}\}.
  \end{align*} 
  We extend this to a sublist $\mathcal{C}_i$ as follows:
  \[
    \{x^1_i, x_i^2\} \succ \{x_i^2, x^3_i\}\succ \{x_i^1, x_i^3\} \succ \mathcal{B}_i^1 \succ \mathcal{B}_i^2 \succ \mathcal{B}_i^3.
  \]
  The complete master list looks as follows.
  \[
    \mathcal{A}_1 \succ \dots \mathcal{A}_m \succ \mathcal{C}_1 \succ \dots \succ \mathcal{C}_n \succ \dots \text{, where the rest is arbitrarily ordered.}
  \]
  We call the constructed \textsc{3-DSR-ML} instance $\mathcal{I}'$.

  \subsubsection{Proof of the forward direction}

  We show how to construct a stable matching from a solution to a
  \textsc{1-in-3 Positive 3-Occurrence-SAT} instance.

\begin{lemma}\label{lforward}
  Let $f: \{x_i\} \rightarrow \{\LogicTRUE, \LogicFALSE\}$ be a solution to the \textsc{1-in-3 Positive 3-Occurrence-SAT} instance $\mathcal{I}$.
  Then $\mathcal{I}'$ admits a stable matching.
\end{lemma}

\begin{proof}
  We construct a stable matching $M$ as follows.
  We start with $M= \emptyset$.

  Denote by $S\coloneqq \{i \in [n] : f(x_i) = \LogicTRUE\}$ the set of indices such that the corresponding variables are set to $\LogicTRUE$ by truth assignment~$f$.
  
  For a clause $c_j$, let $y_j^{\ell_j}$ be the the variable in $c_j$ set to $\LogicTRUE$ by $f$.
  For each $j\in [m]$, we add the 3-set~$\{c_{j}, d_{j}, y_j^{\ell_j}\}$ to $M$.
  For all $i\in [n]\setminus S$, we add the 3-set~$\{x_i^1, x_i^2, x_i^3\}$ to~$M$.
  Finally, for each pair $(i,k) \in [n]\times[3]$, we add the 3-sets $\{z_i^{k,1}, z_i^{k,2}, z_i^{k,3}\}$ and $\{z_i^{k,4}, z_i^{k,5}, z^{k, 6}_i\}$.

  Since $f$ assigns exactly one true variable to each clause, $M$ is indeed a matching.
  It remains to show that~$M$ is stable.
  We do so by showing for each agent that it is not contained in a blocking 3-set.

  \begin{claim}\label{claim:c-d}
    For all $j\le m$, neither $c_j$ nor $d_j$ is contained in a blocking 3-set.
  \end{claim}
  \begin{claimproof}
	  We prove the claim by induction on~$j$.
    For $j= 0$, there is nothing to show, as there are no agents $c_0$ or $d_0$.

    For the induction step, first note that we can ignore all 2-sets containing an agent $c_p$ or $d_p$ for some $p < j$, as we already know that they are not contained in a blocking 3-set.
    Thus, we consider the sublist $\ML'$ of $\ML$ arising by deleting all such 2-sets.
    The first 2-set of~$\ML'$ is~$\{c_j, d_j\}$.
    Let $y_j^{\ell_j} = x_i^k$ such that $\{c_j, d_j, y_j^\ell\} \in M$.
    The variable $y_j^{\ell_j}$ is not contained in any blocking 3-set, as it is matched to the first 2-set of sublist $\ML'$.
    If $c_j$ is contained in a blocking~3-set, then the blocking 3-set is $\{c_j, y_j^p, d_j\}$ for some $p< \ell_j$, as $\{y_j^p, d_j\}$ for $p < \ell_j$ are the only 2-sets $c_j$ prefers to $\{y^\ell_j, d_j\}$.
    However, $d_j$ does not prefer $\{y^p_j, c_j\}$ to $\{y^\ell_j, c_j\}$, and thus, $\{c_j, y_j^p, d_j\}$ is not a blocking 3-set.
    By symmetric arguments, $d_j$ also cannot be contained in a blocking~3-set.
  \end{claimproof}
  \begin{claim}\label{claim:x-z}
    No agent $x_i^k$ or $z_i^{k, p}$ for $i \le n$, $k\in [3]$ and $p\in [6]$ is contained in a blocking 3-set.
  \end{claim}
  \begin{claimproof}
  We prove the claim by induction on $i$.
  For $i= 0$, there is nothing to show.
  
  Note that all 2-sets from $\bigcup_{\ell \in [m]} \mathcal{A}_\ell$ contain an agent of the type $c_\ell$ or $d_\ell$, and thus, no blocking 3-set contains a 2-set from $\bigcup_{\ell \in [m]} \mathcal{A}_\ell$ by \Cref{claim:c-d}.
  Furthermore, by the induction hypothesis, no 2-set from~$\mathcal{C}_q$ for $q < i$ can be contained in a blocking 3-set.
  Thus, it is enough to consider the sublist~$\ML_i$ of the master list arising through the deletion of $\mathcal{A}_j$ for all $j\in [m]$ and $\mathcal{C}_q$ for $q < i$.

  If $x_i$ is set to $\LogicTRUE$, then all agents $x_i^k$ are matched better than any 2-set from $\ML_i$, and thus, cannot be part of a blocking 3-set.
  Otherwise, all of $x^1_i$, $x^2_i$, and $x_i^3$ are matched to the first 2-set not containing itselves in~$\ML_i$, and thus are not contained in a blocking 3-set.

  Considering the sublist $\ML_i'$ arising from $\ML_i$ by deleting all 2-sets containing an agent~$x_i^k$, one can check that the first 15 sets of two agents of this sublist only consider agents of $z_i^{k, p}$ for~$p\in [6]$, and all agents $z_i^{k, p}$ are matched to one of these 15 sets of size two.
  Thus, any blocking 2-set containing an agents $z_i^{k, p}$ consists only of agents from $\{z_i^{k, q}: q\in [6]\}$.
  By enumerating all~20 such 3-sets, one easily verifies that none of them is blocking.
  \end{claimproof}

  Since the set~$A$ of agents in $\mathcal{I}'$ is $\{c_j, d_j : j\in [m]\} \cup \{x_i^k, z_i^{k, p} : i \in [n], k\in [3], p \in [6]\}$, the lemma now directly follows from \Cref{claim:c-d,claim:x-z}.
\end{proof}

\subsubsection{Proof of the backward direction}

Now, we show how to construct a solution to the \textsc{1-in-3 Positive 3-Occurrence-SAT}
instance~$\mathcal{I}$ from a stable matching.
First, we identify small subsets $A'$ of agents such that $\mathcal{I}'$ restricted to the agents from $A'$ does not admit a stable matching, implying that at least one agent of $A'$ must be matched to at least one vertex outside~$A'$.

  \begin{lemma}\label{linstability}
    For any $i\in  [n]$ and $k\in [3]$, the subinstance~$\mathcal{I}_i^k$ of $\mathcal{I}'$, which results from~$\mathcal{I}'$ by deleting all but the agents $x_i^k$ and $\{z_i^{k, p}: p\in [6]\}$, does not admit a stable matching.
  \end{lemma}

  \begin{proof}
    Assume for the sake of contradiction that there exists a stable matching~$M$.
    Note that deleting~$z_{i}^{k, 6}$ from $\mathcal{I}_i^k$ results in an instance identical to the instance~$\mathcal{I}_{\instable}$ (this can be seen by renaming agent~$x_i^k$ to $a$, agent~$z_i^{k, p}$ to $b$, agent $z_i^{k, 2}$ to $c$, agent~$z_i^{k, 3}$ to~$d$, agent~$z_{i}^{k, 4}$ to~$e$, and agent~$z_{i}^{k, 5}$ to $f$).
    Thus, \Cref{oinstability} implies that $M$ matches agent $z_{i}^{k, 6}$ to a 2-set~$\{x,y\}$.
    As all 2-sets containing~$z_i^{k, 6}$ appear on the end of the preference lists, $x$ and $y$ together with any unmatched agent form a blocking 3-set, a contradiction to the stability of~$M$.    
  \end{proof}

  We now show that the two agents $c_j$ and $d_j$ created for clause~$C_j$ have to be matched to an agent corresponding to a literal in this clause;
  indeed, we will later see that this literal satisfies the clause $C_j$ in the found solution.

\begin{lemma}\label{lcidi}
  In any stable matching $M$ and for each $j\in [m]$, there is an $\ell\in [3]$ such that~$\{c_j, d_j, y^\ell_j\}\in M$.
\end{lemma}

\begin{proof}
  We prove the lemma by induction on $j$.

  \bfseries Base case: \mdseries
  If $M$ does not contain the 2-set $\{c_1, d_1, y_1^\ell\}$ for all $\ell \in [3]$, then $\{c_1, d_1, y^\ell_1\}$ is a blocking 2-set for every~$\ell\in [3]$.

  \bfseries Induction step: \mdseries
  By the induction hypothesis, no agent~$c_p$ or $d_p$ for $p < j$ is matched to~$c_j$,~$d_j$, or some $y^\ell_j$.
	Hence, neither $c_j$ nor $d_j$ nor $y^\ell_j$ is matched to a 2-set which comes before~$\mathcal{A}_j$ in the master list.
  Thus, if $M$ does not contain $\{c_j, d_j, y_j^\ell\}$ for all $\ell \in [3]$, then $\{c_j, d_j, y^\ell_j\}$ is a blocking 3-set for every~$\ell\in [3]$, contradicting the stability of $M$.
\end{proof}

We now want to show that for any $i \in [n]$, a stable matching contains either $\{x_i^1, x_i^2, x_i^3\}$ or matches $x_i^k$ to 2-sets of the form $\{c_j, d_j\}$.
In order to do so, we first show that for any $i\in [n]$ and $k \in [3]$, agents $z_i^{k, 1}, \dots, z_i^{k, 6}$ are matched to two 3-sets in any stable matching~$M$ unless at least one agent~$z_i^{k, p}$ is matched to a 2-set containing an agent $c_j$, $d_j$, $x_{i'}^{k'}$, or $z_{i'}^{k', q}$ for some $j\in [m]$, $(i' , k')  \in [n]\times [3]$ with $i' <i$ or $i' = i$ and $k' \le k$, and $q\in [6]$.

\begin{lemma}
\label{lem:zik}
  Let $M$ be any stable matching, and let $i\in [n]$ and $k \in [3]$.
  Let $X\coloneqq \{c_j, d_j : j\in [m]\} \cup \{ x_{i'}^{k'}, z_{i'}^{k', q} : i' < i, k' \in [3], q\in [6]\} \cup \{ x_{i}^{k'}, z_{i}^{k', q} : k' \le k, q\in [6]\}$, and let $Z_i^k \coloneqq \{z_i^{k, p} : p \in [6]\}$.
  If no agent $z_{i}^{k, p}$ is matched to a 2-set containing an agent from $X$, then $M$ contains two 3-sets $t_1$ and $t_2$ which are subsets of $Z_i^k$.
\end{lemma}

\begin{proof}
  Note that every 2-set before $\mathcal{B}_i^k$ in the master list contains an agent from $X$.
	List~$\mathcal{B}_i^k$ contains all 2-sets $\{z_i^{k,p}, z_i^{k,q}\}$ for $p, q\in [6]$ with $p\neq q$ (as well as some 2-sets containing~$x_i^k \in X$).
  Now assume for a contradiction that the lemma does not hold, i.e., there exists some~$z\in Z$ which is not matched to a 2-set $\{z', z''\}$ with $z' , z'' \in Z_i^k$.
  Then there exist three such agents~$z_1, z_2, z_3 \in Z_i^k$.
  Hence, $z_1,z_2$, and $z_3$ are matched to 2-sets which appear after~$\mathcal{B}^k_i$ in the master list.
  Consequently, $\{z_1, z_2, z_3\}$ is a blocking 3-set, contradicting the stability of~$M$.
\end{proof}

Now we can show the following structural statement about the agents $x_i^1$, $x_i^2$, and $x_i^3$, essentially stating that if one of these agents is matched to a 2-set~$\{c_j, d_j\}$ (corresponding to setting variable~$x_i$ to $\LogicTRUE$), then all three of them are.
From this statement, the backward direction of the correctness proof for the reduction 
will then easily follow.

\begin{lemma}\label{lvertexAgents}
  Let $M$ be any stable matching.
  Then for all $i\in [n+1]$, the following holds:
  
	For all $i^* < i$ either~$\{x_{i^*}^1, x_{i^*}^2, x_{i^*}^3\} \in M$ or for each $k\in [3]$, there exists some $j\in [m]$ such that $\{x_{i^*}^k, c_j, d_j\} \in M$.
  Furthermore, for each $k\in [3]$, matching~$M$ contains two 3-sets $t_1$ and~$t_2$ with $t_1, t_2\subseteq \{z_{i^*}^{k, p} : p\in [6]\}$.

\end{lemma}

\begin{proof}
  We prove the lemma by induction on $i$.
  For $i = 1$, there is nothing to show.

  By the induction hypothesis, for every $q < i$, no agent $x_q^k$ or $z_q^{k, p}$ is matched to a 2-set containing an agent~$x_q^r$ or $z_q^{r,s}$ with $r \neq q$, and by \Cref{lcidi}, no agent $z_i^{k, p}$ is matched to a 2-set containing an agent $c_j$ or $d_j$.

  Let $S\coloneqq \{k\in[3] : \exists j\in [m] \text{ s.t.\,} \{x_i^k, c_j, d_j\}\in M\}$ be the set of indices $k\in [3]$ such that $x_i^k$~is matched to a 2-set of the form $\{c_j, d_j\}$.

  {\bfseries Case 1:} $|S| = [3]$.

  By induction on $k$, we can apply \Cref{lem:zik} for every $k\in [3]$, showing that $M$ contains two 3-sets $t_1,t_2 \subseteq \{z_i^{k, p}: p \in [6]\}$. 

  {\bfseries Case 2:} $S = \emptyset$.

  Then $M$ contains $\{x_i^1, x_i^2, x_i^3\}$, as all 2-sets before sublist~$\mathcal{C}_i$ contain an agent which is not matched to an agent $x_i^k$ in $M$.
  As in Case 1, induction on $k$ together with \Cref{lem:zik} implies that for every $k\in [3]$, matching~$M$ contains two 3-sets $t_1$ and $t_2$ with $t_1, t_2\subset \{z_i^{k, p} : p\in [6]\}$.

  {\bfseries Case 3:} $|S| \in \{1, 2\}$.

  We show that this case leads to a contradiction and, therefore, cannot occur.

  {\bfseries Case 3 (a):} No 3-set in $M$ contains two agents of the form $x_i^k$ for $k\in [3]$.

  If there is no $j\in [m]$ such that $\{x_i^1, c_j, d_j\} \in M$, then the agents $x_i^1$, $z_i^{1, 1}$, \dots, $z_i^{1, 5}$ have to be matched to two 3-sets,  as all agents appearing in a 2-set before $\mathcal{B}_i^1$ in the master list cannot be matched to $x_i^1$ or $z_i^{1, p}$ because of \Cref{lcidi} and the induction hypothesis.
  By \Cref{linstability}, this implies that $M$ contains a blocking 3-set inside $x_i^1$, $z_i^{1, 1}$, \dots, $z_i^{1, 5}$, contradicting the stability of~$M$.
  Thus, there exists some $j\in [m]$ such that $\{x_i^1, c_j, d_j\} \in M$.
  Then $M$ contains two 3-sets~$t_1$ and $t_2$ with $t_1, t_2 \subseteq \{z_i^{1, p} : p \in [6]\}$ by \Cref{lem:zik}.

  We can conclude then by the same argument that $x_i^2$ is matched to a 2-set $\{c_{j'}, d_{j'}\}$ for some~$j'\in [m]$, and from this that $x_i^3$ is matched to a 2-set $\{c_{j''}, d_{j''}\}$ for some~$j''\in [m]$, a contradiction.

  {\bfseries Case 3 (b):} There is a 3-set $\{x_i^{k_1}, x_i^{k_2}, z\}\in M$ with $z\notin \{x_i^{k}\}$.

  If $\{x_i^1, c_j, d_j\}\in M$ for some $j\in [m]$, then $M$ contains two tuples inside $z_i^{1, 1}, \dots, z_i^{1,6}$ by \Cref{lcidi,lem:zik} and the induction hypothesis.
  It follows that there exists a blocking 3-set inside $x_i^2, z_i^{2, 1}, \dots, z_i^{2,5}$ by \Cref{linstability}.\

  Consequently, we can assume in the following that $\{x_i^1, c_j, d_j\} \notin M$ for all $j\in [m]$.
  We may also assume that $k_1 =1 $ or $k_2 = 1$, since otherwise $\{x_i^1, x_i^2, x_i^3\}$ was a blocking 3-set.
  Then $x_i^{1}$~prefers to be matched to the 2-set $\{z_i^{1,1}, z_i^{1,2}\}$ or $\{z_i^{1,3}, z_i^{1,5}\}$.
	Every 2-set which agent~$z_i^{1,1} $ respectively $z_i^{1,2}$ prefers to $\{x_1^1, z_i^{1 , 2}\}$ respectively $\{x_1^1, z_i^{1 , 1}\}$ is one of the 2-sets~$\{x_i^1, x_i^2\}$, $\{x_i^1, x_i^3\}$, or $\{x_i^2, x_i^3\}$, contains an agent $x^j_k$ with $j< i$, or contains an agent $c_j$ or $d_j$ for some~$j\in [m]$.
  Thus, unless $z_i^{1,1} = z$ or $z_i^{1,2} = z$ holds, $z_i^{1,1}$ and $z_i^{1,2}$ also prefer to be matched by $\{x_i^{1}, z_i^{1,1}, z_i^{1,2}\}$.
  If $z = z_i^{1,1}$ or $z = z_i^{1,2}$, then $z_i^{1,3}\neq z$ and $z_i^{1,5} \neq z$, and by arguments symmetrically to the above ones, $\{x_i^1, z_i^{1,3}, z_i^{1,5}\}$ is blocking.
\end{proof}

The backward direction now easily follows.

\begin{lemma}\label{lbackward}
  If there exists a stable matching $M$, then there is a truth assignment for $\mathcal{I}$ satisfying exactly one literal in each clause.
\end{lemma}

\begin{proof}

  Consider the assignment $f: \{x_i\} \rightarrow \{\LogicTRUE, \LogicFALSE\}$, with $f (x_i) := \LogicTRUE$ if and only if all $x_i^k$ are matched to a 2-set of the form $\{c_j, d_j\}$.

  Assume that $f$ is not a solution to $\mathcal{I}$. We distinguish two cases.

  \bfseries Case 1: \mdseries
  There is a clause $C_j$ which is not satisfied by $f$.

  By \Cref{lcidi}, for each $j\in [m]$, there exists some $\ell \in [3]$ such that $\{c_j, d_j, y^\ell_j\}\in M$.
  Let~$y^\ell_j = x_i^k$.
  By \Cref{lvertexAgents}, for every $k' \in [3]$, agent $x_i^{k'}$ is matched to 2-set $\{c_p, d_p\}$ for some~$p \in [m]$.
  Thus, the clause $C_j$ is satisfied by $f$, a contradiction.

  \bfseries Case 2: \mdseries
  There is a clause $C_j$ which is satisfied by at least two variables $x_i$ and $x_{i'}$.

  Matching~$M$ can only contain one 3-set containing $c_j$ and $d_j$, and so without loss of generality $x_i^k$ is not matched to $\{c_j, d_j\}$ for any $k\in [3]$.
  By \Cref{lcidi}, literal $x_i^{k'}$ contained in $C_j$ does not match to any 2-set $\{c_q, d_q\}$ for $q\in [m]$, and thus, we have $f (x_i ) = \LogicFALSE$, a contradiction.

  Altogether, we conclude that $f$ is a solution to $\mathcal{I}$.
\end{proof}

  Finally, we are ready to prove the main theorem of this section.

\begin{theorem}\label{tNPc}
  \textsc{3-DSR-ML} is $\NP$-complete.
\end{theorem}

\begin{proof}
  Observe that \textsc{3-DSR-ML} is in $\NP$ (the stability of a matching can be checked in $O(n^3)$~time by enumerating all 3-sets, where $n$ is the number of agents).

  The reduction described in \Cref{sec:NPh-reduction} can be performed in linear time.
  Thus, the $\NP$-completeness of \textsc{3-DSR-ML} follows directly from Lemma~\ref{lforward} and Lemma~\ref{lbackward}.
\end{proof}

We have seen that the strong restriction that the preferences of each agent are complete and derived from a common master list of $d$-sets presumably does not lead to an efficient algorithm, even for $d=2$.
Thus, in order to get tractable cases, other restrictions of the preferences are needed.
One possibility would be to additionally require that every agent has consistent preferences.
In this case, the master list implies that the preferences of every agent is indeed derived from a master poset which is a strict order, and we will see in the beginning of \Cref{sec:ML-agents} that this implies that there exists a unique stable matching which can be found efficiently.

\section{Master poset of agents}
\label{sec:ML-agents}

In this section, we consider the case when there does not exist a master list of $(d-1)$-sets of agents, but a master poset~$\succ_{\ML}$ of single agents; in other words, we study the complexity 
of the problem \dsmpo{}.
Each agent can derive its preferences from this master list, meaning that if for two $(d-1)$-sets~$t\neq t'$, one can find a bijection $\sigma$ from the elements of~$t$ to the elements of $t'$ such that $a\succeq_{\ML} \sigma (a)$ for all~$a\in t$, then any agent (not occurring in $t$ or $t'$) shall prefer~$t$ to $t'$.
We show that this problem is easily polynomial-time solvable if the master poset is a strict order (\Cref{sec:strict-order}).
Afterwards, following the approach of distance-from-triviality parameterization~\cite{GHN04,Nie06},
we generalize this result by showing fixed-parameter tractability for the parameter $\kappa$, the ``maximum number of agents incomparable to a single agent'' (\Cref{sec:fpt}).
On the contrary, for the stronger parameter width of the poset, we show $\Wone$-hardness (\Cref{sec:width}), leaving open whether it can be solved in polynomial time for constant width (in parameterized complexity known as the question for containment in~\XP).
Afterwards, in \Cref{sec:del-dist}, again employing a distance-from-triviality parameterization we show \Wone-hardness of \dsmpo parameterized by the number of agents one needs to delete in order to have the preferences derived from a strict order.
Finally, we show that the variation of \dsmpo where agents may declare an arbitrary part of $(d-1)$-sets as unacceptable (that is, this agent may not be matched to such a $(d-1)$-set) is \NP-complete even if the master poset is a strict order (\Cref{sec:incomplete}).
Note that in order to distinguish from the master list of $(d-1)$-sets in \Cref{sec:ml-tuples}, we will always refer to the master poset as a poset, even if it is a strict order.

\subsection{Strict orders}
\label{sec:strict-order}
We consider the case that the master poset is a strict order.
Then, an easy algorithm solves the problem:
Just match the first $d$ agents from the master poset together, delete them, and recurse.
Note that the preferences of any agent cannot be directly derived from the master poset, as e.g. an agent may prefer either $\{a_1, a_4\}$ to~$\{a_2, a_3\}$ or $\{a_2, a_3\}$ to~$\{a_1, a_4\}$.
Thus, the input contains the complete preferences of all agents, and the input size is $\Theta(d\binom{n}{d-1})$.
In this sense, the running time
of our algorithm is sublinear.

\begin{proposition}\label{tconsistent}
  If $\succeq_{\ML}$ is a strict order, then any \dsmpo instance admits a uniquely determined 
  stable matching.
  Assuming that the poset is given as a ranking $a_1 \succ_{\ML} a_2 \succ_{\ML} \dots \succ_{\ML} a_n$, this stable matching can be found in $O(n)$~time, where $n$ is the number of agents.
  If the poset is given via pairwise comparisions, then the unique stable matching can be found in $O( n^2)$ time.
\end{proposition}

\begin{proof}
  We number the agents in such a way that $a_1\succ_{\ML} a_2\succ_{\ML} \dots \succ_{\ML} a_n$.

  We claim that $M \coloneqq\{\{a_{d (i-1) +1}, a_{d(i -1) + 2},\dots, a_{di}\} : 1\le i \le \lfloor\frac{n}{d}\rfloor\}$ is a stable matching.
  We prove this claim by contradiction, so assume that there is a blocking $d$-set $\{a_{i_1}, a_{i_2},\dots,  a_{i_d}\}$ with $i_1 < i_2 < \dots < i_d$.
  Let~$\{a_{i_1}, b_2, b_3, \dots, b_{d}\}\in M$ be the $d$-set containing $a_{i_1}$.
  Note that such a $d$-set exists as~$M$ leaves at most the $d-1$ last agents of the master poset unmatched, and the agents~$a_{i_j}$, $j\in \{2, 3, \dots, d\}$, are ranked after~$a_{i_1}$ in the master poset.
  Since $a_{i_j}$ is after $a_{i_1}$ in the master poset, we have $b_j \succeq_{\ML} a_{i_{j}}$ for all $j\in \{2, 3, \dots, d\}$.
  Thus, $a_{i_1}$ cannot prefer~$\{a_{i_2}, \dots, a_{i_d}\}$ to $M$, a contradiction.
  If the master poset is given as the order $a_1 \succ_{\ML} \dots \succ_{\ML} a_n$, then $M $ can clearly be constructed in $O(n)$ time.
	If the master poset is given via pairwise comparisions, then we compute the strict order $a_1 \succ_{\ML} \dots \succ_{\ML} a_n$ in $O(n ^2)$ time using a sorting algorithm, and from this, we can compute $M$ in $O(n)$ time.

  It remains to show that $M$~is uniquely determined.
  Assume that there is a stable matching~$M' \neq M$.
  Let $i$ be the smallest index such that $a_i$ is matched differently in~$M$ and~$M'$.
  Since only the at most $d-1$ agents with highest index are unmatched in~$M$, we get that $a_i$ is matched in $M$.
  As either all or no agent from a $d$-set $t\in M$ are matched differently in $M$ and $M'$, and all $d$-sets from $M$ are of the form $\{a_{d (j-1) +1}, a_{d (j-1) + 2}, \dots, v_{dj}\}$, it follows
	that $i= d (j-1) +1$ for some $j\in [\lfloor \frac{n}{d}\rfloor]$.

  We claim that $t\coloneqq\{a_i, a_{i+1}, a_{i+2}, \dots, a_{i+d-1}\}$ is a blocking $d$-set for $M'$.
	By the definition of $i$, matching $M'$ does not contain a $d$-set with one agent with index smaller than~$i$ and one agent with index at least $i$, and we have that $t \notin M'$.
  Thus, for any $a\in t$ with $t_a' \coloneqq M' (a)\neq \emptyset$, the bijection $\sigma_v : t\setminus \{a\} \rightarrow t_a'\setminus \{a\}$ matching the agent with the $j$-th-lowest index in $t\setminus \{a\}$ to the agent with the $j$-th-lowest index in $t'_a\setminus\{a\}$ satisfies $b \succeq_{\ML} \sigma (b)$ for all $b\in t\setminus \{a\}$.
  Since the preferences are derived from $\succeq_{\ML}$, it follows that $a$ prefers $t\setminus \{a\}$ to $t_a'$.
  Thus, $t$ is a blocking $d$-set.
\end{proof}

The polynomial-time solvability of the special case from \Cref{tconsistent} motivates three
distance-from-triviality parameterizations studied in the following.

\subsection{Posets}
\label{sec:posets}

In two-dimensional stable (or popular) matching problems with master lists, also reflecting the needs of typical real-world applications, the master list usually contains ties~\cite{BiroIS11,IMS08,KavithaNN14,OMalley07,PerachPR08}.
In the following, we allow the master list not only to contain ties, but to be an arbitrary poset.
In this case, the problem clearly is \NP-complete, as the poset where each agent is incomparable to each other agent does not pose any restrictions on the preferences of the agents~\cite{NH91}.
Hence, we consider several parameters measuring the similarity of the poset to a strict order---this is
our polynomial-time solvable special case of the previous section.
Thus, for the parameter ``maximum number of agents incomparable to a single agent'', we show fixed-parameter tractability in \Cref{sec:fpt}, and for the stronger parameter width of the poset, we show \Wone-hardness in \Cref{sec:width}.
Here, by ``stronger'' we mean that there are posets for which the width is bounded, while the maximum number of agents incomparable to a single agent is not, while the converse is not true (as an antichain of size~$k$ implies an agent with $k-1$ agents incomparable to it).
Consider for example the poset $\succ$ over a set of $2n$ agents $a_1, \dots, a_n , b_1, \dots, b_n$ with $a_i \succ a_j$ and $b_i \succ b_j$ for all $i <j$ as well as $a_i \perp b_j$ for all $i, j\in [n]$.
Then the width of this poset is two (as it can be decomposed in the two chains $a_1 \succ a_2 \succ \dots \succ a_n$ and $b_1 \succ b_2 \succ \dots \succ b_n$), while every agent is incomparable to $n$ other agents.

\subsubsection{Maximum number of agents incomparable to a single agent}
\label{sec:fpt}

In this section,
we show that \dsmpo is fixed-parameter tractable when parameterized by~$\kappa (\succ_{\ML})$ (recall that $\kappa (\succ_{\ML})$ denotes the maximum number of agents incomparable to a single agent in the master poset~$\succ_{\ML}$).
As a first step of the algorithm, we show how to ``approximate'' the given poset by a strict order, meaning that for any two agents $a$ and $b$ with~$a$ being before~$b$ in the strict order, we have $a \succ_{\ML} b$ or $a \perp_{\ML} b$, and if $a$ is ``much earlier'' in the strict order than $b$, we have that~$a \succ_{\ML} b$.

\begin{lemma}\label{lpo}
  For any poset $(A, \succeq)$, there is an order $a_1, a_2, \dots, a_n$ of\/ $A$~such that
    (i) for all~$i < j$, we have that $a_i \succ a_j$ or $a_i \perp a_j$, and
    (ii) for all $j > i + 2\kappa (\succeq)$, we have $a_i \succ a_j$.
  Moreover, such an order can be found in $O(|A|^2)$ time.
\end{lemma}

\begin{proof}
  We prove the statement by induction on~$|A|$.
  If $|A| = 1$, then there is nothing to show.
  So assume $|A| > 1$.
  
  Let $a_1 \in A$ be an element such that $a_1 \succeq a$ or $a_1 \perp a$ for all $a\in A$.
  Such an element~$a_1$ has to exist in any poset. 
  By induction, we can find an order $a_2, \dots, a_{|A|}$ of $A\setminus \{a_1\}$ satisfying the lemma; note that $\kappa (A) \ge \kappa (A\setminus \{a_1\})$.
  We then add $a_1$ at the beginning of $a_2, \dots, a_{|A|}$.
  Let $A'$ be the set of elements incomparable with $a_1$.
  It remains to show that the elements from $A'$ are among the $2 \kappa (\succeq)$ first elements.
  Note that there is no $a' \in A' $ and $a\in A\setminus A'$ with $a \succ a'$, as otherwise $a_1 \succeq a \succ a'$ and thus~$a_1 \succ a'$, but by the definition of $A'$ we have $a_1 \perp a'$.
  Thus, for any $a'\in A'$, there are at most $|A'| + \kappa_\succ (a') \le 2\kappa (\succeq)$ elements before $a'$, and thus, $a_1, \dots, a_{|A|}$ satisfies the lemma.

  Since we can compute an element~$a_1$ with $a_1 \succeq a$ or $a_1 \perp  a$ for all $a\in A\setminus \{a_1\}$ in linear time, the order can be found in quadratic time.
\end{proof}

For the remainder of \Cref{sec:fpt}, we fix an instance $\mathcal{I} = (A, (\succeq_a)_{a\in A}, \succeq_{\ML})$ of \dsmpo, and an order $a_1, \dots, a_n$ of the agents in~$A $ fulfilling the conditions of \Cref{lpo} for the poset $(A, \succeq_{\ML})$.
Let $\kappa \coloneqq \kappa (\succeq_{\ML})$.
Furthermore, let $A[{\le i}] := \{a_1, \dots, a_i\}$, let $A {[i,j]} := \{a_i, a_{i+1}, \dots, a_j\}$, and let $A[{\ge i}] := \{a_i, a_{i+1}, \dots, a_n\}$ for any $i\in [n]$.

We now show that the agents contained in a $d$-set of a stable matching are close to each other in the order $a_1, \dots, a_n$.

\begin{lemma}\label{lem:small-dist}
  Let $\mathcal{I} = (A, (\succeq_a)_{a\in A}, \succeq_{\ML})$ be an \dsmpo-instance and let $a_1, \dots, a_n$ be an order of the agents in $A$ such that this order fulfills \Cref{lpo} for the poset $(A, \succeq_{\ML})$.

  For any stable matching $M$ and any $d$-set $\{a_{i_1}, a_{i_2},\dots, a_{i_d}\} \in M$ with $i_1 < i_2 < \dots < i_d$, it holds that $i_{j+1} - i_{j} \le \upperBoundDistance$ for all $j\in [d-1]$.
\end{lemma}

\begin{proof}
  Let $M$ be a stable matching, and $\{a_{i_1}, a_{i_2},\dots, a_{i_d}\} \in M$ be a $d$-set contained in~$M$.
  We assume $i_1 < i_2 < \dots < i_d$, and fix some $j\in [d-1]$.

	Let $\mathcal{T}^+$ be the set of $d$-sets in $M$ containing at least one agent from $A {[i_{j} + 2 \kappa + 1, i_{j+1} - 2 \kappa -1]}$ and at least one agent from $A [{\ge i_{j + 1} - 2 \kappa}]$ and
  let $\mathcal{T}^-$ be the set of $d$-sets in $M$ containing at least one agent from $A {[i_{j} + 2 \kappa + 1, i_{j+1} - 2 \kappa -1]}$, and at least one agent from $A [{\le i_{j} + 2 \kappa}]$.
  We now give an example for the definitions of $\mathcal{T}^+$ and $\mathcal{T}^-$.

  \begin{example}
   Let $d = 4$, $\kappa =5 $, and $M$ be a stable matching.
   Assume that $M$ contains the $4$-set~$\{a_3, a_{14}, a_{50}, a_{157}\}$.
   Thus, $i_1 = 3$, $i_2 = 14$, $i_3 = 50$, and $i_4 = 157$.
   For instance, taking $j = 3$, the set $\mathcal{T}^+$ contains all $4$-sets from~$M$ containing an agent from $\{a_{61}, a_{62}, \dots, a_{146}\}$, an agent from $\{a_{147}, a_{148}, \dots, a_n\}$, and two more arbitrary agents.
   The set $\mathcal{T}^-$ contains all $4$-sets from~$M$ containing an agent from $\{a_1, a_2, \dots, a_{60}\}$, an agent from $\{a_{61}, a_{62}, \dots, a_{146}\}$, and two more arbitrary agents.
  \end{example}

  Now, let $t$ be a $d$-set from $ \mathcal{T}^+$.
  We claim that for every $d$-set $t'\in \mathcal{T}^+$ other than~$t$, there exist agents~$a \in t$ and $a'\in t'$ with $a\perp_{\ML} a'$.
  Assume for a contradiction that there are two~$d$-sets~$t, t'\in \mathcal{T}^+$ such that there do not exist $a\in t$ and $a'\in t'$ with $a\perp_{\ML} a'$.
  Let $t^*$ contain the $d$ agents from $t \cup t'$ with minimum index.
  By the definition of~$\mathcal{T}^+$, any $d$-set from~$\mathcal{T}^+$ contains an agent from~$A [{\le i_{j+1} - 2\kappa -1}]$ and one agent from~$A [{\ge i_{j+1} -2\kappa}]$.
  Therefore, at least one agent of~$t^*$ is contained in $t$, and at least one agent of~$t^*$ is contained in $t'$.
  For any agent $a_p\in t \setminus t^*$ and any $a_q\in t' \cap t^*$, it holds by the definition of $t^*$ that $q < p$.
  By \Cref{lpo}, it follows that $a_q \succ_{\ML} a_p$ or $a_q \perp_{\ML} a_p$.
  However, the latter is not possible, since we assumed that there are no two agents $a \in t$ and $a'\in t'$ with $a \perp_{\ML} a'$.
  Thus, we have that each $a\in t \cap t'$ prefers $t^*$ to~$t$, and by symmetric arguments also each $a'\in t'\cap t^*$ prefers $t^*$ to $t'$.
  It follows that the $d$-set~$t^*$ is blocking, contradicting the assumption that $M$ is stable.

  As any agent is incomparable to at most $\kappa $ other agents, it follows that $|\mathcal{T}^+|\le \kappa d + 1$.
  By analogous arguments, one can show that $|\mathcal{T}^-| \le \kappa d + 1$.

  Any $d$-set $s\in M$ consisting solely of agents from
	$A {[i_j + 2\kappa + 1, i_{j+1} - 2\kappa -1]}$ directly implies a blocking $d$-set $\{a_{i_1},\dots, a_{i_j}\} \cup s_{d-j}$, where $s_{d-j}$ is an arbitrary subset of~$s$ containing exactly $d -j$ agents.

  Hence, $M$~contains at most $2(\kappa d + 1)$ sets containing an agent from $A {[i_{j} + 2 \kappa + 1, i_{j+1} - 2 \kappa -1]}$, implying that $(i_{j+1} - 2\kappa -1) - (i_j + 2\kappa + 1) \le d \cdot 2(\kappa d + 1) + d -1$, where $d-1$ is added since there can be at most $d-1$ unmatched agents.
  It follows that $i_{j+1} - i_j \le 2\kappa d^2 +  4\kappa + 3d +1$.
\end{proof}

\Cref{lem:small-dist} implies that in order to find a stable matching, we only have to consider matchings~$M$ such that for every $d$-set $t\in M$, we have for any two agents $a_i, a_j\in t$ that $|i -j| \le d (2\kappa d^2 +  4\kappa + 3d +1)$.
We will call such matchings \emph{local}.
We now develop a dynamic program which decides whether there is a local and stable matching using agents $a_1, \dots, a_i$ for every $i\in [n]$, resulting in an \FPT-algorithm for the combined parameter $\kappa + d$.

\begin{proposition}\label{thm:fptk+d}
  \dsmpo can be solved in $O(n^2) + (\kappa d^4)^{O(\kappa d^4)} n$ time, where $\kappa$ is the maximum number of agents incomparable to a single agent, $d$ is the dimension (i.e., the group size), and $n$ is the number of agents.
\end{proposition}

\begin{proof}
  We first apply \Cref{lpo} to the poset $(A, \succeq_{\ML})$ to get an order $v_1, \dots , v_n$ of the agents in $O(n^2)$ time.
  Let $k\coloneqq 2d(d - 1) (\upperBoundDistance) $.

Our dynamic programming table~$\tau$ has an entry $\tau[i, M]$ for each $i\in [n]$ and each local matching $M$ such that any $d$-set~$t\in M$ contains at least one agent of $a_i, \dots, a_{i+k}$.
  This entry shall be $\myone$ if and only if $M$ can be extended to a local matching $M^*$ not admitting a blocking $d$-set consisting solely of agents from $a_1, \dots, a_{i+k}$. 
  By \Cref{lem:small-dist}, there exists a stable matching if and only if $\tau[n-k, M] = \myone$ for some local matching $M$.
  Thus, it remains to show how to compute these values.

  For $i = 1$, we set $\tau [ 1, M] \coloneqq \myone$ if and only if $M$ does not contain a blocking~$d$-set inside $a_1, \dots, a_{k}$.

  To compute $\tau[i, M]$ for $i > 1$, we need to determine whether we can extend $M$ to a local matching~$M^*$ on $a_1, \dots, a_{i+k}$ such that every blocking $d$-set involves an agent $a_j$ with $ j > i+ k$.
  Any such extension~$M^*$ induces a matching~$M_{i-1}$ by taking all $d$-sets containing an agent from $a_{i-1}, \dots, a_{i - 1 +k}$, and $M^*$ also witnesses that $\tau [i-1, M_{i-1}] = \myone$.
  However, given a matching~$M_{i-1}$ with $\tau [i-1, M_{i-1}] = \myone$ and $M_{i-1} (a_j) = M (a_j)$ for all $j \in \{i, i+1, \dots, i +k-1\}$ does not imply that $\tau [i, M] = \myone$, as there might be a blocking $d$-set involving $a_{i + k}$ and $d-1$ other agents from $a_1, \dots, a_{i -1 + k}$.
  Since $M_{i-1} $ does not store how $a_1, \dots, a_{i-2}$ are matched, we cannot just enumerate all such $d$-sets and check whether they are blocking, but only can do this for $d$-sets consisting solely of agents from $a_{i-1}, \dots, a_{i+k}$.
  We will show that for any matching~$M_{i-1}^*$ witnessing that $\tau[i-1, M_{i-1}] = \myone$ that no blocking $d$-set containing $a_{i+k}$ and at least one agent from $a_1, \dots, a_{i-1+k}$ can occur;
  therefore, it is enough to check for blocking $d$-sets containing only agents from $a_{i-1}, \dots, a_{i+k}$.
	Thus, in order to compute $\tau[i, M]$,
	we look up whether there exists a local matching $M_{i - 1}$ with $\tau[i-1, M_{i-1}]\coloneqq \myone$ and $M(a_j ) = M_{i-1}(a_j)$ for all $j\in [i, i+k -1]$ such that $M_{i-1}\cup M$ does not admit a blocking $d$-set consisting of agents from~$a_{i - 1}, \dots, a_{i+k}$.
  If this is the case, then we set $\tau[i, M] = \myone$, and otherwise we set $\tau[i, M] = \myzero$.

  Since there are at most $k^{O(k)}$ partitions of a $k$-element set~\cite{deBrujin58}, table~$\tau$ contains at most $n k^{O(k)}$ entries.
  Each entry can be computed in $k^{O(k)}$ time, resulting in an overall running time of $k^{O(k)} n = (\kappa d^4)^{O(\kappa d^4)}n$.

  It remains to show the correctness of this dynamic program, which we do by induction.
  For~$i=1$, the values $\tau[i, M]$ are computed correctly by definition.
  Let $i > 1$.
  First assume that $\tau [ i, M] = \myone$.
  Let $M_{i - 1}$ be the local matching with $\tau[ i-1, M_{i-1}] = \myone$ and $M(a_j ) = M_{i-1}(a_j)$ for all $j\in [i, i+k -1]$ and $M_{i-1}\cup M$ not admitting a blocking $d$-set consisting of agents from~$a_{i-1}, \dots, a_{i-k}$.
  By the induction hypothesis, $\tau[i-1, M_{i-1}]$ was computed correctly, implying that there exists a local matching~$M_{i-1}^*$ on $a_1, \dots, a_{i + k-1}$ such that for each $j\in [i-1, i + k - 1]$, we have $M_{i-1} (a_j) = M_{i-1}^* (a_j)$, and no blocking $d$-set consists solely of agents from $a_1, \dots, a_{i+k -1}$.
  We define $M^*\coloneqq M_{i-1}^* \cup M$.
  Note that $M^*$ is indeed a (local) matching since $a_{i-1} $ cannot be matched to agents with index at least $i+k$ and $a_{i+k}$ cannot be matched to agents with index at most $i-1$ by the locality of $M_{i-1}^*$ and $M$.
  Clearly, it holds that $M(a_j) = M^*(a_j) $ for all~$j\in [i, i+k]$.
  Next, we show that~$M^*$ contains no blocking $d$-set consisting solely of agents from~$A[{\le i+ k}]$.
  Any such $d$-set~$t$ must contain~$a_{i+ k}$, since $M_{i-1}^*$ does not admit blocking $d$-sets consisting solely of agents from $a_1, \dots, a_{i-1 + k}$.
  Furthermore, it must contain at least one agent $a_j$ with $j < i -1$ since we checked for all $d$-sets of agents from $a_{i-1}, \dots, a_{i+k}$ whether they are blocking.
  Since $k = 2d (d - 1) \upperBoundDistance$, it follows that there exists some $\ell \in [i + (d-1)\upperBoundDistance, i + k - (d -1) \upperBoundDistance]$ such that $t$ contains no agent from $A {[\ell - (d-1) (\upperBoundDistance), \ell + (d-1) (\upperBoundDistance)]}$, and $a_\ell $ is matched in~$M^*$ (note that such an agent exists since at most $d-1$ agents from $A[{\le i + k -1}]$ can be unmatched in $M^*$).
  Let~$t'\in M^*$ with~$a_\ell\in t'$, and let $t^*$ contain the $d$ agents with minimum index from~$t\cup t'$.
  Every agent~$a \in t^* \cap t$ prefers every agent~$b\in t^* \cap t'$ to every agent~$c \in  t \setminus t^*$ since the index of $b$ is at least $2\kappa $ positions before $c$ in the order $a_1, \dots, a_n$.
  Similarly, every agent~$a \in t^* \cap t'$ prefers every agent~$b\in t^* \cap t$ to every agent~$c \in  t' \setminus t^*$ since the index of $b$ is at least $2\kappa $ positions before $c$ in the order $a_1, \dots, a_n$.
  Hence, $t^*$ forms a blocking $d$-set consisting solely of agents from $A[{\le i + k -1}]$, a contradiction to the definition of $M_{i-1}^*$.

  Now assume that the dynamic program computed $\tau[i, M]$ to be \myzero.
  We assume for a contradiction that the correct value of $\tau [ i, M] $ is \myone.
  Let $M^*$ be a local matching witnessing that the correct value of $\tau [i, M] $ is \myone, i.e., we have $M^* (a_j) = M (a_j) $ for all $j\in [i, i+ k]$, and $M^*$~does not admit a blocking $d$-set consisting solely of agents from $a_1, \dots, a_{i+k}$.
  Let~$M_{i-1}$ be the restriction of $M^*$ to the $d$-sets containing $a_{i-1}, \dots, a_{i + k -1}$.
  The matching $M^*$ also witnesses that $\tau [ i-1, M_{i-1}] = 
  \myone$, and this value is correctly computed by induction.
  Extending $M_{i-1}$ with~$M (a_{i+k})$ does not lead to a blocking $d$-set containing only agents from $a_{i-1}, \dots, a_{i + k}$, as $M^*$~does not contain such a blocking $d$-set.
  It follows that the algorithm computed $\tau [i, M] $ to be \myone, a contradiction.
\end{proof}

Getting rid of parameter~$d$, we now extend \Cref{thm:fptk+d} to an \FPT-algorithm for the single parameter~$\kappa$.
To do so, we show that if $\kappa $ is much smaller than $d$, then there always exists a stable matching.
Note that we prove this by giving an efficient algorithm for \emph{finding} a stable matching.

\begin{lemma}
  \label{lem:big-d}
  If $4\kappa 2^{4\kappa} \le d$, then there exists a stable matching.
\end{lemma}

\label{proof:lem:big-d}
\begin{proof}
  We apply \Cref{lpo} to the poset $(A, \succeq_{\ML})$, getting an order $a_1, \dots, a_n$ of the agents in~$A$.

  \begin{claim*}
	  There exists a $d$-set $t^*$ such that, for any matching containing $t^*$, no blocking $d$-set contains an agent from $t^*$.
  \end{claim*}

  \begin{claimproof}
	  For each $i\in [d-2\kappa]$, let $t_i \coloneqq \{a_i\} \cup t$, where $t$ is the first $(d-1)$-set in the preferences of~$a_i$.
  By \Cref{lpo}, it holds that $t_i$ contains $a_j$ for each $j\in [d-2\kappa]$, and all agents of $t_i$ are from $A[{\le d + 2\kappa}]$.
  Thus, there are at most $\binom{4\kappa}{2\kappa} \le 2^{4\kappa} - 2$ different classes.
  Since $d - 4\kappa \ge 4\kappa (2^{4\kappa} -1) > 4\kappa (2^{4\kappa} - 2)$, there exists a $d$-set $t^*$ with $t^* = t_i$ for at least~$4\kappa$ agents $a_i$.

  Consider any matching $M$ containing $t^*$, and assume that there is a blocking $d$-set $t$ containing an agent $a^*\in t^*$.
  Let $t \setminus t^* = \{a_{i_1}, \dots, a_{i_r}\}$ with $i_1 < i_2 < \dots < i_r$ and $t^*\setminus t = \{a_{j_1}, \dots, a_{j_r}\}$ with $j_1 < j_2 < \dots j_r$.
  Each agent from $t \setminus t^*$ is contained in $A[{\ge d- 2\kappa + 1}]$ (because $A[{\le d- 2\kappa}] \subseteq t^*$), and $t\setminus t^*$ contains at most $2\kappa $ agents from $A {[d-2\kappa + 1, d+ 2\kappa]}$ (as the other $2\kappa$ agents of this set are contained in $t^*$).
  Thus, we have that $i_p \ge d - 2\kappa + p$ for all~$p \in [r]$, and $i_p \ge d + 2\kappa + (p - 2\kappa) = d+ p$ for all $p > 2\kappa$.
  Every agent $a_i$ with $i\in [d- 2\kappa ]$ such that $t_i = t^*$ is matched to its first choice and thus not contained in $t$.
  Thus, we have $j_{4 \kappa} \le d- 2\kappa$ and therefore, for all $p\in [4\kappa]$, we have that $j_p \le d  - 2\kappa  - (4\kappa -p ) = d - 6\kappa + p < d - 4\kappa + p \le i_p - 2\kappa$.
  For~$p > 4 \kappa$ we have that $j_p \le d + 2\kappa$, while $i_p \ge d + p > j_p + 2\kappa$.
  By \Cref{lpo} it follows that $a_{j_p} \succ_{\ML} a_{i_p}$ for all $p \in [r]$.
  Thus, $a^*$ prefers $t^*$ to $t$, contradicting that $t$ is a blocking $d$-set.
  \end{claimproof}

	\medskip

  From the claim, the lemma follows easily:
  We start with an empty matching $M = \emptyset$ and as long as there are at least $d$ unmatched agents, we successively compute such a $d$-set $t^*$, add $t^*$ to $M$, delete the agents from $t^*$, and repeat.
  The resulting matching is clearly stable, as the agents from the $d$-sets added to $M$ are not part of a blocking $d$-set.
\end{proof}

Finally, it directly follows that \dsmpo is fixed-parameter tractable even when parameterized solely by~$\kappa$.

\begin{theorem}\label{thm:FPT-kappa}
  \dsmpo can be solved in $O(n^2) + (\kappa^5 2^{16\kappa})^{O(\kappa^5 2^{16\kappa})} n$ time, where $\kappa $ is the maximum number of agents an agent is incomparable to, and $n $ is the total number of agents.
\end{theorem}

\begin{proof}
  If $4\kappa 2^{4\kappa} \le d$, then we can safely answer yes by \Cref{lem:big-d}.
  Otherwise we have $d \le 4\kappa 2^{4\kappa}$ and thus, \Cref{thm:fptk+d} yields an algorithm running in $h(\kappa) + O(n^2)$ time with $h (\kappa) = f(\kappa, 4\kappa 2^{4\kappa})$ where $f(\kappa, d) = (\kappa d^4)^{O(\kappa d^4)}$.
\end{proof}

So far, we considered the decision version of \dsmpo.
However, the algorithms from \Cref{thm:fptk+d,lem:big-d} also solve (in the case of \Cref{thm:fptk+d} after the straightforward modification of storing a local matching in the dynamic program instead of $\LogicTRUE$) the search version of \dsmpo, i.e., find a stable matching if one exists.

There is also a natural generalization of \textsc{Stable Marriage} to dimension~$d$, namely \textsc{$d$-dimensional Stable Marriage}, and it is natural to ask whether our algorithm carries over to this setting.
In \textsc{$d$-dimensional Stable Marriage}, the set~$A$ of agents is partitioned into $d$ sets $A^1, \dots, A^d$ of agents, and each agent
of $A^i$ has preferences over all $(d-1)$-sets containing exactly one agent from $A^j$ for all $j\in [d]\setminus \{i\}$.
This problem is also fixed-parameter tractable parameterized by $\kappa + d$:
The master poset of agents can then be decomposed into $d$ master posets of agents, one for each set $A^i$.
Then, one can apply \Cref{lpo} to each of these $d$ master posets to get a strict order for the agents from $A^i = \{a^i_1, \dots, a^i_n\}$.
Similarly to \Cref{lem:small-dist}, one can show that, for any stable matching~$M$ and any $d$-set $\{a^1_{i_1}, \dots, a^d_{i_d}\}$ (w.l.o.g.\ we have $i_j \le i_{j+1}$), it holds that $i_{j+1} \le i_j + O(\kappa d^2)$.
Now one can apply an algorithm similar to \Cref{thm:fptk+d} (sweeping over the sets $A^1, \dots, A^d$ from top to bottom, considering any matching on $k= f(\kappa, d)$ consecutive agents in each set $A^i$) to get an \FPT-algorithm parameterized by~$\kappa + d$.
However, \Cref{lem:big-d} does not seem to generalize to this case:
for $d= 3$, there exists a small instance with $|A_1| = |A_2| = |A_3| =3$ without a stable matching.
``Cloning'' the agents from one of the sets, say $A_3$, an arbitrary number of times will result in an instance of unbounded $d$ but $\kappa = 3$.
It remains therefore unclear whether \Cref{thm:FPT-kappa} generalizes to \textsc{$d$-dimensional Stable Marriage}.

\begin{remark}
Until now, we assumed that the input is encoded naively, i.e., for each agent, its complete preference list is given as part of the input.
However, this list is of length $\Omega(n^{d-1})$, which would result in a total input size of~$O(n^d)$.
Thus, it might be more reasonable to assume that the input is given by an oracle, which can answer queries about the preferences.
In fact, the \FPT-algorithm with the combined parameter $\kappa$ and~$d$ only needs one type of queries, namely given two $(d-1)$-sets $t$ and $t'$ and an agent $a$, the oracle tells whether $a$ prefers $t $ to~$t'$.
Thus, our \FPT-algorithm parameterized only by $\kappa$ also works when only using this query;
however, in the case that $\kappa$ is much smaller than $d$, it cannot compute a stable matching, but only state its existence.
	In order to also compute a stable matching efficiently, the algorithm would also need to be able to query, given an agent~$a$ and a set~$X$ of agents, what is the first $(d-1)$-set in $a$'s preference list not containing an agent from~$X$.
\end{remark}

Having shown that \dsmpo is fixed-parameter tractable for the parameter~$\kappa$, in \Cref{sec:width} we turn to a stronger parameter, the width of the master poset.

\subsubsection{Width of the master poset}
\label{sec:width}

Reducing from \textsc{Multicolored Independent Set} parameterized by solution size, we next show that \tdsrp\ is \Wone-complete parameterized by the width of the master poset.

In this section, at several points it does not matter how the preferences between a set of 2-sets look, as long as the preferences are derived from the poset (which will be described later in this section).
Thus, whenever we describe the preferences of an agent, and these preferences contain a set~$\mathcal{X}$ of 2-sets, then this means that one gets the preferences of the agent through replacing~$\mathcal{X}$ by an arbitrary strict order~$\succ_{\mathcal{X}}$ of the 2-sets contained in~$\mathcal{X}$ such that~$\succ_{\mathcal{X}}$ is derived from the master poset.
Furthermore, we only describe the beginning of the preferences of an agent, followed by~``$\pend$''.
For example, if we write~$\{\{v, w\} : v, w\in \{v_1, v_2,v_3\}\} \pend$ and the master poset is $v_1 \succ_{\ML} v_i$ for $i \in \{2,3,4\}$, then the preferences start with $\{v_1 , v_2 \} \succ \{v_1, v_3\} \succ \{v_2, v_3\}$ or $\{v_1, v_3\} \succ \{v_1, v_2\} \succ \{v_2, v_3\}$, and end with $\{v_1, v_4\} \succ \{v_2, v_4\} \succ \{v_3, v_4\}$ or $\{v_1, v_4\} \succ \{v_3, v_4\} \succ \{v_2, v_4\}$.
Then, all 2-sets not listed before $\pend$ can be added in an arbitrary way obeying the master poset.

\paragraph{The reduction.}

We provide a parameterized reduction from \textsc{Multicolored Independent Set} parameterized by solution size, which is \Wone-hard~\cite{DBLP:series/txcs/DowneyF13,Pietrzak03}.
\defProblemTask{\textsc{Multicolored Independent Set}}
{
A $k$-partite graph $G = (V^1 \disunion V^2 \disunion \dots \disunion V^k, E)$ with $|V^i| = n$ for all $i\in [k]$.}
{
Decide whether $G$ contains an independent set $I$ such that $I\cap V^i \neq \emptyset$ for all $i\in [k]$.}
Let $V^i= \{v_1^i, \dots, v_n^i\}$ for every $i \in [k]$.
The basic idea of the reduction is as follows.
For every $V^i$, we add a \emph{vertex-selection gadget}, encoding which vertex from $V^i$ shall be part of the multicolored independent set.
For every edge $e \in E$, we add an \emph{edge gadget}, ensuring that not both end vertices can be ``selected'' to be part of the multicolored independent set by the corresponding vertex-selection gadget.
Thereby, the edge gadgets ensure that the vertices ``selected'' to be part of a multicolored independent set indeed form an independent set.
Both the vertex-selection gadget and the edge gadget use a third kind of gadget, which we call \emph{cut-off gadget}.
A cut-off gadget for a given agent $v$ and a 2-set~$t$ ensures that agent~$v$ has to be matched at least as good as~$t$ in every stable matching.

We now describe these three gadgets, starting with the cut-off gadget.

\subparagraph{Cut-off gadget.}
\label{par:cut-off}

For many agents, the reduction contains a \emph{cut-off gadget}.
The cut-off gadget for an agent~$a$ ``cuts off'' the preference list of agent $a$ after a specific 2-set $p_a$, and enforces $a$ to be matched to a 2-set $p$ with $p \succeq_a p_a$.

Given an instance $\mathcal{I}$ of \dsmpo with master poset $\succ_{\ML}$, the cut-off gadget for an agent $a$ and a 2-set $p_a$, denoted by $\CO_a$, contains six agents $z_a^1, \dots, z_a^6$.
The basic idea is that deleting all agents but $a$ and $z_a^1, \dots, z_a^5$ results in the instance~$\mathcal{I}_{\operatorname{instable}}$ (see \Cref{sec:no-stable-matching}).
Because no stable matching will match agent~$z_a^r$ to a 2-set which $z_a^r$ prefers to every 2-set $\{x, y\}$ with~$x, y \in \{a, z_a^1, z_a^2, \dots, z_a^6\}$, agent~$a$ has to be matched better than every 2-set~$\{z_a^r, z_a^{r'}\}$ for $r,r'\in [6]$ in every stable matching.

Given a \tdsrp-instance~$\mathcal{I} = (A, (\succ_a)_{a \in A}, d, \succ_{\ML})$, we construct an instance~$\mathcal{I}'$ \emph{arising from $\mathcal{I}$ by adding a cut-off gadget for an agent~$a\in A$ and a 2-set~$p_a$} if we add six agents~$z_a^1, \dots, z_a^6$ to $\mathcal{I}$, and the preferences of the agents in $\mathcal{I}'$ are as follows (see \Cref{ex:cut-off} for an example).
The preferences of any agent $z_a^r$ start with all 2-sets in $\binom{B}{2} \cup (W\times \{z_a^\ell : \ell \in [6]\})$, where $B \coloneqq \{b, z_b^\ell : b \succ_{\ML} a, \ell \in [6]\}$, and are then followed by (ignoring all 2-sets containing the agent itself)
\begin{align*}
    \mathcal{A}:= \{a, z_a^{1}\} & \succ \{a, z_a^{2}\} \succ \{a, z_a^{3}\} \succ \{a, z_a^{5}\}  \succ \{z_a^{1}, z_a^{4}\}\succ \{z_a^{2}, z_a^{3}\}\succ \{a, z_a^{4}\} \succ \{z_a^{1}, z_a^{5}\} \\
    & \succ \{z_a^{2}, z_a^{4}\}\succ \{z_a^{1}, z_a^{3}\}\succ \{z_a^{3}, z_a^{4}\} \succ \{z_a^{1}, z_a^{2}\} \succ \{z_a^{2}, z_a^{5}\} \succ \{z_a^{3}, z_a^{5}\} \succ \{z_a^{4}, z_a^{5}\} \\
    &\succ \{a, z_a^{6}\} \succ \{z_a^{1}, z_a^{6}\}\succ \{z_a^{2}, z_a^{6}\} \succ \{z_a^{3}, z_a^{6}\} \succ \{z_a^{4}, z_a^{6}\} \succ \{z_a^{5}, z_a^{6}\} \pend.
\end{align*}
The preferences of $a$ in $\mathcal{I}'$ arise from the preferences of $a$ in $\mathcal{I}$ by inserting after the 2-set $p_a$ all 2-sets $\{z_a^r, z_a^{r'}\}$ with $r, r' \in [6]$ in the same order as in~$\mathcal{A}$, and appending all other 2-sets containing an agent $z_a^r$ at the end of the preferences of $a$.
The preferences in~$\mathcal{I}'$ of an agent~$b\in A\setminus \{ a\}$ arise from $b$'s preferences in~$\mathcal{I}$ as follows:
If~$a\succ_{\ML} w$, then we add all 2-sets containing an agent from the cut-off gadget at the beginning of their preference, while otherwise we add all 2-sets containing an agent $z_a^r$ at the end of their preferences in an arbitrary order (consistent with the master poset).

\begin{example}
\label{ex:cut-off}
 Consider the \tdsrp-instance~$\mathcal{I}$ with four agents~$a_1, a_2, a_3$, and $a_4$, master poset $a_1 \succ_{\ML} a_2 \succ_{\ML} a_3 \succ_{\ML} a_4$, and the following preferences:
 \begin{align*}
  a_1 &: \{a_2, a_3\} \succ \{a_2, a_4\} \succ \{a_3, a_4\}, \\
  a_2 &: \{a_1, a_3\} \succ \{a_1, a_4\} \succ \{a_3, a_4\}, \\
  a_3 &: \{a_1, a_2\} \succ \{a_1, a_4\} \succ \{a_2, a_4\}, \\
  a_4 &: \{a_1, a_2\} \succ \{a_1, a_3\} \succ \{a_2, a_3\}.
 \end{align*}
 The instance~$\mathcal{I}'$ arising from $\mathcal{I}$ by adding a cut-off gadget for $a := a_2$ and $p_a := \{a_1, a_4\}$ has the following preferences (where for $z_a^r$ all 2-sets containing $z_a^r$ shall be deleted):
 \begin{align*}
  a_1 &: \{a_2, a_3\} \succ \{a_2, a_4\} \succ \{a_3, a_4\}... \pend,\\
  a_2 &: \{a_1, a_3\} \succ \{a_1, a_4\} \succ \{z_a^{1}, z_a^{4}\}\succ \{z_a^{2}, z_a^{3}\} \succ \{z_a^{1}, z_a^{5}\} \succ \{z_a^{2}, z_a^{4}\}\succ \{z_a^{1}, z_a^{3}\}\succ \{z_a^{3}, z_a^{4}\}, \\
  & \succ \{z_a^{1}, z_a^{2}\} \succ \{z_a^{2}, z_a^{5}\} \succ \{z_a^{3}, z_a^{5}\} \succ \{z_a^{4}, z_a^{5}\} \succ \{a, z_a^{6}\} \succ \{z_a^{1}, z_a^{6}\}\succ \{z_a^{2}, z_a^{6}\} \succ \{z_a^{3}, z_a^{6}\}, \\
  & \succ \{z_a^{4}, z_a^{6}\} \succ \{z_a^{5}, z_a^{6}\} \succ \{a_3, a_4\},\\
  a_3 &: \{\{z_a^r, z_a^{r'}\}: r, r' \in [6], r\neq r'\} \succ \{a_1, a_2\} \succ \{a_1, a_4\} \succ \{a_2, a_4\} \pend,\\
  a_4 &: \{\{z_a^r, z_a^{r'}\}: r, r' \in [6], r\neq r'\} \succ \{a_1, a_2\} \succ \{a_1, a_3\} \succ \{a_2, a_3\} \pend,\\
  z_a^r & : \{a, z_a^{1}\} \succ \{a, z_a^{2}\} \succ \{a, z_a^{3}\} \succ \{a, z_a^{5}\}  \succ \{z_a^{1}, z_a^{4}\}\succ \{z_a^{2}, z_a^{3}\}\succ \{a, z_a^{4}\} \succ \{z_a^{1}, z_a^{5}\} \\
    & \succ \{z_a^{2}, z_a^{4}\}\succ \{z_a^{1}, z_a^{3}\}\succ \{z_a^{3}, z_a^{4}\} \succ \{z_a^{1}, z_a^{2}\} \succ \{z_a^{2}, z_a^{5}\} \succ \{z_a^{3}, z_a^{5}\} \succ \{z_a^{4}, z_a^{5}\} \\
    &\succ \{a, z_a^{6}\} \succ \{z_a^{1}, z_a^{6}\}\succ \{z_a^{2}, z_a^{6}\} \succ \{z_a^{3}, z_a^{6}\} \succ \{z_a^{4}, z_a^{6}\} \succ \{z_a^{5}, z_a^{6}\} \pend.
 \end{align*}

\end{example}

We now show that adding for every agent $a$ a cut-off gadget for~$a$ and a 2-set~$p_a$ enforces that $a$ is matched at least as good as~$p_a$ in every stable matching.

\begin{lemma}\label{lem:cog}
  Let $\mathcal{I}_0$ be an \dsmpo-instance with set~$A$ of agents, and for each $a\in A$, let $p_a$ be a 2-set not containing $a$.
  Let $A = \{a_1, \dots, a_n\}$ with $a_i \succ a_j$ or $a_i \perp a_j$ for all $i < j$.
  Let~$\mathcal{I}_s$ arise from~$\mathcal{I}_{s-1}$ by adding a cut-off gadget for agent $a_s$ and 2-set $p_{a_s}$.
  Let $\mathcal{M}_s$ be the set of stable matchings $M$ in $\mathcal{I}_0$ with $M(a) \succ_a p_a$ for all $a\in \{a_1, \dots, a_s\}$.
  
  Then, for any matching~$M \in \mathcal{M}_s$, the matching $M_s\coloneqq M \cup \{\{z_a^1, z_a^2, z_a^3\}, \{z_a^4, z_a^5, z_a^6\}: a\in \{a_1, \dots, a_s\}\}$ is stable in $\mathcal{I}_s$.
  Furthermore, for every matching~$M'$ which is stable in $\mathcal{I}_s$, we have that $M' \cap \binom{A}{3}\in \mathcal{M}_s$.
\end{lemma}

\begin{proof}
  We prove the lemma by induction on $s$.
  For $s =0$, there is nothing to show.
  Fix $s >0$.

  First, we show that for any stable matching $M\in \mathcal{I}_0$ with $M(a) \succeq_a p_a$ for all $a\in \{a_1, \dots, a_s\}$, it holds that $M_s\coloneqq M \cup \{\{z_b^1, z_b^2, z_b^3\}, \{z_b^4, z_b^5, z_b^6\}: b\in \{a_1, \dots, a_s\}\}$ is a stable matching in~$\mathcal{I}_s$.
  Let $a\coloneqq a_s$.
  No blocking 3-set consists solely of agents from $\{z_a^r : r\in [6]\}$;
  this can be seen by checking all 20 such 3-sets.
  Any blocking 3-set needs to contain at least one agent~$z_a^r$ (else it would already be a blocking 3-set in $\mathcal{I}_{s-1}$, contradicting the induction hypothesis).
  Agent $a$ is not part of such a blocking 3-set, as it ranks all 2-sets containing an agent $z_a^r$ after $M(a) \succeq_a p_a$.
  All other 2-sets preferred by agent~$z_a^r$ contain an agent $a'$ or $z_{a'}^q$ with $a' \succ_{\ML} a$.
  However, $a'$ is matched to a 2-set~$M(a')$ it prefers to~$p_{a'}$, and therefore $a'$~does not prefer any 2-set containing~$z_a^r$ to $M(a')$. Thus, $a'$~is not part of a blocking 3-set.
  Consequently, the blocking 3-set contains solely agents from cut-off gadgets for agents $a'$ with $a' \succ_{\ML} a$.
  However, no agent~$z_{a_i}^r$ prefers any pair~$\{z_{a_{i'}}^{i'}, z_{a_{i''}}^{i''}\}$ with $i' \ge i$ and $i'' > i$. 
  Thus, the blocking 3-set cannot contain agents from different cut-off gadgets, a contradiction.

	To see the second part of the lemma, let $M_s$ be a stable matching in $\mathcal{I}_s$.
  Let $M$ be the matching arising from~$M_s$ by deleting all 3-sets containing an agent $z_{a_i}^r$.
  We need to show that $M$ is stable and $M(a) \succeq_a p_a$  for all $a\in \{a_1, \dots, a_s\}$.
  If there exists some~$a\in \{a_1, \dots, a_s\}$ such that $M(a) \succeq_a p$ does not hold, then let $a_i$ be the agent of minimal index for which $p_{a_i} \succ M(a_i)$.
  Then for all $j < i$, matching~$M$ matches the agents from $\{z_{a_j}^r : r\in [6]\}$ in two 3-sets:
  These agents are not matched to $a_{j'}$ with $j' \le j$ since $a_{j'}$ is matched at least as good as $p_{a_{j'}}$, and if $z_{a_j}^r$ is not matched solely to agents from~$\CO_{a_j}$, then it forms a blocking 3-set together with any two other agents of $\CO_{a_j}$ which are not matched solely to agents from $\CO_{a_j}$.
  Therefore, no agent from $\CO_{a_i}$ including $a_i$ is matched to a 2-set which it prefers to any 2-set consisting of agents from $\CO_{a_i}$.
  By \Cref{linstability}, it follows that there is a blocking 3-set in $\{a_i\} \cup \{z_{a_i}^r: r\in [5]\}$, contradicting the stability of~$M$.
  Thus, $M(a) \succeq_a p$ holds for all $a\in \{a_1, \dots, a_s\}$.
  
  It remains to show that $M$ is stable in $\mathcal{I}_0$.
	First note that the above argument shows that $M_s$ matches every agent~$z_{a_j}^r$ from cut-off gadget $\CO_{a_j}$ to a 2-set consisting solely of agents from~$\CO_{a_j}$.
  Thus, every agent $a\in A$ is matched the same in $M_s$ and $M$.
  Hence, any blocking 3-set for $M$ is also a blocking 3-set for $M_s$, and thus, $M$ is stable.
\end{proof}

Having shown that the cut-off gadgets have the desired effect, we now show how we can model the selection of one vertex being part of the independent set.

\subparagraph{Vertex-Selection Gadget.}

A vertex-selection gadget has $6n$ agents $a_p$, $b_p$ ($p \in [n]$), $c_q$, $\barc_q$ ($q\in [n- 1]$), $d_r$, and $\bard_r$ ($r\in [n+1]$).
The intuitive idea is the following.
The agents~$a_p$ and $b_p$ want to be matched to 2-sets $\{c_q, \barc_{q'}\}$.
As $c_q$ prefers $\{a_p, \barc_{q'}\}$ to $\{b_{p'}, \barc_{q'}\}$ while~$\barc_{q'}$ prefers~$\{b_{p'}, c_q\}$ to~$\{a_p, c_q\}$, we can match the $n-1$ sets of size two of the form $\{c_q, \barc_q\}$ to the agents $\{a_p : p < \ell\} \cup \{b_{p} : p < n + 1 - \ell\}$ for any $\ell \in [n]$, corresponding to selecting the vertex~$v_\ell$ to be part of the independent set.
The agents $\{a_p : p \ge \ell\} \cup \{b_p : p \ge n + 1 -\ell\}$ are then matched to the 2-sets $\{ d_r, \bard_r \}$, and can form blocking 3-sets with the edge gadgets (which are described later).

Formally, the preferences look as follows.
\begin{align*}
  a_p  \colon &\{\{c_q, \barc_{q'}\} : q, q'\in [n  - 1]\} \succ \{\{d_r, \bard_{r'}\} : r, r' \in [n+ 1]\}\succ \CO_{a_p},\\
  b_p  \colon &\{\{c_q, \barc_{q'}\} : q, q'\in [n  - 1]\} \succ \{\{d_r, \bard_{r'}\} : r, r' \in [n+ 1]\}\succ \CO_{b_p},\\
  c_q \colon &\{\{a_1, \barc_{q'}\} : q'\in [n-1]\} \succ \{\{a_2, \barc_{q'}\} : q'\in [n-1]\} \succ \dots \succ \{\{a_n, \barc_{q'}\} : q'\in [n-1]\} \\
  & \succ \{\{b_1, \barc_{q'}\} : q'\in [n-1]\} \succ \{\{b_2, \barc_{q'}\} : q' \in [n-1]\} \succ \dots \succ \{\{b_k, \barc_{q'}\} : q' \in [n-1]\} \\
  & \succ\CO_{c_q},\\
  \barc_q \colon & \{\{b_1, c_{q'}\} : q'\in [n-1]\} \succ \{\{b_2, c_{q'}\} : q'\in [n-1]\} \succ \dots \succ \{\{b_n, c_{q'}\} : q'\in [ n-1]\}\\
  & \succ \{\{a_1, c_{q'}\} : q'\in [n-1]\} \succ \{\{a_2, c_{q'}\} : q'\in [n-1]\} \succ \dots \succ \{\{a_n, c_{q'}\} : q'\in [n-1]\} \\
  & \succ \CO_{\barc_q},\\
  d_r \colon & \{\{a_1, \bard_{r'}\} : r'\in [n + 1]\} \succ \{\{a_2, \bard_{r'}\} : {r'}\in [n + 1]\} \succ \dots \succ \{\{a_n, \bard_{r'}\} : {r'}\in [n + 1]\} \\
  & \succ \{\{b_1, \bard_{r'}\} : {r'}\in [n+1]\} \succ \{\{b_2, \bard_{r'}\} : {r'}\in [n+1]\} \succ \dots \succ \{\{b_n, \bard_{r'}\} : {r'}\in [n+1]\} \\
  & \succ \CO_{d_r},\\
  \bard_r  \colon & \{\{b_1, d_{r'}\} : {r'}\in [n-1]\} \succ \{\{b_2, d_{r'}\} : {r'}\in [n-1]\} \succ \dots \succ \{\{b_n, d_{r'}\} : {r'}\in [n-1]\} \\
  & \succ \{\{a_1, d_{r'}\} : {r'}\in [n+1]\} \succ \{\{a_2, d_{r'}\} : {r'}\in [n+1]\} \succ \dots \succ \{\{a_n, d_{r'}\} : {r'}\in [n+1]\} \\
  & \succ \CO_{\bard_r}.
\end{align*}
In the master poset, we have $x_p \succ x_q $ for $p< q$ for all $x\in \{a, b, c, \bar{c}, d, \bar{d}\}$ as well as $x_i \perp y_j$ for all $x\neq y$.
See \Cref{fig:vsg-width} for an example of the 3-sets $\{v_1, v_2, v_3\}$ such that $\{v_1, v_2, v_3\} \setminus \{v_p\}$ is before the cut-off gadget of $v_p$ in the preferences of $v_p$ for all $p\in [3]$.

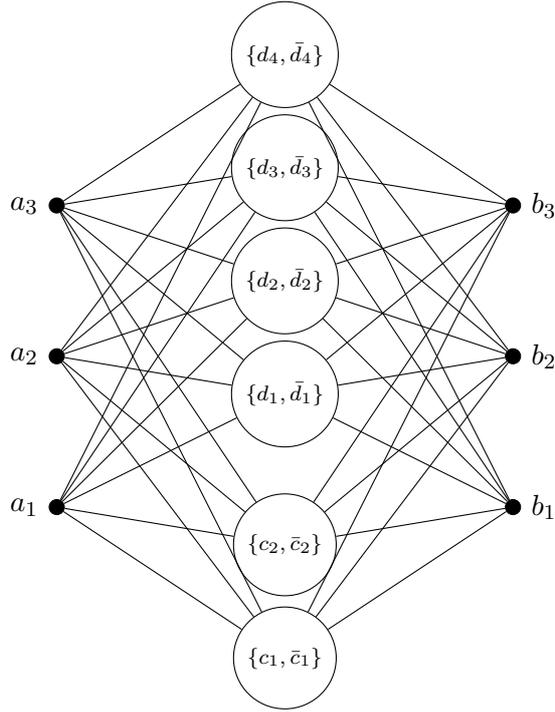
\begin{figure}
 \begin{center}
  \begin{tikzpicture}
      \node[vertex, label=180:$a_1$] (a1) at (-3,0) {};
      \node[vertex, label=180:$a_2$] (a2) at (-3,2) {};
      \node[vertex, label=180:$a_3$] (a3) at (-3,4) {};
      \node[vertex, label=0:$b_1$] (b1) at (3,0) {};
      \node[vertex, label=0:$b_2$] (b2) at (3,2) {};
      \node[vertex, label=0:$b_3$] (b3) at (3,4) {};

      \node[circle, draw] (cc1) at ($0.5*(a1) + 0.5 *(b1) + (0, -2)$) {\scriptsize{$\{c_1, \barc_1\}$}};
      \node[circle, draw] (cc2) at ($(cc1) + (0, 1.5)$) {\scriptsize{$\{c_2, \barc_2\}$}};
      \node[circle, draw] (cd1) at ($(cc2) + (0, 2)$) {\scriptsize{$\{d_1, \bard_1\}$}};
      \node[circle, draw] (cd2) at ($(cd1) + (0, 1.5)$) {\scriptsize{$\{d_2, \bard_2\}$}};
      \node[circle, draw] (cd3) at ($(cd2) + (0, 1.5)$) {\scriptsize{$\{d_3, \bard_3\}$}};
      \node[circle, draw] (cd4) at ($(cd3) + (0, 1.5)$) {\scriptsize{$\{d_4, \bard_4\}$}};

      \draw (a1) edge (cc1) edge (cc2) edge (cd1) edge (cd2) edge (cd3) edge (cd4);
      \draw (a2) edge (cc1) edge (cc2) edge (cd1) edge (cd2) edge (cd3) edge (cd4);
      \draw (a3) edge (cc1) edge (cc2) edge (cd1) edge (cd2) edge (cd3) edge (cd4);
      \draw (b1) edge (cc1) edge (cc2) edge (cd1) edge (cd2) edge (cd3) edge (cd4);
      \draw (b2) edge (cc1) edge (cc2) edge (cd1) edge (cd2) edge (cd3) edge (cd4);
      \draw (b3) edge (cc1) edge (cc2) edge (cd1) edge (cd2) edge (cd3) edge (cd4);
  \end{tikzpicture}

 \end{center}
 \caption{The acceptable 3-sets (i.e., 3-sets $\{x_1, x_2, x_3\}$ such that $\{x_1, x_2, x_3\} \setminus \{x_i\}$ is before the cut-off gadget of $x_i$ in the preferences of $x_i$ for all $i\in [3]$) of a vertex-selection gadget.
 For the sake of readability, we ignored 3-sets containing the agents $c_i$ and $\barc_j$ or $d_i$ and $\bard_j$ for $i \neq j$.
 Each edge corresponds to a 3-set containing the one endpoint $a_i$ or $b_i$ and the two vertices contained in the circle of the other endpoint of the edge.
 For example, the edge between $a_1$ and $\{c_1, \barc_1\}$ indicates that $\{a_1, c_1, \barc_1\}$ is an acceptable 3-set.}
 \label{fig:vsg-width}
\end{figure}

\begin{lemma}\label{lem:vsg}
  Let $\mathcal{I}$ be a \dsmpo-instance containing a vertex-selection gadget such that every agent $v$ in the vertex-selection gadget except for $\{a_i, b_i : i \in [3n'+1]\}$ prefers all 2-sets consisting of two agents in the vertex-selection gadget to a 2-set containing an agent outside the vertex-selection gadget.

  For every stable matching~$M$ in~$\mathcal{I}$, there exists some $p^*\in [n]$ such that $M(a_p) \in \{\{c_q, \barc_{q'}\} : q, q'\in [n - 1]\}$ for all $p < p^*$, $M(a_p) \in \{\{ d_r, \bard_{r'}\} : r, r'\in [n+1]\}$ for $p \ge p^*$, $M(b_p) \in \{\{ c_{q}, \barc_{q'}\} : q, q'\in [n-1]\}$ for $p \le n - p^*$, and $M(b_p) \in \{\{d_{r}, \bard_{r'}\} : r, r'\in [n+1]\}$ for $i> n - p^*$.
\end{lemma}

\begin{proof}
  Consider any stable matching $M$.
  Due to the cut-off gadgets, any agent $c_q$ is matched to a 2-set $\{a_p, \barc_{q'}\}$ or $\{b_p, \barc_{q'}\}$ for some $p \in [n]$ and $q' \in [n-1]$, and any agent $d_r$ is matched to a 2-set $\{a_p, \bard_{r'}\}$ or $\{b_p, \bard_{r'}\}$ for some $p\in [n]$ and $r' \in [n+1]$.
  Since $|\{a_p, b_p : p \in [n]\}| = 2n = |\{c_q : q\in [n-1]\}| + |\{d_r : r\in [n+1]\}|$, it follows that any agent~$a_p$ or~$b_p$ is matched to a 2-set~$\{c_q, \barc_{q'}\} $ or $\{d_r, \bard_{r'}\}$.
  
  Assume for a contradiction that there exists some $p\in [n]$ such that $\{a_p,  c_q, \barc_{q'}\}\notin M$ for all $q, q' \in [n-1]$ but $\{a_{p+1}, c_s, \barc_{s'}\}\in M$ for some $s,s' \in [n-1]$ (the case that there exists such $b_p$ and $b_{p+1}$ is symmetric).
  Then $\{a_p, c_s, \barc_{s'}\}$ is a blocking 3-set, contradicting the stability of~$M$.
  \end{proof}

  We now show that it is indeed possible to select an arbitrary vertex in any vertex-selection gadget.

  \begin{lemma}\label{lem:vsgBackwards}
  Let $\mathcal{I}$ be an \dsmpo-instance containing a vertex-selection gadget.

  For any $p^*\in [n]$, the matching $M_{p^*} := \{\{a_p, c_p, \barc_{n - p^* + p}\}: p  < p^*\} \cup \{\{a_p, d_{p - p^* + 1}, \bard_{p+1}\}: p \ge p^*\}\cup \{\{b_p, c_{p + p^* -1}, \barc_{p}\}: p \le n  -p^*  \} \cup \{ \{b_p, d_{p+1}, \bard_{p - (n- p^*)}\} : p > n - p^*\}$ contains no blocking 3-set solely consisting from agents in the vertex-selection gadget.
  \end{lemma}

  \begin{proof}
  Let $p^*\in [n]$ and assume that $M_{p^*}$ contains a blocking 3-set~$t$ inside the vertex-selection gadget.
  Then $t$ contains two agents $c_q$ and $\barc_{q'}$ or $d_r $ and $\bard_{r'}$, and an agent $a_p$ or $b_p$.

  \bfseries Case 1: \mdseries $ t = \{a_p, c_q, \barc_{q'}\}$.
  If $p > p^*$, then $\barc_{q'}$ does not prefer $t \setminus \{\barc_{q'}\}$ to $M (\barc_{q'}) $, so assume $p \le p^*$, implying that $\{a_p, c_p, \barc_{n-p^*+ i}\} \in M_{p^*}$.
  If $p\le \min\{q, q' + p^* -n\}$, then $a_p$ does not prefer $t \setminus\{a_p\}$ to~$M_{p^*} (a_p)$.
  If $q \le \min \{p, q' + p^* - n\}$, then $\{a_q, c_q, \barc_{n-p^*+ q}\} \in M_{p^*}$ and $c_q$ does not prefer $t\setminus \{c_q\}$ to $M_{q^*} (c_q)$.
  If $q' \le \min \{p, q\} + n - p^*$, then $\barc_{q'}$ does not prefer $t\setminus \{\barc_{q'}\}$ to~$M_{p^*} (\barc_{q'})$.

  \bfseries Case 2: \mdseries $ t = \{a_p, d_r, \bard_{r'}\}$.
  If $p\le \min\{r + p^* - 1, {r'}- 1\}$, then $a_p$ does not prefer $t \setminus \{a_p\}$ to~$M_{p^*} (a_p)$.
  If $r \le \min \{p + 1, {r'}\} - p^*$, then $\{a_{r+p^*  -1}, d_r, d_{r + p^*} \}\in M_{p^*}$ and $d_r$ does not prefer~$t\setminus\{d_r\}$ to $M_{p^*} (d_r)$.
  If~${r'} \le \min \{p, r + p^* - 1\} + 1$, then $\bard_{r'}$ does not prefer $t \setminus \{\bard_{r'}\}$ to~$M_{p^*} (\bard_{r'})$.

  \bfseries Case 3: \mdseries $ t = \{b_p, c_q, \barc_{q'}\}$.
  If $p > n -  p^*$, then $c_q$ does not prefer $t \setminus\{c_q\}$ to $M_{p^*} (c_q)$, so assume $p \le n - p^*$, implying that $\{b_p, c_{p+ p^* -1}, \barc_p\} \in M_{p^*}$.
  If $p\le \min\{q - p^* + 1, q'\}$, then $b_p$ does not prefer $t \setminus\{b_p\} $ to~$M_{p^*} (b_p)$.
  If $q \le \min \{p, q'\} + p^* - 1$, then $c_q$ does not prefer $t \setminus \{c_q\}$ to $M_{p^*} (c_q)$.
  If $q' \le \min \{p, q - p^* +1\}$, then $\barc_{q'}$ does not prefer $t \setminus \{\barc_{q'}\}$ to $M_{p^*} (\barc_{q'})$.

  \bfseries Case 4: \mdseries $ t = \{b_p, d_r, \bard_{r'}\}$.
  If $p\le \min\{r-  1, {r'} + n - p^*\}$, then $b_p$ does not prefer $t \setminus \{b_p\}$ to $M_{p^*} (b_p)$.
  If~$r \le \min \{p, {r'} + n - p^*\} + 1$, then $d_r$ does not prefer $t \setminus \{d_r\}$ to $M_{p^*} (d_r)$.
  If $r' \le \min \{p, r - 1\} - (n - p^*)$, then $\{b_{{r'} + n - p^*}, d_{{r'}+n-p^* + 1}, \bard_{r'} \}\in M_{p^*}$ and $\bard_{r'}$ does not prefer~$t\setminus \{\bard_{r'}\}$ to~$M_{p^*} (\bard_{r'})$.

  In each of the four cases, we have a contradiction, finishing the proof of the lemma.
\end{proof}

The parameterized reduction will create a vertex-selection gadget~$S^i$ for each set~$V^i$, $i\in [k]$.
We refer to agent $a_p$ (respectively $b_p, c_q, \barc_q, d_r$, or $\bard_r$) from the $i$-th vertex-selection gadget via~$a_p^i$ (respectively $b_p^i, c_q^i, \barc_q^i, d_r^i$, or $\bard_r^i$).

We say that a matching~$M$ \emph{selects} the vertex $v_{p^*}^i$ in vertex-selection gadget $S^i$ if, for all~$p < p^*$, we have $M (a_p) = \{c_q, \barc_{q'}\}$ for some $q, q' \in [n-1]$, and for all $p \ge p^*$, we have $M(a_p) = \{d_r, \bard_{r'}\}$ for some $r, {r'} \in [n+1]$.

We now turn to the edge gadgets, which model the edges of~$G$.

\subparagraph{Edge gadget.}

Fix a pair $(i,j) \in [k]^2$ with $i < j$, and let $E^{i,j}\coloneqq \{ \{v, w\}\in E(G): v\in V^i, w\in V^j\}$ be the set of edges between $V^i$ and $V^j$.
For each edge $e = \{v^i_{r_i}, v^j_{r_j}\}\in E^{i,j}$, our reduction contains an edge gadget between $S^i$ and $S^j$, containing the 15 agents $h^{e, i}_a$, $h^{e, i}_b$, $h^{e, j}_a$, $h^{e,j}_b$, $g^e_1$, $g^e_2$, $g^e_3$, $\barg^e_1$, $\barg^e_2$, $\barg^e_3$, $f^e$, $\barf^e$, $\alpha^e_1$, $\alpha^e_2$, and $\alpha^e_3$.

The intuitive function of the edge gadget is the following.
By \Cref{lem:vsg}, in any stable matching~$M$ there exists some $i^*$ such that $S^i$ selects $v^i_{i^*}$, i.e.,~$a_\ell^i$~is matched to a pair~$\{d_p, \bard_q\}$ if and only if $\ell \ge i^*$.
Given an edge~$e$, for any $\ell \ge i^*$, agent~$a_{\ell}^i$ prefers being matched to the edge gadget (more specifically, to the 2-set~$\{h_a^{e, i}, \alpha_1^{e}\}$) to being matched to $M$.
Similarly, $b_\ell$ for $\ell > n-i^*$ prefers to be matched to the edge gadget, namely to the 2-set~$\{h_b^{e, i}, \alpha_1^{e}\}$ to being matched to $M$.
The agent $h_a^{e, i}$ [$h_b^{e, i}$] prefers the 2-set $\{a^i_{i^*}, \alpha_1^{e}\}$ [$\{b_{n + 1- i^*}^i, \alpha_1^{e}\}$] to the 2-set~$\{f^e, \barf^e\}$ if $r_i \le i^*$ [$r_i \ge i^*$].
Thus, if $v^i_{r_i}= v^i_{i^*}$, i.e., if the vertex selected by $S^i$ is an endpoint of~$e$, then both $h_a^{e, i}$ and $h_b^{e,i}$ cannot be matched to $\{f^e, \barf^e\}$.
The edge gadget is now designed in such a way that an arbitrary of these vertices~$h_x^{e, y}$ ($x\in \{a, b\}$ and $y\in \{i,j\}$) has to be matched to $\{f^e, \barf^e\}$, which is possible if and only if $e$ is not an edge between the two vertices selected by the vertex-selection gadgets, i.e., $e \neq \{v^i_{i^*}, v^j_{j^*}\}$.

We now give a formal description of the edge gadget.
We fix an arbitrary order $\succ^{i,j}$ of the edges from $E^{i,j }$, and denote by $\betterEdgese$ the set of edges $e'\in E^{i, j}$ with $e' \succeq^{i,j} e$.
In the master poset, an agent $x^e$ is before $y^{e'}$ for $x\in \{h_a^{\cdot, i}, h_a^{\cdot, j}, h_b^{\cdot, i}, h_b^{\cdot, j}, g_1, g_2, g_3, \barg_1, \barg_2, \allowbreak \barg_3,\allowbreak f, \barf, \alpha_1, \alpha_2, \alpha_3\}$ if and only if $x= y$ and $e $ is before $e'$ in the order of edges from $E^{i,j}$.
Let $e = \{v_{r_i}^i, v_{r_j}^j\}$.
The preferences of the agents look as follows.
\begin{align*}
 h_a^{e,i} \colon &\{ \{ g_p^{e'}, \barg^{e''}_q\}: p, q\in [3], e', e''\in \betterEdgese\} \succ \{\{a_p^i, \alpha_1^{e'}\}: p \le r_i, e'\in \betterEdgese\} \\
 &\succ \{f^{e'}, \barf^{e''} : e', e''\in \betterEdgese\} \succ \CO_{h_a^{e, i}},\\
 h_a^{e,j} \colon  & \{ \{ g_p^{e'}, \barg^{e''}_q\}: p, q\in [3], e', e''\in \betterEdgese\} \succ \{\{a_p^j, \alpha_1^{e'}\}: p \le r_j, e'\in \betterEdgese\} \\
 & \succ \{f^{e'}, \barf^{e''}: e' , e''\in \betterEdgese\} \succ \CO_{h_a^{e, j}},\\
 h_b^{e,i} \colon  & \{ \{ g^{e'}_p, \barg^{e''}_q\}: p, q\in [3], e', e''\in \betterEdgese\} \succ \{\{b_p^i, \alpha_1^{e'}\} : p \le n + 1 - r_i, e'\in \betterEdgese\} \\
 & \succ \{f^{e'}, \barf^{e''}: e'\in \betterEdgese\} \succ \CO_{h_b^{e, i}},\\
 h_b^{e,j} \colon  & \{ \{ g^{e'}_p, \barg^{e''}_q\}: p, q\in [3], e', e''\in \betterEdgese\} \succ \{\{b_p^j, \alpha_1^{e'}\} : p \le n + 1 - r_j, e'\in \betterEdgese\} \\
 & \succ \{f^{e'}, \barf^{e''} : e', e''\in \betterEdgese\} \succ \CO_{h_b^{e, j}},\\
 g^e_{p} \colon  & \{\{h_a^{{e'}, i}, \barg^{e''}_p\}: e', e''\in \betterEdgese\} \succ \{\{h_b^{{e'}, i}, \barg^{e''}_p\} : e', e''\in \betterEdgese\} \\
 & \succ \{\{h_a^{{e'}, j}, \barg^{e''}_p\}: e', e''\in \betterEdgese\} \succ \{\{h_b^{{e'}, j}, \barg^{e''}_p\} : e', e''\in \betterEdgese\} \succ \CO_{g^e_{p}}, \\
 \barg^e_{p} \colon  & \{\{h_b^{{e'}, j}, g^{e''}_p\} :  e', e''\in \betterEdgese\} \succ \{\{h_a^{{e'}, j}, g^{e''}_p\} :  e', e''\in \betterEdgese\}\\
 & \succ \{\{h_b^{{e'}, i}, g^{e''}_p\} :  e', e''\in \betterEdgese\} \succ \{\{h_a^{{e'}, j}, g^{e''}_p\} :  e', e''\in \betterEdgese\} \succ \CO_{\barg^e_{p}}, \\
 f^e \colon  & \{h_a^{{e'}, i}, \barf^{e''} :  e', e''\in \betterEdgese\} \succ \{h_b^{{e'}, i}, \barf^{e''} :  e', e''\in \betterEdgese\} \\
 & \succ \{h_a^{{e'}, j}, \barf^{e''} :  e', e''\in \betterEdgese\} \succ \{h_b^{{e'}, j}, \barf^{e''} :  e', e''\in \betterEdgese\} \succ \CO_{f^e}, \\
 \barf^e \colon  &  \{h_b^{{e'}, j}, f^{e''} :  e', e''\in \betterEdgese\} \succ \{h_a^{{e'}, j}, f^{e''} :  e', e''\in \betterEdgese\} \\
 & \succ \{h_b^{{e'}, i}, f^{e''} :  e', e''\in \betterEdgese\} \succ \{h_a^{{e'}, j}, f^{e''} :  e', e''\in \betterEdgese\} \succ \CO_{\barf^e}, \\
 \alpha_1^e \colon  & \{\{a_p^i, h^{{e'}, i}_a\}, \{b_p^i, h^{{e'}, i}_a\}, \{a_p^j, h^{{e'}, j}_a\}, \{b_p^j, h^{{e'}, j}_a\} : p\in [n] :  e' \in \betterEdgese\} \\
 & \succ \{\{\alpha_2^{e'}, \alpha_3^{e''}\} :  e', e''\in \betterEdgese\} \succ \CO_{\alpha_1^e},\\
 \alpha_2^e \colon  & \{ \{\alpha^{e'}_1, \alpha_3^{e''} \}:  e', e''\in \betterEdgese\} \succ \CO_{\alpha_2^e},\\
 \alpha_3^e \colon  & \{\{\alpha_1^{e'}, \alpha_2^{e''} \}:  e', e''\in \betterEdgese\} \succ \CO_{\alpha_3^e}.
\end{align*}
\Cref{fig:eg-width} visualizes the 3-sets~$t=\{x_1, x_2, x_3\}$ such that $x_i$ prefers $t\setminus \{x_i\}$ to the 2-sets containing agents from its cut-off gadget for all $i\in [3]$.

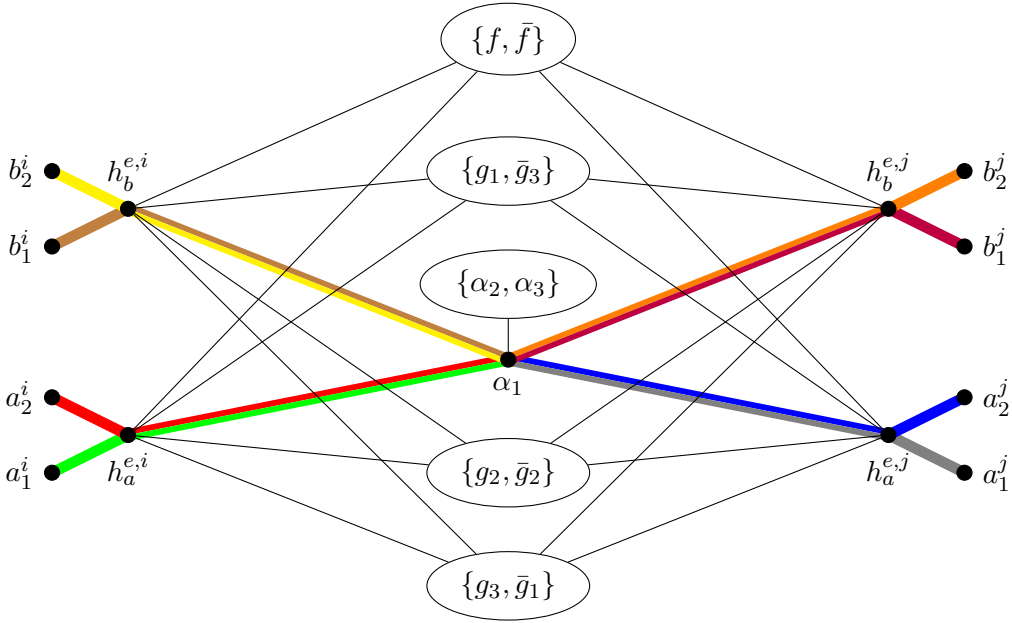
\begin{figure}
 \begin{center}
  \begin{tikzpicture}
      \node[vertex, label=180:$a_1^i$] (a1) at (-6,0) {};
      \node[vertex, label=180:$a_2^i$] (a2) at ($(a1) + (0, 1)$) {};
      \node[vertex, label=180:$b_1^i$] (b1) at ($(a1) + (0, 3)$) {};
      \node[vertex, label=180:$b_2^i$] (b2) at ($(a1) + (0, 4)$) {};
      \node[vertex, label=0:$a_1^j$] (a1j) at ($(a1) + (12, 0)$) {};
      \node[vertex, label=0:$a_2^j$] (a2j) at ($(a2) + (12, 0)$) {};
      \node[vertex, label=0:$b_1^j$] (b1j) at ($(b1) + (12, 0)$) {};
      \node[vertex, label=0:$b_2^j$] (b2j) at ($(b2) + (12, 0)$) {};

      \node[vertex, label=90:$h_b^{e, i}$] (dbi) at ($0.5*(b1) + 0.5*(b2) + (1, 0)$) {};
      \node[vertex, label=270:$h_a^{e, i}$] (dai) at ($0.5*(a1) + 0.5*(a2) + (1, 0)$) {};
      \node[vertex, label=90:$h_b^{e, j}$] (dbj) at ($0.5*(b1j) + 0.5*(b2j) + (-1, 0)$) {};
      \node[vertex, label=270:$h_a^{e, j}$] (daj) at ($0.5*(a1j) + 0.5*(a2j) + (-1, 0)$) {};

      \node[ellipse, draw] (cc1) at ($0.5*(dai) + 0.5*(daj) + (0, -2.)$) {$\{g_3, \barg_1\}$};
      \node[ellipse, draw] (cc2)  at ($(cc1) + (0, 1.5)$) {$\{g_2, \barg_2\}$};
      \node[ellipse, draw] (cd1) at ($(cc2) + (0, 4)$) {$\{g_1, \barg_3\}$};
      \node[ ellipse, draw] (cd2) at ($(cd1) + (0, 1.75)$) {$\{f, \barf\}$};

      \node[vertex, label=270:$\alpha_1$] (al1) at ($(cc2) + (0, 1.5)$) {};
      \node[ellipse, draw] (al2) at ($(al1) + (-0., 1)$) {$\{\alpha_2, \alpha_3\}$};
      \draw (al2) -- (al1);

      \draw (dai) edge (cc1) edge (cc2) edge (cd1) edge (cd2);
      \draw (dbi) edge (cc1) edge (cc2) edge (cd1) edge (cd2);
      \draw (daj) edge (cc1) edge (cc2) edge (cd1) edge (cd2);
      \draw (dbj) edge (cc1) edge (cc2) edge (cd1) edge (cd2);

        \begin{scope}[on background layer]
          \newcommand{\colorBetweenTwoNodes}[3]{
            \fill[#1] ($(#2) + (0, .08)$) to ($(#2) - (0, .08)$) to ($(#3) - (0,.08)$) to ($(#3) + (0,.08)$) -- cycle;
          }
          \newcommand{\alternateColorBetweenTwoNodes}[4]{
            \fill[#1] ($(#2) + (0, .08)$) to ($(#2) - (0, .08)$) to ($0.75*(#2) + 0.25*(#3) - (0,.08)$) to ($0.75*(#2) + 0.25*(#3) + (0,.08)$) -- cycle;
            \fill[#4] ($0.75*(#2) + 0.25*(#3) + (0, .08)$) to ($0.75*(#2) + 0.25*(#3) - (0, .08)$) to ($0.5*(#2) + 0.5*(#3) - (0,.08)$) to ($0.5*(#2) + 0.5*(#3) + (0,.08)$) -- cycle;
            \fill[#1] ($0.5*(#2) + 0.5*(#3) + (0, .08)$) to ($0.5*(#2) + 0.5*(#3) - (0, .08)$) to ($0.25*(#2) + 0.75*(#3) - (0,.08)$) to ($0.25*(#2) + 0.75*(#3) + (0,.08)$) -- cycle;
            \fill[#4] ($0.25*(#2) + 0.75*(#3) + (0, .08)$) to ($0.25*(#2) + 0.75*(#3) - (0, .08)$) to ($(#3) - (0,.08)$) to ($(#3) + (0,.08)$) -- cycle;
          }
          \newcommand{\twoColorsBetweenNodes}[4]{
            \fill[#1] ($(#2) + (0, .08)$) to ($(#2) - (0, .0)$) to ($(#3) + (0,.)$) to ($(#3) + (0,.08)$) -- cycle;
            \fill[#4] ($(#2) + (0, -.08)$) to ($(#2) - (0, .0)$) to ($(#3) - (0,.0)$) to ($(#3) - (0,.08)$) -- cycle;
          }
          \colorBetweenTwoNodes{green}{a1}{dai}
          \colorBetweenTwoNodes{red}{a2}{dai}
          \twoColorsBetweenNodes{red}{dai}{al1}{green}
          \colorBetweenTwoNodes{brown}{b1}{dbi}
          \colorBetweenTwoNodes{yellow}{b2}{dbi}
          \twoColorsBetweenNodes{brown}{dbi}{al1}{yellow}
          \colorBetweenTwoNodes{gray}{a1j}{daj}
          \colorBetweenTwoNodes{blue}{a2j}{daj}
          \twoColorsBetweenNodes{blue}{daj}{al1}{gray}
          \colorBetweenTwoNodes{purple}{b1j}{dbj}
          \colorBetweenTwoNodes{orange}{b2j}{dbj}
          \twoColorsBetweenNodes{orange}{dbj}{al1}{purple}
        \end{scope}

  \end{tikzpicture}

 \end{center}
 \caption{The acceptable 3-sets (i.e., 3-sets which are preferred over the cut-off gadget by all agents they contain) of the edge gadget for the edge $\{v^i_1, v^j_2\}$.
 Acceptable 3-sets are drawn in two ways:
 If they contain one of the ellipses, then they consist of the two agents inside the ellipse and the other endpoint of an edge incident to the ellipse (for example, the edge between~$h^{e, i}_b$ and $\{ f, \barf\}$ corresponds to the acceptable 3-set $\{h^{e, i}_b, f, \barf\}$).
 Otherwise, they are marked by the bold colored paths of length two
 (for example, the red path $a^i_2$-$h^{e, i}_a$-$\alpha_1$ corresponds to the 3-set~$\{a^i_2, h^{e , i}_a, \alpha_1\}$).}
 \label{fig:eg-width}
\end{figure}

Furthermore, we extend the preferences of $a_p^\ell$ by inserting the 2-sets $\{h_a^{e, \ell}, \alpha_1^e\}$ for all $e\in E^{i,j}$ directly after $\{c^\ell_{n-1}, \barc^\ell_{n-1}\}$ (and before $\{d^\ell_1, \bard_1^\ell\}$), and the preferences of $b_p^\ell$ by the 2-sets $\{h_b^{e, \ell}, \alpha_1\}$ for all $e\in E^{i,j}$ directly after $\{c^\ell_{n-1}, \barc^\ell_{n-1}\}$ for every $p \in [n]$ and $\ell \in \{i, j\}$.

We now show that if $E^{i,j}$ contains an edge between the two vertices selected by $S^i$ and $S^j$ for a matching $M$, then $M$ is not stable.

\begin{lemma}\label{lem:eg}
  Let $\mathcal{I}$ consist of two vertex-selection gadgets $S^i$ and $S^j$ and the edge gadget for~$E^{i,j }$ between $S^i$ and $S^j$.

  If for a matching $M$ in $\mathcal{I}$ and an edge $e = \{v^i_r, v^j_s\} \in E^{i,j}$ vertex-selection gadgets $S^i$ and~$S^j$ select the vertices $v^i_r$ and $v^j_s$, then matching $M$ contains a blocking 3-set.
\end{lemma}

\begin{proof}
  We need to show that $M$ contains a blocking 3-set.
  By the cut-off gadgets, we know that any 3-set from $M$ contains three agents belonging to the same edge.
  By the cut-off gadgets for~$f^e$ and $\barf^e$, we know that $M' $ must contain a 3-set $\{h^{e, \ell}_x, f^e, \barf^e\}$ for some $\ell\in \{i,j\}$ and $x\in \{a, b\}$.
	We assume $\ell = i$; the case~$ \ell =j $ is symmetric.
	We make a case distinction on whether $x =a $ or $x= b$.

  If $x = a$, then
  we claim that $\{a^i_r, h^{e, i}_a, \alpha_1^e\}$ is a blocking 3-set.
  Agent $a^i_r$ prefers $\{h^{e, i}_a, \alpha^e_1\}$ to~$M(a^i_r) = \{d^i_p, \bard^i_q\}$.
  Agent $h^{e, i}_a$ prefers $\{a^i_r, \alpha^e_1\}$ to $M'(h^{e, i}_a) = \{f^e, \barf^e\}$ as $v^i_r$ is an endpoint of~$e$.
  By the cut-off gadget for $\alpha^e_2$, we have $M' (\alpha^e_1) = \{\alpha^e_2, \alpha^e_3\}$, and therefore, $\alpha_1^e $ prefers~$\{a^i_r, d^{e, i}_a\}$ to~$M'$.

  If $x = b$, then we claim that $\{b^i_{n - r + 1}, h^{e, i}_b, \alpha_1^e\}$ is a blocking 3-set.
  Agent $b^i_{n - r + 1}$ prefers $\{h^{e, i}_b, \alpha^e_1\}$ to $M(b^i_{n - r + 1} ) = \{d^i_p, \bard^i_q\}$ for some $p, q \in [n + 1]$.
  Also agent $h^{e, i}_b$ prefers $\{b^i_{n - r + 1}, \alpha^e_1\}$ to $M'(h^{e, i}_a) = \{f^e, \barf^e\}$ as $v^i_r$ is an endpoint of $e$.
  By the cut-off gadget for $\alpha^e_2$, we have $M' (\alpha^e_1) = \{\alpha^e_2, \alpha^e_3\}$, and therefore, $\alpha_1^e $ prefers~$\{a^i_r, d^{e, i}_a\}$ over $M'$.
\end{proof}

We now turn to the reverse direction, i.e., if $S^i$ and $S^j$ select two vertices which are not connected by an edge from $E^{i,j}$, then we can find a stable matching inside the edge gadget.

\begin{lemma}\label{lem:egBackwards}
	Let $\mathcal{I}$ consist of two vertex-selection gadgets $S^i$ and $S^j$ and the edge gadget for every edge~$e\in E^{i,j }$ between $S^i$ and~$S^j$.

  Given a matching~$M$ inside the vertex-selection gadgets such that at least one endpoint of~$e$ is not selected by a vertex-selection gadget for every $e\in E^{i,j}$, then there exists a matching~$M'$ containing~$M$ such that there is no blocking 3-set containing an agent from an edge gadget for an edge~$e\in E^{i,j}$.
\end{lemma}

\begin{proof}
  Assume that $S^i$ selects $v^i_{i^*}$ and $S^j$ selects~$v^j_{j^*}$.
  Let~$e = \{v^i_r, v^j_s\} \in E^{i,j}$.
  Since $e \neq \{v_{i^*}^i, v_{j^*}^j\}$, we have $r\neq i^*$ or $s\neq j^*$.
  We assume that $r \neq i^*$ (the case $s \neq j^*$ is symmetric by switching the roles of $i$ and $j$).
  If $r < i^*$, then we extend $M$ to the edge gadget by adding the 3-sets $\{h_a^{e, i}, f^e, \barf^e\}$, $\{h_b^{e, i}, g^e_1, \barg^e_1\}$, $\{h_a^{e,j}, g^e_2, \barg^e_2\}$, $\{h_b^{e,j}, g^e_3, \barg^e_3\}$, and $\{\alpha_1^e, \alpha^e_2, \alpha_3^e\}$.
  If $r > i^*$, then we extend $M$ to the edge gadget by adding the 3-sets $\{h_a^{e, i}, g^e_1, \barg^e_1\}$, $\{h_b^{e, i}, f^e, \barf^e\}$, $\{h_a^{e,j}, g^e_2, \barg^e_2\}$, $\{h_b^{e,j}, g^e_3, \barg^e_3\}$, and $\{\alpha_1^e, \alpha^e_2, \alpha_3^e\}$.

  It remains to show that this does not lead to a blocking 3-set.
  First, note that no blocking 3-set can contain two agents belonging to different edge gadgets (because for two edges $e, e' \in E^{i,j}$ with $e \succ^{i,j} e'$, every agent~$x$ from the edge gadget for $e'$ prefers $M(x)$ to any 2-set containing an agent from the edge gadget for $e'$).
  We assume again without loss of generality that $r \neq i^*$.
  None of the agents $g_p^e$, $\barg^e_p$, $f^e$, and $\barf^e$ is contained in a blocking 3-set (any blocking 3-set containing~$g_p^e$ also contains $\barg_p^e$; however, $\barg_p^e$~ranks the acceptable triples in reverse order to $g_p^e$; a symmetric argument applies to $f^e$ and~$\barf^e$).

  If $r < i^*$, then agents $h^{e, i}_b$, $h^{e,j}_a$, and $h^{e,j}_b$ are not part of a blocking 3-set because all 2-sets which they prefer to $\{g^e_p, \barg^e_p\}$ contain an agent $g_{p'}^{e'}$.
  The only remaining possible blocking 3-set is $\{a^i_p, \alpha_1^e, h_a^{e, i}\}$ for $p \le r < i^*$.
  However, $M(a^i_p) = \{c^i_q, \barc^i_{q'}\}$ and thus~$a^i_p$ prefers $M (a^i_p)$ to $\{\alpha_1^e, h_a^{e, i}\}$.

  If $r > i^*$, then agents $h^{e, i}_a$, $h^{e, j}_a$, and $h^{e,j}_b$ are not part of a blocking 3-set because all 2-sets which they prefer to $\{g^e_p, \barg^e_p\}$ contain an agent $g_{p'}^{e'}$.
  The only remaining possible blocking 3-set is thus $\{b^i_p, \alpha_1^e, h_b^{e, i}\}$ for $p \le n- r + 1$.
  However, since $n - r + 1 \le n- i^*$, it follows that $M(b^i_p ) = \{c^i_q, \barc^i_{q'}\}$ and therefore $b^i_p$ prefers $M (b^i_p)$ to $\{\alpha_1^e, h_b^{e, i}\}$.
\end{proof}

Having described the reduction and the crucial properties of the gadgets, the correctness of the reduction now easily follows.

\paragraph{Proof of the forward direction.}

We split the proof of correctness into two parts and start by showing that a multicolored independent set implies a stable matching.

\begin{lemma}\label{lem:fw}
  If $G$ contains a multicolored independent set $I$, then there exists a stable matching.
\end{lemma}

\begin{proof}
  We construct a stable matching as follows.
  Each vertex-selection gadget $S^i$ selects the vertex from $I\cap V^i$.
  This matching is extended to the edge gadgets and cut-off gadgets as described in \Cref{lem:cog,lem:egBackwards}.
  The stability of the constructed matching follows from \Cref{lem:cog,lem:egBackwards,lem:vsgBackwards}.
\end{proof}

\paragraph{Proof of the backward direction.}

We now turn to the backward direction, showing that a stable matching implies a multicolored independent set.

\begin{lemma}\label{lem:bw}
  If there exists a stable matching, then $G$ contains a multicolored independent set.
\end{lemma}

\begin{proof}
  By \Cref{lem:vsg}, every vertex-selection gadget selects a vertex.
  By \Cref{lem:eg}, no two selected vertices are adjacent.
  Thus, the selected vertices form a multicolored independent set.
\end{proof}

\paragraph{The parameter.}\label{sec:par-poset}

It remains to show that the preferences of the constructed \dsmpo-instance can be derived from a poset~$\succ_{\ML}$ of bounded width.

The poset~$\succ_{\ML}$ looks as follows.
For each vertex-selection gadget $S^i$, we have $a_j^i \succ_{\ML} a_{j'}^i$, $b_j^i \succ_{\ML} b_{j'}^i$, $c_j^i \succ_{\ML} c_{j'}^i$, $\barc_j^i \succ_{\ML} \barc_{j'}^i$, $d_j^i \succ_{\ML} d_{j'}^i$, and $\bard_j^i \succ_{\ML} \bard_{j'}^i$ if and only if $j < j'$.
Furthermore, we have for every $q \in [6]$ that $z^q_{a_j^i} \succ_{\ML} z^q_{a_{j'}^i}$, $z^q_{b_j^i} \succ_{\ML} z^q_{b_{j'}^i}$, $z^q_{c_j^i} \succ_{\ML} z^q_{c_{j'}^i}$, $z^q_{\barc_j^i} \succ_{\ML} z^q_{\barc_{j'}^i}$, $z^q_{d_j^i} \succ_{\ML} z^q_{d_{j'}^i}$, and $z^q_{\bard_j^i} \succ_{\ML} z^q_{\bard_{j'}^i}$.

For each $i, j\in [k]$ with $i < j$, each $e, e' \in E^{i,j}$ and $p \in [3]$, we have $h^{e, i}_a \succml h^{e', i}_a$, $h^{e, i}_b \succml h^{e', i}_b$, $h^{e, j}_a \succml h^{e', j}_a$, $h^{e, j}_b \succml h^{e', j}_b$, $h^{e, i}_b \succml h^{e', i}_b$, $g^e_p \succml g^{e'}_p$, $f^e \succml f^{e'}$, and $\alpha^e_p \succml \alpha^{e'}_p$ if and only if $ e \succ^{i, j} e'$.
Furthermore, we have for every $p\in [3]$ and $q\in [6]$ that $z^q_{h^{e, i}_a} \succml z^q_{h^{e', i}_a}$, $z^q_{h^{e, i}_b} \succml z^q_{h^{e', i}_b}$, $z^q_{h^{e, j}_a} \succml z^q_{h^{e', j}_a}$, $z^q_{h^{e, j}_b} \succml z^q_{h^{e', j}_b}$, $z^q_{h^{e, i}_b} \succml z^q_{h^{e', i}_b}$, $z^q_{g^e_p} \succml z^q_{g^{e'}_p}$, $z^q_{f^e} \succml z^q_{f^{e'}}$, and $z^q_{\alpha^e_p} \succml z^q_{\alpha^{e'}_p}$ if and only if $ e \succ^{i, j} e'$.

We arrive at the following observation.

\begin{observation}\label{lem:pw}
	The preferences are derived from a poset of width at most $O(k^2)$.
\end{observation}

\begin{proof}
 It is easy to verify that the preferences are derived from~$\succ_{\ML}$, so it remains to show that $\succ_{\ML}$ has width $O(k^2)$.
 Note that $\succ_{\ML}$ decomposes in $O(k^2)$ chains, from which the observation follows by Dilworth's Theorem~\cite{Dilworth50}:
 For every vertex-selection gadget~$S^i$, we have 42 chains $a^i_1 \succ_{\ML} \dots \succ_{\ML} a^i_n$, $b^i_1 \succml \dots \succml b^i_n$, $c^i_1 \succml \dots \succml c^i_n$, $\barc^i_1 \succml \dots \succml \barc^i_n$, $d^i_1 \succml \dots \succml d^i_n$, and $\bard^i_1 \succml \dots \succml \bard^i_n$ as well as for every $q \in [6]$, chains $z^q_{a^i_1} \succ_{\ML} \dots \succ_{\ML} z^q_{a^i_n}$, $z^q_{b^i_1} \succml \dots \succml z^q_{b^i_n}$, $z^q_{c^i_1} \succml \dots \succml z^q_{c^i_n}$, $z^q_{\barc^i_1} \succml \dots \succml z^q_{\barc^i_n}$, $z^q_{d^i_1} \succml \dots \succml z^q_{d^i_n}$, and $z^q_{\bard^i_1} \succml \dots \succml z^q_{\bard^i_n}$.
 For every $i, j \in [k]$ with $i < j$, let $E^{i,j} = \{e_1, \dots, e_s\}$ such that $e_r \succ^{i,j} e_{r+1}$ for every $r \in [s-1]$.
 We have 77 chains $h^{e_1, i}_a \succml \dots \succml h^{e_s, i}_a$, $h^{e_1, i}_b \succml \dots \succml h^{e_s, i}_b$, $h^{e_1, j}_a \succml \dots \succml h^{e_s, j}_a$, $h^{e_1, j}_b \succml \dots  \succml h^{e_s, j}_b$, $h^{e_1, i}_b \succml\dots \succml  h^{e_s, i}_b$, $g^{e_1}_p \succml \dots \succml g^{e_s}_p$ for every $p\in[3]$, $f^{e_1} \succml \dots \succml f^{e_s}$, and $\alpha^{e_1}_p \succml \alpha^{e_s}_p$ for every $p\in[3]$ as well as for every $q\in [6]$, chains~$z^q_{h^{e_1, i}_a} \succml \dots \succml z^q_{h^{e_s, i}_a}$, $z^q_{h^{e_1, i}_b } \succml \dots \succml z^q_{h^{e_s, i}_b}$, $z^q_{h^{e_1, j}_a }\succml \dots \succml z^q_{h^{e_s, j}_a}$, $z^q_{h^{e_1, j}_b} \succml \dots  \succml z^q_{h^{e_s, j}_b}$, $z^q_{h^{e_1, i}_b} \succml\dots \succml  z^q_{h^{e_s, i}_b}$, $z^q_{g^{e_1}_p} \succml \dots \succml z^q_{g^{e_s}_p}$ for every $p\in[3]$, $z^q_{f^{e_1}} \succml \dots \succml z^q_{f^{e_s}}$, and $z^q_{\alpha^{e_1}_p} \succml z^q_{\alpha^{e_s}_p}$ for every $p\in[3]$.
\end{proof}

We now have all ingredients to obtain the following main result.

\begin{theorem}\label{thm:width_w-h}
  \tdsrp\ parameterized by poset width is \Wone-hard.
\end{theorem}

 \begin{proof}
    The reduction clearly runs in polynomial time.
	\Cref{lem:fw,lem:bw} prove its correctness, and \Cref{lem:pw} shows that the width of the poset is bounded by $O(k^2)$.
	Thus, we found a parameterized reduction from \textsc{Multicolored Independent Set} to \tdsrp\ parameterized by the width of the poset, proving \Wone-hardness.
 \end{proof}

After an \FPT-algorithm for the parameter maximum number of agents incomparable to a single agent in \Cref{sec:fpt}, we have seen \Wone-hardness for the stronger parameter width of the master poset.
It remains open whether \tdsrp\ parameterized by the width of the master poset lies in \XP.

We now investigate a third parameter measuring similarity to a strictly ordered master poset, namely the (agent) deletion distance to a strictly ordered master poset.

\subsection{Deletion distance to a strictly ordered master poset}
\label{sec:del-dist}

We saw that \dsmpo is fixed-parameter tractable when parameterized by the maximum number of agents incomparable to a single agent, but it is \Wone-hard when parameterized by the width of the poset.
We now consider another natural parameter measuring the similarity to a strict order, namely the deletion distance to a strict order, i.e., the minimum number of agents which need to be deleted such that the resulting preferences are derived from a strict order (note that this does not pose any condition on the preferences of the deleted agents).
Notably, this parameter is orthogonal to the two parameters investigated before (width of the poset and maximum number of agents incomparable to an agent), as can be seen from the following two examples:
If the master poset is the weak order~$a_1 \perp_{\ML} a_2 \succ_{\ML} a_3 \perp_{\ML} a_4 \succ_{\ML} a_5 \perp_{\ML} a_6 \succ_{\ML} \dots \succ_{\ML} a_{n-1} \perp_{\ML} a_{n}$, thus $\kappa (\ML) = 2$, while one has to delete $n/2$~agents in order to obtain a strict order.
If the preferences of all but one agent are derived from a strict order, and the last agent's preferences are derived from the inverse of this strict order, then the deletion distance is one while any master poset from which this preferences are derived contains a tie of size at least $n-1$ and thus has width at least~$n-1$.
In this section, reducing from \textsc{Multicolored clique} we show that \mdsr\ is \Wone-hard parameterized by the deletion distance to a strictly ordered master poset.
First, we formally introduce the parameter.

\begin{definition}
  For an \mdsr-instance $\mathcal{I}$, let $\lambda (\mathcal{I})$ denote the minimum number of agents such that the preferences of the instance arising through the deletion of these agents are derived from a strict order.
\end{definition}

Note that the agents which were deleted to arrive at a strictly ordered master poset may have arbitrary preferences.
We now present an example for this parameter.

\begin{example}
 Consider an \mdsr\ instance $\mathcal{I}$ with $d = 3$ and the following preferences.
 \begin{align*}
  a_1 : \{a_2, a_3\} \succ \{a_3, a_4\} \succ \{a_2, a_5\}\succ \{a_2, a_4\} \succ \{a_3, a_5\} \succ \{a_4, a_5\},\\
  a_2 : \{a_1, a_3\} \succ \{a_3, a_4\} \succ \{a_1, a_5\}\succ \{a_1, a_4\} \succ \{a_3, a_5\} \succ \{a_4, a_5\},\\
  a_3 : \{a_1, a_2\} \succ \{a_1, a_4\} \succ \{a_2, a_5\}\succ \{a_2, a_4\} \succ \{a_1, a_5\} \succ \{a_4, a_5\},\\
  a_4 : \{a_1, a_5\} \succ \{a_1, a_2\} \succ \{a_2, a_5\}\succ \{a_1, a_3\} \succ \{a_3, a_5\} \succ \{a_2, a_3\}, \\
  a_5 : \{a_2, a_3\} \succ \{a_3, a_4\} \succ \{a_1, a_2\}\succ \{a_2, a_4\} \succ \{a_1, a_3\} \succ \{a_1, a_4\}.
 \end{align*}
 After the deletion of $a_5$, the preferences are
 \begin{align*}
  a_1 : \{a_2, a_3\} \succ \{a_3, a_4\} \succ \{a_2, a_4\} ,\\
  a_2 : \{a_1, a_3\} \succ \{a_3, a_4\} \succ \{a_1, a_4\} ,\\
  a_3 : \{a_1, a_2\} \succ \{a_1, a_4\} \succ \{a_2, a_4\} ,\\
  a_4 : \{a_1, a_2\} \succ \{a_1, a_3\} \succ \{a_2, a_3\}.
 \end{align*}
 which are derived from the strict order $a_1 \succ a_2 \succ a_3 \succ a_4$.
 Hence, $\lambda(\mathcal{I} ) = 1$.
\end{example}

To show that \textsc{3-DSR} parameterized by $\lambda ( \mathcal{I})$ is \Wone-hard,
we give a parameterized reduction from \textsc{Multicolored Clique}, which is \Wone-complete \cite{DBLP:series/txcs/DowneyF13,Pietrzak03}.
\defProblemTask{\textsc{Multicolored Clique}}
{A $k$-partite graph $G = (V^1 \disunion V^2 \disunion \dots \disunion V^k, E)$.}
{Decide whether $G$ contains a clique $C$ with $C\cap V^i\neq \emptyset$ for all $i\in [k]$.}

The sets $V^1,\dots, V^k$ are called \emph{color classes}.
Let $E^{i,j}$ be the set of edges with one endpoint in $V^i$ and one endpoint in $V^j$.
By adding vertices and edges, we may assume without loss of generality that there exists some $n'\in \mathbb{N}$ such that $|V^i| = 3n' +1$ for all $i\in [k]$, and that there exists some $m'\in \mathbb{N}$ such that $|E^{i,j}|= 3m'+1$ for all~$i, j \in [k]$.
For each color class $V^i$, we fix an arbitrary order of the vertices, i.e., $V^i = \{v^i_1, v^i_2, \dots, v^i_{3n'+1}\}$.
Furthermore, for each $i<j \in [k]$, we fix an arbitrary order of $E^{i,j} = \{e^{i,j}_1, \dots, e^{i,j}_{3m'+1}\}$, and we set $e^{j, i}_\ell\coloneqq e^{i,j}_\ell$.
For a vertex $v\in V(G)$, we denote by $\delta (v)$ the set of all edges incident to~$v$.

As in \Cref{sec:width}, we only describe the beginning of the preferences of an agent, followed by $\pend$.
The remaining acceptable 2-sets can be added in an arbitrary way obeying the strictly ordered master poset (extended to the $\lambda (\mathcal{I})$ agents not contained in this poset).

The basic idea of the reduction is that we create for every $i\in [k]$ a vertex-selection gadget~$S^i$ and for every $i < j \in [k]$ an incidence-checking gadget.
The vertex-selection gadget~$S^i$ then encodes the selection of a vertex from~$V^i$ to be part of the multicolored clique, and the incidence-checking gadgets ensure that all vertices selected by the vertex-selection gadgets are indeed incident.

We begin by describing the agents in the \dsmpo instance constructed by the reduction, and the master poset.
Afterwards, we describe the gadgets used in the reduction.

\subsubsection{Master poset}
\label{sec:master-poset}
Every vertex-selection gadget $S^i$ has $3n' +3 $ agents, namely $s^i$, $\overleftarrow{s}^i$, and every vertex from $V^i$ will also be an agent in~$S^i$, i.e., $S^i$ also contains agent~$v_1^i, \dots, v_{3n'+1}^i$.
Every incidence-checking gadget $B^{i,j}$ contains $3 (3m' + 1) + 6 $ agents:
$b^{i, j}_\ell$ and $b^{j, i}_\ell$ for $\ell \in [3m' + 1]$, agent~$c^{i,j}_\ell$ for $\ell \in [3m']$, and agents $x^{i,j}$, $\overleftarrow{x}^{i,j}$, $x^{j, i}$, $\overleftarrow{x}^{j, i}$, $ z^{i,j}_1$, $z^{i,j}_2$, and~$z^{i,j}_3$.
Additionally, there are $O(k^2)$ agents contained in so-called ``cut-off gadgets'' (see \Cref{sec:cog});
however, all these agents will be deleted in order to have the preferences derived from a strictly ordered master poset, and thus, are not considered here.
Furthermore, in order to arrive at a strict order as master poset, we delete agents $s^i$ and $\overleftarrow{s}^i$ for every $i\in [k]$ as well as agents $x^{i,j}$, $\overleftarrow{x}^{i,j}$, $x^{j, i}$, $\overleftarrow{x}^{j, i}$, $ z^{i,j}_1$, $z^{i,j}_2$, and~$z^{i,j}_3$ for every $i < j \in [k]$.
Inside the vertex-selection gadget~$S^i$, the master poset has the form $\mathcal{A}^i \coloneqq v^i_1 \succml v^i_2 \succml \dots \succml v^i_{3n'+1}$.
Inside incidence-checking gadget~$B^{i,j}$, the master poset has the form $\mathcal{B}^{i,j} \coloneqq 
    b_{1}^{i, j} \succ b_{1}^{j, i} \succ c^{i,j}_1 \succ b_{e}^{i, j} \succ b_{2}^{j, i} \succ c_2^{i, j} \succ \dots \succ c_{3m'}^{i, j} \succ b_{{3m'+1}}^{i, j} \succ b_{{3m'+1}}^{j, i}$.
The complete master poset looks as follows.
\begin{align*}
  \mathcal{A}^1 & \succ_{\ML} \mathcal{A}^2 \succ_{\ML} \dots \succ_{\ML} \mathcal S^k \succml \mathcal{B}^{1,2} \succml \mathcal{B}^{1,3} \succml \dots \succml \mathcal{B}^{1, k}\\
  &\succml \mathcal{B}^{2,3} \succml \mathcal{B}^{2, 4} \succml \dots \succml \mathcal{B}^{2,k} \succml \mathcal{B}^{3, 4} \succml \dots \succml \mathcal{B}^{k-1, k}.
\end{align*}

We continue by describing the gadgets used in the reduction.

\subsubsection{Cut-off gadget}
\label{sec:cog}

As in \Cref{sec:width}, we have a cut-off gadget ensuring for a given agent~$a$ and a 2-set~$p$ that $a$~is matched to a 2-set it likes at least as much as~$p$.
Indeed, we use the cut-off gadget from~\Cref{sec:width}.
Since we will use only $O(k^2)$ many cut-off gadgets, adding these agents may increase the deletion distance to a strict master poset by at most $O(k^2)$;
thus, we can choose the preferences arbitrarily and do not have to consider the master poset.

We already showed in \Cref{lem:cog} that cut-off gadgets work as desired.
Having described the cut-off gadgets, we can now describe the remaining gadgets, namely the vertex-selection gadget and the incidence-checking gadget.
We start with the vertex-selection gadget.

\subsubsection{Vertex-selection gadget}

For each color class $V^i$, the reduction adds a vertex selection gadget $S^i$.
This vertex-selection gadget contains an agent for each vertex $v^i_p$ with $p\in [3n' + 1]$;
we identify the agent and the vertex and call both $v^i_p$.
Furthermore, the vertex-selection gadget contains two agents $s^i$ and~$\overleftarrow{s}^i$.
The vertex-selection gadget also contains a cut-off gadget for each of the three agents $s^i$, $\overleftarrow{s}^i$, and~$v^i_{3n'+1 }$.

The intuitive idea behind the vertex-selection gadget is as follows.
Selecting a vertex $v^i_\ell$ to be part of a multicolored clique corresponds to matching this vertex to $\{s^i, \overleftarrow{s}^i\}$, and partitioning the remaining agents from~$\{v^i_1, \dots, v^i_{3n'+1}\} \setminus \{v^i_\ell\}$ into $n'$ 3-sets which are contained in the stable matching.
Since $s^i$ prefers being matched to~$\{\overleftarrow{s}^i, v^i_\ell\}$ with small $\ell$ while $\overleftarrow{s}^i$ prefers being matched to $\{s^i, v^i_\ell\}$ with large $\ell$, how much $s^i$ and $\overleftarrow{s}^i$ like their partners in a matching encodes which vertex is selected.
The cut-off gadget for $v^i_{3n' + 1}$ actually implies that in every stable matching~$M$, there has to exist some $\ell \in [3n'+1]$ such that $\{v_\ell^i, s^i, \overleftarrow{s}^i\} \in M$ (see \Cref{lem:vertex-sel}).
This implies that every vertex has to be matched this way, and thus, the vertex-selection gadget has to select a vertex in each stable matching.

The preferences of an agent $v^i_p \in V^i$ start with all 2-sets containing an agent from a vertex-selection gadget~$S^j$ with $j < i$, and continue with $\{s^i, \overleftarrow{s}^i\} \succ \bigl\{\{v^i_1, v^i_{q} \}: q\in [3n'+1] \bigr\} \succ \bigl\{\{v^i_2, v^i_{q} \}: q\in [3n'+1] \bigr\} \succ \dots \succ \bigl\{\{v^i_{3n' + 1},  v^i_{q}\} : q\in [3n'+1] \bigr\} \pend$.

For the agent $v^i_{3n'+1}$, a cut-off gadget follows; for the other agents, the preferences are extended in an arbitrary way obeying the master poset.

The preferences of $s^i$ and $\overleftarrow{s}^i$ are as follows (the agents $x^{i,j}$, $\overleftarrow{x}^{i,j}$, and $b_\ell^{i,j}$ are contained in incidence-checking gadgets (see \Cref{sec:incidence})).
\begin{align*}
  s^i \colon &\{\overleftarrow{s}^i, v^i_1\}\succ \bigl\{\{b_\ell^{i, j}, \overleftarrow{x}^{i, j}\} : {j\in [k] \setminus \{i\}, e^{i,j}_\ell \in E^{i,j} \cap \delta (v^i_1)} \bigr\}
  \succ \{\overleftarrow{s}^i, v^i_2\}\\
  & \succ \bigl\{\{b^{i, j}_\ell, \overleftarrow{x}^{i,j }\} : {j\in [k] \setminus \{i\}, e^{i,j}_\ell \in E^{i,j} \cap \delta (v^i_2)} \bigr\}\succ \{\overleftarrow{s}^i, v^i_3\} \succ \dots \succ \{\overleftarrow{s}^i, v^i_{3n' + 1}\} \succ \CO_{s_i}\\
  \overleftarrow{s}^i \colon &\{s^i, v^i_{3n' + 1} \} \succ \bigl\{\{b^{i,j}_\ell, x^{i,j }\} : {j\in [k] \setminus \{i\}, e^{i,j}_\ell \in E^{i,j} \cap \delta (v^i_{3n'+1})}\bigr\} \succ \{s^i, v^i_{3n'}\}, \\
  & \succ \bigl\{\{b_\ell^{i,j}, x^{i,j} \} : {j\in [k] \setminus \{i\}, e^{i,j}_\ell\in E^{i,j} \cap \delta (v^i_{3n'})} \bigr\} \succ \{s^i, v^i_{3n' -1}\} \succ\dots \succ \{s^i, v^i_1\} \succ \CO_{\overleftarrow{s}_i}.
\end{align*}

Next, we turn to the incidence-selection gadget.

\subsubsection{Incidence-checking gadget}
\label{sec:incidence}

For each pair $(V^i, V^j)$ of color classes with $i < j$, we add an incidence-checking gadget $B^{i,j}$.
For each edge $\ell\in [3m' + 1]$, the incidence-checking gadget~$B^{i,j}$ contains two agents $b_{\ell}^{i,j} $ and~$b_{\ell}^{j,i}$.
Furthermore, there are $3m'$ agents $c_1^{i,j}, \dots, c_{3m'}^{i,j}$.
The gadget also contains seven agents $x^{i,j}$, $\overleftarrow{x}^{i,j}$, $x^{j,i}$, $\overleftarrow{x}^{j,i}$, $z^{i, j}_1$, $z^{i, j}_2$, and~$z^{i, j}_3$, which do not have the master preferences, and a cut-off gadget for~$z^{i, j}_2$.

The idea of the gadget is as follows:
Every stable matching has to contain $\{b_\ell^{i,j}, b_\ell^{j, i}, c_\ell^{i,j}\}$ or $\{b_\ell^{i,j}, b_\ell^{j, i}, c_{\ell - 1}^{i,j}\}$ for all $\ell \in [3m'+1]\setminus\{r\}$ for some $r\in [3m' + 1]$, while for the remaining~$r$, agent~$b_r^{i,j}$ must be matched to $\{x^{i, j}, \overleftarrow{x}^{i,j}\}$ and $b_r^{j, i}$ must be matched to $\{x^{j,i}, \overleftarrow{x}^{j, i}\}$.
Thus, the better $x^{i,j}$ is matched, the worse $\overleftarrow{x}^{i,j}$ is matched.
Unless $e_r$ is incident to the vertex selected by vertex-selection gadget~$S^i$, it follows that either $x^{i,j}$ and~$s^i$ or $\overleftarrow{x}^{i,j}$ and $\overleftarrow{s}^i$ are part of a blocking 3-set together with the vertex selected by~$S^i$.
By symmetric arguments, $e_r$ is incident to the vertex selected by $S^j$ in any stable matching, and thus, the vertices selected by the vertex-selection gadgets form a clique.

Let $Z^{i,j} $ be the set of agents which are before $b_1^{i,j}$ in the master poset.
As the agents from~$Z^{i,j}$ are before every agent from~$B^{i', j'}$ in the master poset, we have to add them in the beginning of the preferences of agent~$b^{i,j}_\ell$ and $b^{j, i}_\ell$.
Let $X\coloneqq \{x^{i,j}, \overleftarrow{x}^{i,j}, x^{j, i}, \overleftarrow{x}^{j, i}\}$ and let $A$ denote the set of all agents.

The preferences of $b_{\ell}^{i,j}$ (resp.~$b_{\ell}^{j, i}$) are as follows (deleting all 2-sets containing $b_{\ell}^{i,j}$ (resp.~$b_{\ell}^{j, i}$)).
\begin{align*}
  \{ \{z, a\} & : z\in Z^{i,}, a\in A\} \succ
  \bigl\{\{a, x\} : a\in \{s^i,\overleftarrow{s}^i, s^j, \overleftarrow{s}^j\}, x\in  X \bigr\} \\
  & \succ \{\{x, x'\}: x, x' \in X\}
  \succ  \bigl\{ \{x, b_q^{i,j}\}, \{x, b_q^{j, i}\}: x\in X,  q\in [3m' + 1]\bigr\}\\
  &\succ \bigl\{\{b_q^{i,j}, z^{i,j}_1 \}, \{b_q^{j, i}, z^{i, j}_1\} : q\in [\ell - 1] \bigr\} \succ \binom{\{b_q^{i,j}, b_q^{j, i}, c_q^{i,j } : q\in [\ell]\}}{2} \\
  & \succ \bigl\{\{b_q^{i,j}, z^{i, j}_1\}, \{b_q^{j, i}, z^{i, j}_1\} : q\in [\ell, 3m' + 1]\bigr\} \pend.
\end{align*}
The preferences of $c_q^{i,j}$ are arbitrary preferences obeying the master poset.\\
The preferences of $z^{i, j}_1$ are as follows.
$\bigl\{\{b_p^{i,j}, b_q^j\} : p\in [3m'+ 1], q\in [3m' + 1]\bigr\} \succ \{z^{i,j}_2, z^{i,j}_3\} \pend$.
The preferences of $z^{i,j}_2$ are as follows.
$\{z^{i,j}_1, z^{i,j}_3\} \succ \CO_{z^{i,j}_2}$.\\
The preferences of $z^{i,j}_3$ are as follows.
$\{z^{i,j}_1, z^{i,j}_2\} \pend$.

To describe the preferences of $x^{i,j}$ and $\overleftarrow{x}^{i,j}$, we define sublists $\mathcal C_\ell^\alpha$ and $\overleftarrow{\mathcal C}_\ell^\alpha$ for $\ell \in {[3n' + 1]}$ and $\alpha \in \{i, j\}$.
The sublist~$C_\ell^i$ contains the 2-sets $\{\overleftarrow{x}^{i,j}, b_r^{i,j}\}$ for all $e_r^{i,j}\in E^{i,j} \cap \delta (v^i_{\ell})$, ordered increasingly by $r$.
The sublist~$\overleftarrow{\mathcal C}_\ell^i$ contains the 2-sets~$\{x^{i,j}, b_r^{i,j}\}$ for $e_r^{i,j}\in E^{i,j} \cap \delta (v^i_{\ell})$, but is ordered \emph{decreasingly} by~$r$.
Similarly, sublist $\mathcal C_\ell^j$ contains the 2-sets $\{\overleftarrow{x}^{j, i}, b_r^{j, i}\}$ for all~$e_r^{j, i}\in E^{i,j} \cap \delta (v^i_{\ell})$, ordered increasingly by~$r$, and sublist $\overleftarrow{\mathcal C}_\ell^j$ contains 2-sets~$\{x^{j, i}, b_r^{j, i}\}$ ordered decreasingly by~$r$.
The preferences of $x^{i,j}$ and~$\overleftarrow{x}^{i,j}$ look as follows.
\begin{align*}
  x^{i,j} \colon& \mathcal C_{3n' + 1}^i \succ \{\overleftarrow{s}^i, v_i^{3n'+1}\}  \succ \mathcal C_{3n'}^i \succ \{\overleftarrow{s}^i, v_i^{3n'}\} \succ \dots \succ \mathcal C_1^i  \succ \{\overleftarrow{s}^i, v_i^{1}\} \succ \CO_{x^{i,j}},\\
  \overleftarrow{x}^{i,j} \colon& \overleftarrow{\mathcal C}_1^i \succ \{{s}^i, v_i^{1}\} \succ \overleftarrow{\mathcal C}_2^i \succ \{{s}^i, v_i^2\} \succ \dots \succ \overleftarrow{\mathcal C}_{3n'+1}^i \succ\{{s}^i, v_i^{3n'+1}\}\succ \CO_{\overleftarrow{x}^{i,j}},\\
  x^{j, i} \colon& \mathcal C_{3n' + 1}^j \succ \{\overleftarrow{s}^j, v_j^{3n'+1}\}  \succ \mathcal C_{3n'}^j \succ \{\overleftarrow{s}^j, v_j^{3n'}\} \succ \dots \succ \mathcal C_1^j  \succ \{\overleftarrow{s}^j, v_j^{1}\} \succ \CO_{x^{j, i}},\\
  \overleftarrow{x}^{j, i} \colon& \overleftarrow{\mathcal C}_1^j \succ \{{s}^j, v_j^{1}\} \succ \overleftarrow{\mathcal C}_2^j \succ \{{s}^j, v_j^2\} \succ \dots \succ \overleftarrow{\mathcal C}_{3n'+1}^j \succ\{{s}^j, v_j^{3n'+1}\}\succ \CO_{\overleftarrow{x}^{j, i}}.
\end{align*}

\subsubsection{The reduction}

Given an instance $(G,k)$ of \textsc{Multicolored Clique}, we construct an \textsc{MDSR}-instance $\mathcal{I}'$ as follows.
Instance~$\mathcal{I}'$ contains $k$~vertex-selection gadgets $S^i$, one for each color class $V^i$.
Between each pair $(S^i, S^j)$ of vertex-selection gadgets with $i< j$, there is an incidence-checking gadget~$B^{i,j}$.

We first show that our parameter $\lambda $ is indeed bounded by $O(k^2)$ for the constructed instance~$\mathcal{I}'$.

\begin{lemma}\label{okappa}
  $\lambda (\mathcal{I}')= O(k^2)$.
\end{lemma}

\begin{proof}
  The master poset is described in \Cref{sec:master-poset}.
  It is a strict order and contains all agents but the $O(k^2)$ agents~$s^i$, $\overleftarrow{s}^i$, $x^{i,j}$, $\overleftarrow{x}^{i,j}$, $z^{i,j}$, $z^{i,j}_1$, $z^{i,j}_2$ for $i, j \in [k]$ with $i \neq j$, and all agents contained in cut-off gadgets.
  It is easy to verify that the preferences obey this master poset.
\end{proof}

\subsubsection{Proof of the forward direction}

We prove that if $G$ contains a clique of size $k$, then $\mathcal{I}'$ admits a stable matching.
So let $\{v^1_{p_1}, \dots, v^k_{p_k}\}$ be a multicolored clique.
We construct a stable matching $M$ as follows.

For the vertex-selection gadget $S^i$, we add the 3-set $\{v_i^{p_i}, s^i, \overleftarrow{s}^i\}$.
All other vertex agents~$v_i^q$ are matched to each other, according to their index $q$ (i.e., we match the three agents with lowest index together, then the next three agents, and so on).
Next, we consider an incidence gadget~$B^{i,j}$.
Assume that $\{v^i_{p_i}, v^j_{p_j}\}$ is the $\alpha$-th edge in the order of $E^{i,j}$ fixed in \Cref{sec:incidence}.
We add $\{x^{i,j}, \overleftarrow{x}^{i,j}, b^{i,j}_\alpha\}$ and $\{x^{j, i}, \overleftarrow{x}^{j, i}, b_\alpha^j\}$ to $M$.
Furthermore, we add the 3-sets $\{b_\ell^{i,j}, b_\ell^{j, i}, c_\ell^{i, j}\}$ for $\ell < \alpha $, and $\{b^{i,j}_\ell, b_\ell^{j, i}, c_{\ell- 1}^{i,j }\}$ for $\ell > \alpha$.
Finally, we add the 3-set~$\{z_1^{i,j}, z_2^{i,j }, z_3^{i,j}\} $.
The agents from the cut-off gadgets are matched as described in \Cref{lem:cog}.
We call the resulting matching~$M$.

It remains to show that $M$ is stable.
In order to do so, we will show step by step that no agent is part of a blocking 3-set.
We start with the agents $s^i$ and $\overleftarrow{s}^i$.

\begin{lemma}\label{lsip}
  For any $i\in [k]$, agents $s^i$ and $\overleftarrow{s}^i$ are not part of a blocking 3-set.
\end{lemma}

\begin{proof}
  We prove the lemma for $\overleftarrow{s}^i$;
  the proof for $s^i$ is symmetric.
  All 2-sets which $\overleftarrow{s}^i$ ranks better than $\{s^i, v^i_{p_i}\}$ are of the form $\{s^i, v^i_\ell\}$ for~$\ell > p_i$ or~$\{\overleftarrow{x}^{i,j}, b^{i,j}_{\ell}\}$ for an edge $e_\ell \in E^{i,j}$ whose endpoint in $V^i$ is $v^i_q$ with $q> p_i$.

  For the 2-sets $\{s^i, v^i_\ell\}$ with $\ell > p_i$, note that $s^i$ does not prefer $\{\overleftarrow{s}^i, v^i_\ell\}$ to $\{\overleftarrow{s}^i, v^i_{p_i}\}$, and thus, $\{\overleftarrow{s}^i, s^i, v^i_\ell\}$ is not blocking.

  For the 2-set $\{\overleftarrow{x}^{i,j}, b^{i,j}_{\ell}\}$, note that $\overleftarrow{x}^{i,j}$ does not prefer $\{b_{\ell}^{i,j}, \overleftarrow{s}^i\}$ to $\{\overleftarrow{x}^{i,j}, b_{r}^{i,j}\}$ for the edge~$e_r = \{v^i_{p_i}, v^j_{p_j}\}$ as $q > p_i$.

  Thus, $\overleftarrow{s}^i$ is not part of a blocking 3-set.
\end{proof}

We now turn to the remaining agents from vertex-selection gadgets.

\begin{lemma}\label{lnbv}
  For each $i\le k$, no agent from $V^i$ is part of a blocking 3-set.
\end{lemma}

\begin{proof}
  We prove the statement by induction.
  For $i= 0$ there is nothing to show.

  Fix $i \in [k]$.
  Note that for all $j\in [k]$, agents $v^j_{p_j}$ are matched to their first choice, and thus not part of a blocking 2-set.
  By \Cref{lsip}, no agent $s^j$ or $\overleftarrow{s}^j$ is involved in a blocking 3-set.
  Thus, by induction on $p$, one easily sees that all 2-sets which $v^i_p$ prefers to the 2-set it is matched to in $M$, contain an agent about which we already know that it is not contained in a blocking 3-set, implying that also $v_i^p$ is not contained in a blocking 3-set.

  Thus, no agent of $V^i$ is part of a blocking 3-set.
\end{proof}

Next, we turn to the incidence-checking gadgets and start with agents $x^{i,j}$ and $\overleftarrow{x}^{i,j}$.

\begin{lemma}\label{lx}
  For any $i, j \in [k]$ with $i \neq j$, agents $x^{i,j}$ and $\overleftarrow{x}^{i,j}$ are not part of a blocking 3-set.
\end{lemma}

\begin{proof}
  A blocking 3-set cannot contain $s^i$, $\overleftarrow{s}^i$, $s^j$, or $\overleftarrow{s}^j$ by \Cref{lsip}.
  Thus, it is of the form~$\{x^{i,j}, \overleftarrow{x}^{i,j}, b^{i,j}_\alpha\}$ for some $\alpha \in [3m' + 1]$.
  Since~$x^{i,j}$ prefers $\{\overleftarrow{x}^{i,j}, b^{i,j}_\alpha\}$ to $M(x^{i,j}) = \{\overleftarrow{x}^{i,j}, b^{i,j}_{\beta}\}$ for some $\beta \in [3m' +1]$, agent~$\overleftarrow{x}^{i,j}$ does not prefer $\{x^{i,j}, b^{i,j}_\alpha\}$ to $\{x^{i,j}, b^{i,j}_{\beta}\}$, and thus, $\{x^{i,j}, \overleftarrow{x}^{i, j}, b^{i,j}_\alpha\}$ is not blocking.
\end{proof}

Now we consider agents $z_1^{i,j}$, $z_2^{i, j}$, and $z_3^{i,j}$.

\begin{lemma}\label{lem:z}
 For every $i < j \in [k]$, agents $z_1^{i,j}$, $z_2^{i,j}$, and $z_3^{i,j}$ are not part of a blocking 3-set.
\end{lemma}

\begin{proof}
 Agents $z_2^{i, j}$ and $z_3^{i,j}$ are matched to the first 2-sets in their preference lists and thus are not part of a blocking 3-set.

 All 2-sets which $z_1^{i,j}$ prefers to $\{z_2^{i,j}, z_3^{i,j} \}$ are of the form $\{b_p^\ell, b_q^{\ell'}\}$ for $\ell = i,j$ or $\ell = j,i$ and $\ell' = {i,j}$ or $\ell' = {j, i}$.
 We assume without loss of generality that $p \le q$ and $\ell = i, j$;
 the case $q > p $ or $\ell = j$ is symmetric.
 However, agent $b_p^{i,j}$ is matched to $\{c_p^{i,j}, b_p^{j, i}\}$ if~$ p < p_i$, to $\{x^{i,j}, \overleftarrow{x}^{i, j}\} $ if~$p  = p_i$, and to~$\{c_{p-1}^{i, j}, b_p^{j, i}\}$ if $p > p_i$.
 In all cases, agent $b_p^{i,j}$ prefers $M(b_p^{i,j})$ to~$\{z_1^{i,j}, b_q^{\ell'}\}$, and thus, $z_1^{i, j}$ is not part of a blocking 3-set.
\end{proof}

Finally, we turn to the remaining agents from the incidence-checking gadgets.
\begin{lemma}\label{lev}
  For every $i < j\in [k]$ and each $\ell \in [3m']$, agents $b^{i,j}_\ell$, $b_\ell^{j, i}$, and $c_\ell^{i,j}$ as well as~$b^{i,j}_{3m' + 1}$ and $b^{j, i}_{3m'+1}$ are not part of a blocking 3-set.
\end{lemma}

\begin{proof}
  The argument is similar to the proof of \Cref{lnbv}.

  We show that agents $b^{i,j}_\ell, b^{j, i}_\ell$, and $c^{i, j}_\ell$ are not contained in a blocking 3-set via induction on~$k i + j$.
  For $ k i + j < k+1$, there are no agents $b^{i,j}_\ell $, $b^{j, i}_\ell$, and $c^{i,j}_\ell$, and thus, there is nothing to show.
  So fix an incidence-checking gadget~$B^{i,j}$ with $ki + j \ge k +1$.
  A blocking 3-set cannot contain an agent from a vertex-selection gadget (by \Cref{lsip,lem:z}), an agent from an incidence-checking gadget~$B^{i', j'}$ with $k i' + j' < ki + j$ (by the induction hypothesis), or an agent~$x^{i', j'}$ or $\overleftarrow{x}^{i', j'}$ (by \Cref{lx}).
  Thus, every blocking 3-set only consists of agents $b_\ell^{i,j}, b_\ell^{j, i}$, or $c^{i, j}_\ell$.
  However, the agent with minimal index does not prefer the 3-set to $M$, a contradiction.
  Therefore, no blocking 3-set contains an agent~$b^{i,j}_\ell$, $b^{j, i}_\ell$, or $c^{i,j}_\ell$.
\end{proof}

Now, we can conclude the stability of~$M$.

\begin{lemma}\label{lforward2}
  Matching $M$ is stable.
\end{lemma}

\begin{proof}
  By \Cref{lsip,lx,lev,lem:z,lnbv}, no agent outside a cut-off gadget is contained in a blocking 3-set.
  Thus, by \Cref{lem:cog}, there is no blocking 3-set.
\end{proof}

\subsubsection{Proof of the backward direction}

Finally, we show that if $\mathcal{I}'$ admits a stable matching, then
$G$ contains a clique of size $k$.
We start by showing that every stable matching selects a vertex $v^i_{p_i}$ in every vertex-selection gadget~$S^i$, meaning that it contains the 3-set $\{s^i, \overleftarrow{s}^i, v_i^{p_i}\}$ for some $p_i\in [3n' + 1]$.

\begin{lemma}\label{lem:vertex-sel}
  For  $i\in [k]$, every stable matching $M$ contains a 3-set $\{s^i, \overleftarrow{s}^i, v^i_{p_i}\}$ for some~$p_i\in [3n'+ 1]$, and all other agents $v^i_q$ are matched to each other for $q\neq p_i$.
\end{lemma}

\begin{proof}
  Consider a stable matching $M$.
  We prove the statement by induction on $i$.
  For $i = 0$, there is nothing to show.
  So fix $i> 0$.
	By the induction hypothesis, we know that no agent from a vertex-selection gadget~$S^{i'}$ with $i' < i$ is matched to a 2-set containing an agent from~$S^i$.
  Ignoring 2-sets containing an agent from~$S^{i'}$ for $i'  < i$, agent $v^i_p$ prefers most to be matched to $\{v^i_q, v^i_{q'}\}$ for $q, q'\in [3n'+1]$, or to~$\{s^i, \overleftarrow{s}^i\}$.
  If no agent $v^i_p$ is matched to $\{s^i, \overleftarrow{s}^i\}$, then $\{v^i_{3r-2}, v^i_{3r-1}, v^i_{3r}\} \in M$ for all $r\in [n']$ by the same arguments as in the proof of~\Cref{tconsistent} (as after ignoring also $\{s^i, \overleftarrow{s}^i\}$, the beginning of the preferences is derived from the strict order $v^i_1 \succ v^i_2 \succ \dots \succ v^i_{3n'+ 1}$).
  The cut-off gadget for $v^i_{3n' + 1}$ then implies that $M$ is not stable.
  Thus, there exists some $p_i\in [3n' + 1]$ such that $\{v^i_{p_i}, s^i, \overleftarrow{s}^i\} \in M$.
  It follows that all other agents $v^i_q$ are matched to each other for $q\neq p_i$ by the same arguments as in the proof of~\Cref{tconsistent}.
\end{proof}

\begin{lemma}
  \label{lem:incidence-checking}
  For every stable matching~$M$ and every incidence-checking gadget~$B^{i,j}$, every 3-set~$t$ containing at least one agent from $B^{i,j}$ contains three agents from~$B^{i,j}$.
\end{lemma}

\begin{proof}
  Assume for a contradiction that there exists a stable matching~$M$, an incidence-checking gadget $B^{i,j}$, and a 3-set~$t$ containing one agent~$b$ from $B^{i,j}$ and one agent~$v$ not contained in~$B^{i,j}$.
  Assume that the incidence-checking gadget is chosen such that $i$ is minimal, and such that $j$ is minimal among all incidence-checking gadgets with minimal $i$.
  By \Cref{lem:vertex-sel}, no agent from $B^{i,j}$ is matched to a 2-set containing an agent from a vertex-selection gadget.
  By the choice of $B^{i,j}$, no agent from $B^{i,j}$ is matched to a 2-set containing an agent from an incidence-checking gadget $B^{i', j'}$ with $i'< i$ or $i' = i $ and $j' \le j$.
  Since there are $9 m' + 9$ agents in the incidence-checking gadget, there are three agents $a_1, a_2$, and $a_3$ which are matched to a 2-set containing at least one agent which is not contained in $B^{i,j}$.
  The cut-off gadget for~$z_2^{i,j}$ implies that $M$ contains~$\{z_1^{i,j}, z_2^{i,j}, z_3^{i,j}\}$.
  If $x^{i,j}$ is not matched to a 2-set~$\{\overleftarrow x^{i,j}, b^{i,j}_r\}$ for some $r \in [3m'+ 1]$, then $\{x^{i,j}, \overleftarrow x^{i,j}, b^{i,j}_r\}$ is blocking, a contradiction to the stability of~$M$.
  By symmetric arguments, $\{x^{i,j}, \overleftarrow x^{i,j}, b^{i,j}_r\}\in M$ for some $r\in [3m' + 1]$.
  Thus, we have that $a_r \in \{b^{i,j}_p, b^{j, i}_p : p \in [3m' + 1]\} \cup \{c^{i,j}_q : q\in [3m']\}$ for every $r \in [3]$.
  It follows that $\{a_1, a_2, a_3\}$ blocks~$M$, a contradiction.
\end{proof}

We now show that for any pair of vertex-selection gadgets, the vertices selected by the vertex-selection gadgets are adjacent in $G$.

\begin{lemma}\label{licg}
  Let $M$ be a stable matching such that vertex-selection gadget $S^i$ selects vertex~$v^i_p$ and $S^j$ selects $v^j_q$.
  Then $G$ contains the edge $\{v^i_p, v^j_q\}$.
\end{lemma}

\begin{proof}
  We assume for a contradiction that the edge $\{v_i^p, v_j^q\}$ is not contained in $G$.
  We need to show that $M$ contains a blocking 3-set.

  Since \Cref{lem:vertex-sel} implies that neither $s^i$ nor $\overleftarrow{s}^i$ is matched to a 2-set containing $x^{i, j}$ or~$ \overleftarrow{x}^{i,j}$, the cut-off gadgets for $x^{i,j}$ and $\overleftarrow{x}^{i,j}$ imply that $x^{i,j}$ and $\overleftarrow{x}^{i,j}$ are matched to an agent~$b^{i,j}_r$, i.e., $\{x^{i,j}, \overleftarrow{x}^{i,j}, b^{i,j}_r\}\in M$.
  Similarly, $x^{j, i}$ and $\overleftarrow{x}^{j,i}$ are matched to an agent $b_s^{j, i}$.

  Let $e_r^{i,j } = \{v^i_{r_i}, v^j_{r_j}\}$.
  If $r_i < p$, then $\{b_{r}^{i,j}, \overleftarrow{s}^i, x^{i,j}\}$ is a blocking 3-set.
  If $r_i > p$, then $\{b^{i,j}_{r}, {s}^i, \overleftarrow{x}^{i,j}\}$ is a blocking 3-set.
  Thus, $v^i_p$ is an endpoint of $e_r$.
  By symmetric arguments, we get that $v^j_q$ is an endpoint of $e_s$.
  Since $\{v_i^p, v_j^q\} \notin E^{i,j}$, it follows that $s \neq r$.

  First assume $r < s$.
  The cut-off gadget for $z^{i,j}_2$ implies that matching $M$ contains the 3-set~$\{z^{i,j}_1, z^{i,j}_2, z^{i,j}_3\}$.
  First, we show by induction on $\ell $ that for every $\ell <r$, matching~$M$ contains $\{b^{i,j}_{\ell'}, b_{\ell'}^{j, i}, c_{\ell'}^{i,j }\}$ for every $\ell' < \ell$.
  For $\ell = 0$, there is nothing to show.
  So fix $\ell  > 0$.
  
  By the induction hypothesis, $M$ contains $\{b^{i,j}_{\ell'}, b_{\ell'}^{j, i}, c_{\ell'}^{i,j }\}$ for every $\ell' < \ell$.
  Every 2-set which $b_{\ell}^{i,j} $ prefers to $\{b^{j, i}_{\ell}, c_\ell^{i,j}\}$ contains an agent which is before $b^{i,j}_1$ in the master poset, an agent from~$X:= \{x^{i,j}, \overleftarrow x^{i,j}, x^{j, i}, \overleftarrow x^{j,i}\}$, or contains an agent $b^{i,j}_{\ell'}, b^{j, i}_{\ell'}$ or $c^{i,j}_{\ell'}$ for some $\ell ' < \ell$.
  By \Cref{lem:incidence-checking}, agents~$b^{i,j}_\ell$, $b^{j, i}_\ell$, and $c^{i,j}_\ell$ are matched to 2-sets containing only agents from the incidence checking gadget.
  Every 2-set which~$b^{i,j}_\ell$ prefers to~$\{b^{j,i}_\ell, c^{i,j}_\ell\}$ contains an agent outside~$B^{i,j}$, an agent from $X := \{x^{i,j}, \overleftarrow x^{i,j}, x^{j, i}, \overleftarrow x^{i,j}\}$, or an agent~$b^{i,j}_q$, $b^{j, i}_q$, or $c^{i,j}_q$ with $q < \ell$.
  Thus, $b^{i,j}_\ell$ is not matched to a 2-set it prefers to $\{b^{j, i}_\ell, c^{i,j}_\ell\}$.
  By symmetric arguments, we have that $b^{j, i}_\ell$ is not matched to a 2-set it prefers to $\{b^{i,j}_\ell, c^{i,j}_\ell\}$.
  Every 2-set which~$c^{i,j}_\ell$ prefers to~$\{b^{i,j}_\ell, b^{j, i}_\ell\}$ contains an agent outside~$B^{i,j}$ (which is not matched to $c^{i,j}_\ell$ by \Cref{lem:vertex-sel,lem:incidence-checking}), an agent from~$X$ (which is not matched to $c^{i, j}_\ell$), an agent~$z_r^{i,j}$ for some $r \in [3]$ (which is not matched to $c^{i,j}_\ell$ as~$\{z_1^{i,j}, z_2^{i,j}, z_3^{i,j}\} \in M$), or an agent $b^{i,j}_{\ell'}, b^{j, i}_{\ell'}$, or $c^{i,j}_{\ell'}$ for some $\ell ' < \ell$ (by the induction hypothesis on~$\ell$).
  Thus, if $\{b^{i,j}_\ell  , b^{j, i}_\ell, c^{i,j}_\ell\} \notin M$, then $\{b^{i,j}_\ell, b^{j, i}_\ell, c^{i,j}_\ell\}$ blocks~$M$, contradicting the stability of~$M$.
  
  We now show that the 3-set $t= \{z^{i, j}_1, b^{i,j}_{{r+1}}, b^{j, i}_{r}\}$ blocks~$M$.
  Agent~$z^{i,j}_1$ prefers~$\{b^{i,j}_{r+1}, b^{j, i}_r\}$ to~$M(z^{i,j}_1) = \{z^{i,j}_2, z^{i,j}_3\}$.
  Every 2-set which~$b_{r+1}^{i,j}$ prefers to $\{z^{i,j}_1, b^{j, i}_r\}$ contains an agent outside~$B^{i,j}$, an agent from $X = \{x^{i,j}, \overleftarrow x^{i,j}, x^{j, i}, \overleftarrow x^{j, i}\}$, or~$z^{i,j}_1$.
  Since $b_{r+1}^{i,j}$ is not matched to a 2-set containing any of these agents, it follows that $b_{r+1}^{i,j}$ prefers $\{z^{i,j}_1, b^{j, i}_r\}$ to $M(b_{r+1}^{i,j})$.
  Every 2-set which~$b_r^{i,j}$ prefers to $\{z^{i,j}_1, b^{j, i}_{r+1}\}$ contains an agent outside~$B^{i,j}$, an agent from $X$, agent~$z^{i,j}_1$, or two agents from~$\{b^{i,j}_q, b^{j, i}_q, c^{i,j}_q\}$.
  From all these agents, $b^{i,j}_r$ can only be matched to $c^{i,j}_r$.
  It follows that $b^{i,j}_r$ prefers $\{z^{i,j}_1, b^{i,j}_{r+1}\}$ to~$M (b^{i,j}_r)$.
  Thus, $t$ blocks $M$, a contradiction to the stability of~$M$.

  If $s < r $, then symmetric arguments show that $\{z^{i,j}_1, b^{i,j}_{s}, b_{{s+1}}^{j, i}\}$ is a blocking 3-set for $M$.
\end{proof}

It now easily follows that $G$ contains a multicolored clique.

\begin{lemma}\label{lbackward2}
  If $\mathcal{I}'$ admits a stable matching, then $G$ admits a clique of size $k$.
\end{lemma}

\begin{proof}
  By \Cref{lem:vertex-sel}, every vertex-selection gadget selects a vertex, and by \Cref{licg}, these $k$ vertices form a clique.
\end{proof}

\Cref{tWh} now directly follows from \Cref{lforward2,lbackward2,okappa}.

\begin{theorem}\label{tWh}
  \textsc{3-DSR} parameterized by $\lambda (\mathcal{I})$ is \Wone-hard, where
  $\lambda (\mathcal{I})$ denotes the minimum number of agents such that the preferences of the instance arising through the deletion of these agents are derived from a strict order.
\end{theorem}

\begin{proof}
  The reduction clearly runs in polynomial time.
  By \Cref{lforward2,lbackward2} it is correct.
  By \Cref{okappa}, we have that $\lambda (\mathcal{I}) = O(k^2)$.
  Thus, we have a parameterized reduction from \textsc{Multicolored Clique} parameterized by solution size~$k$ to \textsc{3-DSR} parameterized by $\lambda (\mathcal{I})$.
  Since \textsc{Multicolored Clique} parameterized by solution size~$k$ is \Wone-hard~\cite{DBLP:series/txcs/DowneyF13,Pietrzak03}, \textsc{3-DSR} is \Wone-hard parameterized by $\lambda (\mathcal{I})$.
\end{proof}

We have shown that \tdsrp\ parameterized by $\lambda $ is \Wone-hard.
A natural question is whether there is an \XP-algorithm for \tdsrp.

Next, we drop the assumption that preferences are complete (i.e., every agent is allowed to be matched to any set of $d-1$ other agents), but instead require that the master poset is a strict order.

\subsection{Incomplete preferences derived from a strictly ordered master poset}
\label{sec:incomplete}

Let \textsc{MDSRI} be the \textsc{MDSR} problem with incomplete preference lists, i.e., $\succ_a$ is not a strict order of $\binom{A\setminus \{a\}}{d -1 }$, but a strict order of a subset $X_a \subseteq \binom{A\setminus \{a\}}{d-1}$ for each~$a\in A$.
In this case, we call a $d$-set~$t$ \emph{acceptable} if $t\setminus \{a\}\in X_a$ for every $a\in t$, and define a \emph{matching}~$M$ to be a set of disjoint, acceptable $d$-sets.
\defProblemTask{\textsc{MDSRI}}
{A set $A$ of agents together with preference lists $\succ_a$ over $X_a$ for a subset $X_a \subseteq \binom{A\setminus \{a\}}{d-1}$ for each $a\in A$.}
{Decide whether a stable matching exists.}
Similarly, \textsc{MDSRI-ML} is the \textsc{MDSRI} problem restricted to instances where the preferences are derived from a strict order of agents (which we call \emph{master poset}), and \textsc{$\ell$-DSRI} is \textsc{MDSRI} for the special case $d = \ell$.
Here, preferences of an agent~$a$ are derived from the master poset~$\succ_{\ML}$ if they are the restriction to~$X_a$ of a preference list derived from $\succ_{\ML}$.
\defProblemTask{\textsc{MDSRI-ML}}
{An \textsc{MDSRI} instance, and a strict order $\succ_{\ML}$ of the agents (called \emph{master poset}) such that for each agent~$a$, $\succ_a$ arises from $\succ_{\ML}$ through the deletion of some $(d-1)$-sets.}
{Decide whether there exists a stable matching.}

In this section, we show that \textsc{3-DSRI-ML}, the restriction of \textsc{MDSRI-ML} to $d = 3$, is \NP-complete, even if the master poset is strictly ordered.
In order to do so, we reduce from \textsc{Perfect-SMTI-ML}.
The input of this problem is an instance of \textsc{Maximum Stable Marriage with Ties and Incomplete Preferences}, where the preferences of men are derived from a weak order of women with maximum tie size two (called \emph{master list of women}) and the preferences of women are derived from a strict order of men (called \emph{master list of men}).
Here, the preferences of men (women) are derived from a master list $\succ_w$ ($\succ_m$) if the preferences of each man $m$ (woman $w$) arise through the deletion of a set of agents from $\succ_w$ ($\succ_m$).
\textsc{Perfect-SMTI-ML} then asks whether there exists a perfect (weakly) stable matching, i.e., a set $M $ of man-woman pairs such that every man and every woman is contained in exactly one pair, and there is no pair $(m, w)$ preferring each other to their partner assigned in the matching.
Note that for \textsc{Perfect-SMTI-ML}, we denote the assignments of a matching (as well as blocking pairs) as pairs in order to avoid confusion with 2-sets contained in the preferences of an agent in a \textsc{3-DSRI-ML} instance.
\defProblemTask{\textsc{Perfect-SMTI-ML}}
{
A \textsc{Stable Marriage with Ties and Incomplete Preferences} instance, where the preferences are derived from two master lists $\succ_w$ (which is a strict order) and $\succ_m$ (which may contain ties of size at most two).
}
{Decide whether there exists a perfect stable matching.}
\textsc{Perfect-SMTI-ML} is \NP-complete~\cite{IMS08}.

For the rest of this section, we fix a \textsc{Perfect-SMTI-ML} instance $\mathcal{I} = (G, \succ_m, \succ_w)$, where $G$ is the acceptability graph (i.e., the graph where each agent is a vertex, and two agents are connected by an edge if and only if they are contained in each other's preference list), and $\succ_m$ and $\succ_w$ are the master lists of men and women, respectively.
We denote the set of men by~$U$, and the set of women by~$W$.
We assume that $|U| = |W|$.
Let $W = \{w_1,\dots, w_{|W|}\}$ such that $w_i \succ_m w_{i+1}$ or $w_i \perp_m w_{i+1}$ for all $i\in [|W|- 1]$ and let $U = \{m_1, \dots, m_{|U|}\}$ such that $m_i \succ_w m_{i+1}$ for all~$ i\in [|U|-1]$.

The basic idea of the reduction is as follows.
For each~$m_i$, we add an agent $a_i$, and for each~$w_j$, we add an agent $b_j$.
For every acceptable pair~$(m_i, w_j)$, we add an agent $c_{i,j}$, and the 3-set~$\{a_i, b_j, c_{i,j}\}$ will be acceptable.
Intuitively, a stable 3-dimensional matching matches each agent~$a_i$ to a 2-set~$\{b_j, c_{i,j}\}$, which corresponds edge~$(m_i, w_j)$ being part of a stable matching.
Thus, the preferences of $a_i$ and~$b_j$ correspond to those of $m_i$ and $w_j$, i.e., $a_i$ prefers $\{b_j, c_{i,j}\}$ to~$\{b_{j'}, c_{i,j'}\}$ if $m_i $ prefers $w_j$ to $w_{j'}$ (but $a_i$ preferring $\{b_j, c_{i,j}\}$ to~$\{b_{j'}, c_{i,j'}\}$ does not imply $m_i$ prefering $w_j$ to $w_{j'}$ as $m_i$ may tie $w_j$ and $w_{j'}$), and $b_j$ prefers $\{a_i, c_{i,j}\}$ to $\{a_{i'}, c_{i',j}\}$ if and only if $w_j$ prefers~$m_i$ to $m_{i'}$.
However, the preferences of men may contain ties (of size two), while the preferences of $a_i$ must not contain ties.
We will use so-called \emph{tie gadgets} to model such ties.
Finally, the reduction shall ensure that every $a_i$ is matched (as this implies that every $m_i$ is matched).
This will be done by a \emph{cut-off gadget}.

We now describe the two gadgets (tie gadget and cut-off gadget) used in  the reduction.
Afterwards, we describe the reduction in detail and prove its correctness.

\subsubsection{Tie gadget}\label{sec:tie-gadget}
Given a man $m_i\in U$ who ties two women $w_j $ and $w_{j+1}$, we construct a tie gadget $T^{j}_i$.
This gadget models this tie, i.e., it allows~$a_i$ to be matched to $b_j$ or $b_{j+1}$.
The idea
is the following:
There are two stable matchings inside the gadget, one leaving $c_{i,j}$ unmatched while the other matches~$c_{i,j}$.
The first one allows to match $m_i$ to $w_j$ via the 3-set $\{a_i, b_j, c_{i,j}\}$, while the second allows to match $m_i$ to $w_{j+1}$ via $\{a_i, b_{j+1}, c_{i,j+1}\}$ (note that in this case $c_{i,j}$ prevents the 3-set~$\{a_i, b_j, c_{i,j}\}$ from being blocking).
In this case, the 3-set $\{a_i, b_j, c_{i,j}'\}$ ensures that if $\{a_i, b_{j+1}, c_{i, j+1}\}$ is not part of the matching, then the 3-set $\{a_i, b_j, c_{i,j}'\}$ can be blocking to represent the possibly blocking pair $(m_i, w_j)$.

We add nine agents $c_{i,j}'$ and $d_{i,j}^1, \dots, d_{i,j}^8$, together with the acceptable 3-sets $\{a_i, c_{i,j}', b_{j+1}\}$, $\{c_{i,j}, d_{i,j}^5, d_{i,j}^8\}$, $\{d_{i,j}^1, d_{i,j}^2, d_{i,j}^8\}$, $\{d_{i,j}^1, d_{i,j}^4, d_{i,j}^6\}$, $\{d_{i,j}^2, d_{i,j}^3, d_{i,j}^7\}$, and $\{d_{i,j}^3, d_{i,j}^4, d_{i,j}^5\}$.
See \Cref{fig:tg} for an example.

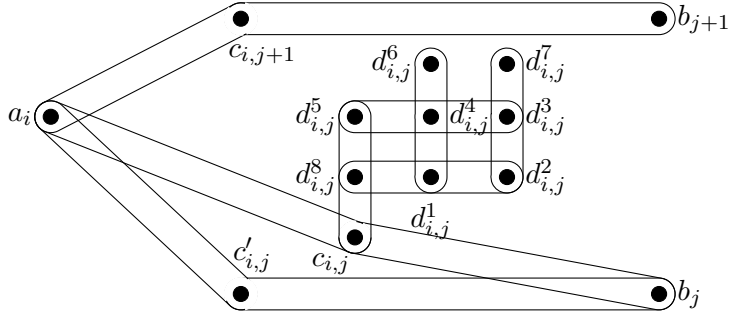
\begin{figure}[t]
  \begin{center}
    \begin{tikzpicture}
      \node[vertex, label=180:$a_i$] (a) at (-2,0) {};
      \node[vertex, label={[xshift = 0.28cm, yshift = -0.1cm]270:$c_{i,j+ 1}$}] (c1) at (0.5,1.3) {};
      \node[vertex, label={[xshift = 0.15cm, yshift = 0.05cm]90:$c_{i,j}'$}] (c2p) at (0.5, -2.35) {};
      \node[vertex, label=0:$b_{j+1}$] (bj) at ($(c1) + (5.5, 0)$) {};
      \node[vertex, label=0:$b_j$] (bl) at ($(c2p) + (5.5,0)$) {};
      \node[vertex, label={[xshift = 0.15cm, yshift = 0.0cm]225:$c_{i,j}$}] (c2) at (2,-1.6) {};
      \begin{scope}[xshift = 2cm]
      \node[vertex, label={[xshift = 0.cm, yshift = -0.1cm]270:$d^1_{i,j}$}] (d1) at (1,-0.8) {};
      \node[vertex, label=0:$d^2_{i,j}$] (d2) at ($(d1) + (1, 0)$) {};
      \node[vertex, label=0:$d^3_{i,j}$] (d3) at (2,0.) {};
      \node[vertex, label=0:$d^4_{i,j}$] (d4) at (1, -0) {};
      \node[vertex, label=180:$d^5_{i,j}$] (d5) at (0, 0) {};
      \node[vertex, label=180:$d^6_{i,j}$] (d6) at (1,0.7) {};
      \node[vertex, label=0:$d^7_{i,j}$] (d7) at ($(d6) + (1, 0)$) {};
      \node[vertex, label=180:$d^8_{i,j}$] (d8) at ($(d1) + (-1, 0)$) {};
      \end{scope}

      \begin{pgfonlayer}{background}
      \draw \convexpath{a,c1}{6pt};
      \draw \convexpath{c1,bj}{6pt};
      \draw \convexpath{bl,c2}{6pt};
      \draw \convexpath{a,c2}{6pt};
      \draw \convexpath{bl,c2p}{6pt};
      \draw \convexpath{a,c2p}{6pt};
\tikzset{
  pics/carc/.style args={#1:#2:#3}{
    code={
      \draw[pic actions] (#1:#3) arc(#1:#2:#3);
    }
  }
}

 \draw[white, thick] (c2) ++ (0:6pt) arc (0:60:6pt);
 \draw[white, thick] (c2) pic{carc=85:180:6pt};
 \draw[white, thick] (c2p) pic[white]{carc= 70:205:6pt};
 \draw[white, thick] (c2p) pic[white]{carc= -75:65:6pt};
 \draw[white, thick] (c1) pic{carc=-70:75:6pt};
 \draw[white, thick] (c1) pic{carc=135:280:6pt};

 \draw[white, thick] (c2) ++ (0:5pt) arc (0:60:5pt);
 \draw[white, thick] (c2) pic{carc=85:180:5pt};
 \draw[white, thick] (c2p) pic[white]{carc= 70:205:5pt};
 \draw[white, thick] (c2p) pic[white]{carc= -75:65:5pt};
 \draw[white, thick] (c1) pic{carc=-70:75:5pt};
 \draw[white, thick] (c1) pic{carc=135:280:5pt};

\draw[white, thick] (c2) pic{carc=0:55:5.5pt};

\draw[white, thick] (c2) ++ (0:5.5pt) arc (0:60:5.5pt);
\draw[white, thick] (c2) pic{carc=85:180:5.5pt};
\draw[white, thick] (c2p) pic[white]{carc= 70:205:5.5pt};
\draw[white, thick] (c2p) pic[white]{carc= -75:65:5.5pt};
\draw[white, thick] (c1) pic{carc=-70:75:5.5pt};
\draw[white, thick] (c1) pic{carc=135:280:5.5pt};

      \draw \convexpath{d1,d4,d6}{6pt};
      \draw \convexpath{d5,d4,d3}{6pt};
      \draw \convexpath{d2,d3,d7}{6pt};
      \draw \convexpath{d1,d2,d8}{6pt};
      \draw \convexpath{c2,d8,d5}{6pt};
      \end{pgfonlayer}

        \begin{scope}[on background layer]
          \newcommand{\colorBetweenTwoNodes}[3]{
            \fill[#1] ($(#2) + (0, .08)$) to ($(#2) - (0, .08)$) to ($(#3) - (0,.08)$) to ($(#3) + (0,.08)$) -- cycle;
        }
        \end{scope}
    \end{tikzpicture}

  \end{center}
  \caption{The acceptable 3-sets of a tie gadget $T_i^{j}$.
  For example, the line around $a_i$, $c_{i,j+1}$, and~$b_{j+1}$ indicates that the 3-set $\{a_i, c_{i,j+1}, b_{j+1}\}$ is acceptable.}
  \label{fig:tg}
\end{figure}

The preferences of any agent arise from the following preferences through the deletion of all 2-sets which are not acceptable for an agent.
$\{d_{i,j}^1, d_{i,j}^2\} \succ \{d_{i,j}^1, d_{i,j}^4\} \succ \{d_{i,j}^2, d_{i,j}^3\} \succ \{d_{i,j}^3, d_{i,j}^4\} \succ \{d_{i,j}^1, d_{i,j}^6\} \succ \{d_{i,j}^3, d_{i,j}^5\} \succ\{d_{i,j}^4, d_{i,j}^5\} \succ \{d_{i,j}^2, d_{i,j}^7\} \succ \{d_{i,j}^3, d_{i,j}^7\} \succ \{d_{i,j}^1, d_{i,j}^8\} \succ \{d_{i,j}^2, d_{i,j}^8\} \succ \{d_{i,j}^4, d_{i,j}^6\} \succ \{d_{i,j}^5, d_{i,j}^8\} \succ \{d_{i,j}^5, c_{i,j}\} \succ \{d_{i,j}^8, c_{i, j}\} \succ \{a_i, c_{i, j}\}\succ \{a_i, c_{i,j+1}\} \succ \{a_i, c_{i, j}'\} \succ \{b_j, c_{i, j }\} \succ \{b_{j+ 1}, c_{i, j + 1}\} \succ \{b_j, c_{i, j}'\}$, which can be derived from the following strict order of agents: $d_{i,j}^1 \succ d_{i,j}^2\succ\dots \succ d_{i,j}^8 \succ a_i \succ b_j \succ b_{j+1} \succ c_{i, j} \succ c_{i,j+1} \succ c_{i, j}'$.

The following observation shows that the tie gadget indeed models ties, i.e., it contains a stable matching which matches $a_i$ to $ b_j$ (corresponding to matching~$m_i$ to~$w_j$) and one which matches $a_i$ to $b_{j+1}$ (corresponding to matching~$m_i$ to~$w_{j+1}$).
Furthermore, given a matching $M$ which matches $a_i$ or both $w_j$ and $w_{j+1}$ to 2-sets they prefer to every 2-set of the tie gadget, we can extend~$M$ to the tie gadget without introducing a blocking 3-set.

We consider $a_i$, $b_j$, and $b_{j+1}$ to be part of the tie gadget~$T^j_i$.
Note that $a_i$, $b_j$, and $b_{j+1}$ may also be part of other tie gadgets.
For a set~$X$ of agents, we denote by $T^{j}_i - X$ the instance arising from $T^j_i$ through the deletion of all agents from $X$ as well as every 2-sets containing an agent from $X$ which appears in the preferences of some agent.

\begin{observation}\label{otg}
  Let $T^{j}_i$ be a tie gadget.
  The matchings $M_1 = \{\{a_i, c_{i,j+1}, b_{j+1}\},\allowbreak \{c_{i,j}, d_{i,j}^5, d_{i,j}^8\},\allowbreak \{d_{i,j}^2, d_{i,j}^3,d_{i,j}^7\},\{d_{i,j}^1, d_{i,j}^4, d_{i,j}^6\}\}$ and $M_2 = \{ \{a_i, c_{i, j}, b_j\}, \{d_{i,j}^1, d_{i,j}^2,d_{i,j}^8\},\allowbreak\{d_{i,j}^3, d_{i,j}^4, d_{i,j}^5\}\}$ are stable.
	In $T^{ j}_i - \{a_i\}$ or $T^{ j}_i - \{b_j, b_{j+1}\}$, also the matching $M = \{\{c_{i,j}, d_{i,j}^5, d_{i,j}^8\},\allowbreak \{d_{i,j}^2, d_{i,j}^3,d_{i,j}^7\},\allowbreak\{d_{i,j}^1, d_{i,j}^4, d_{i,j}^6\}\}$ is stable.
\end{observation}

\subsubsection{Cut-off gadget}\label{sig}

A cut-off gadget for an agent~$a$ consists of $a$ together with five agents $x_2^a, \dots, x_6^a$.
The only acceptable 3-sets are $\{a, x_{5}^a, x_6^a\}$, $\{x_2^a, x_4^a, x_6^a\}$, and $\{x_3^a, x_{4}^a, x_{5}^a\}$.
See \Cref{fig:ig} for an example.
\begin{figure}[t]
  \begin{center}
    \begin{tikzpicture}
      \node[vertex, label=90:$a$] (d1) at (3,2) {};
      \node[vertex, label=270:$x_2^a$] (d2) at (2, -1) {};
      \node[vertex, label=90:$x_3^a$] (d3) at (1.,2.) {};
      \node[vertex, label=180:$x_4^a$] (d4) at (1, -0) {};
      \node[vertex, label=0:$x_5^a$] (d5) at (2, 1) {};
      \node[vertex, label=0:$x_6^a$] (d6) at (3,0) {};

      \begin{pgfonlayer}{background}
      \draw \convexpath{d5,d1,d6}{6pt};
      \draw \convexpath{d6,d2,d4}{6pt};
      \draw \convexpath{d4,d3,d5}{6pt};
\tikzset{
  pics/carc/.style args={#1:#2:#3}{
    code={
      \draw[pic actions] (#1:#3) arc(#1:#2:#3);
    }
  }
}

      \end{pgfonlayer}

        \begin{scope}[on background layer]
          \newcommand{\colorBetweenTwoNodes}[3]{
            \fill[#1] ($(#2) + (0, .08)$) to ($(#2) - (0, .08)$) to ($(#3) - (0,.08)$) to ($(#3) + (0,.08)$) -- cycle;
        }
        \end{scope}
    \end{tikzpicture}

  \end{center}
  \caption{The acceptable 3-sets of a cut-off gadget.}
  \label{fig:ig}
\end{figure}
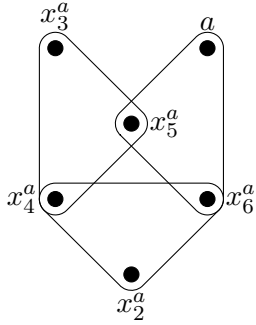

Each agent derives its preferences from $\{x_2^a, x_4^a\} \succ \{a, x_5^a\} \succ \{a, x_6^a\} \succ \{x_3^a, x_4^a\} \succ\{x_3^a, x_5^a\} \succ \{x_2^a, x_6^a\} \succ \{x_4^a, x_5^a\}\succ \{x_4^a, x_6^a\}\succ \{x_5^a, x_6^a\}$.
Note that this list can be derived from $a\succ x_2^a \succ x_3^a \succ x_4^a \succ x_5^a \succ x_6^a$.

We now observe that a cut-off gadget does not admit a stable matching, implying that one of the agents has to be matched outside the gadget.
As agent $a$ is the only agent which accepts 2-sets containing agents not contained in the cut-off gadget, it follows that $a$ is matched to a 2-set outside the cut-off gadget.

\begin{lemma}\label{lig}
  A cut-off gadget does not admit a stable matching.
\end{lemma}

\begin{proof}
  Note that any acceptable 3-set contains two agents from $\{x_4^a, x_5^a, x_6^a\}$.
  Thus, any matching contains only one 3-set of a cut-off gadget.
  For each of the three possible matchings, we give a blocking 3-set in \Cref{tig}.
  \begin{table}
    \begin{center}
      \begin{tabular}{c | c}
        Matching & Blocking 3-set \\
        \hline
        $\{a, x_5^a, x_6^a\}$ & $\{x_2^a, x_4^a, x_6^a\}$ \\
        $\{x_2^a, x_4^a, x_6^a\}$ & $\{x_3^a, x_4^a, x_5^a\}$\\
        $\{x_3^a, x_4^a, x_5^a\}$ & $\{a, x_5^a, x_6^a\}$
      \end{tabular}
	    \caption{The blocking 2-sets in the subinstance from \Cref{lig}.}
\label{tig}
    \end{center}
  \end{table}
\end{proof}

We observe that if agent $a$ is matched outside the cut-off gadget, then the cut-off gadget does not contain a blocking 3-set.

\begin{observation}\label{oig}
  The cut-off gadget without $a$ admits a stable matching, namely $\{x_3^a, x_4^a, x_5^a\}$.
\end{observation}

Having described the gadgets needed for the reduction, we can now describe the complete reduction.

\subsubsection{The reduction}

  Our reduction is structured similarly to the \NP-completeness proof of \textsc{3-Dimensional Stable Marriage with Incomplete Cyclic Preferences} by Bir\'{o} and McDermid~\cite{BM10}.
  In both reductions, there is one agent for each man and each woman.
  Each such agent is forced to be matched in any stable matching by a gadget based on a small unsolvable instance.
  However, modelling the ties in the preferences is a bit more complicated in our case, and is done by the tie gadget described in \Cref{sec:tie-gadget}.

  We construct a \textsc{3-DSRI-ML} instance $\mathcal{I}'$ with a strictly ordered master poset from an instance~$\mathcal{I}$ of \textsc{Perfect-SMTI-ML} as follows.
  For each man~$m_i$, we add an agent~$a_i$, and for each woman~$w_j$, we add an agent~$b_j$.
  For each man~$m_i$, and each woman $w_j$ who is not tied with another woman in $m_i$'s preference list, we add an agent~$c_{i,j}$.
  For each man $m_i$, and each tie $w_j \perp_m w_{j+1}$ in $m_i$'s preference list, we add a tie gadget $T^j_{i}$ (described in \Cref{sec:tie-gadget}).

  It remains to describe the preferences. 
  For each man $m_i$, we define a sublist $\mathcal{A}_i$ as follows.
  Process all woman $w_j$ adjacent to~$m_i$ by increasing $j$.
  If the woman is not tied with another woman adjacent to~$m_i$, then add the 2-set $\{a_i, c_{i,j}\}$, followed by $\{b_j, c_{i,j}\}$.
  Otherwise, $w_j$ is tied with $w_{j+1}$ in the preference list of $m_i$.
  Then add the 2-sets $\{a_i, c_{i, j}\}\succ_{\ML} \{a_i, c_{i,j+1}\} \succ_{\ML} \{a_i, c_{i, j}'\} \succ_{\ML} \{b_j, c_{i, j }\} \succ_{\ML} \{b_j, c_{i, j + 1}\} \succ_{\ML} \{b_j, c_{i, j}'\}$ to $\mathcal{A}_i$ (see \Cref{sec:tie-gadget}).
  
  For each man $m_i$, we add a cut-off gadget $I_i$ (described in \Cref{sig}) for agent $a_i$.

  We start by showing that the preferences are indeed derived from a strictly ordered master poset of agents.
  \begin{observation}\label{lem:consistent-ml}
    The master poset $\succ_{\ML}$ is derived from a strict order.
  \end{observation}

  \begin{proof}
  First, we order the tie gadgets $T_{i_1}^{ j_1}$, $T_{i_2}^{ j_2}$, \dots, $T_{i_r}^{ j_r}$ such that $i_{\ell } \le i_{\ell + 1}$ and if $i_\ell = i_{\ell + 1}$, then $j_{\ell} < j_{\ell + 1}$.
  For each tie gadget $T_{i}^{ j}$, define the sublist $\mathcal{D}_{i}^{j}$ via $d_{i,j}^1 \succ d_{i,j } ^2 \succ \dots \succ d_{i,j}^8$.
  Furthermore, we define the sublist $\mathcal{C}_i^j$ via $c_{i,j} \succ c_{i,j+1} \succ c_{i,j}'$.
  For every $i$ and $j$ such that $m_i$ does not tie $w_j$ with another woman but $w_j$ is contained in the preferences of $m_i$, we define $\mathcal{C}_i^j$ to be the sublist containing only $c_{i,j}$.
  For every $i$ and $j$ such that $w_j$ is not contained in the preferences of $m_i$ or $m_i$ ties $w_j$ with $w_{j-1}$, we define $\mathcal{C}$ to be an empty sublist.
	  The master poset now looks as follows:
  \begin{align*}
    \mathcal{D}_{i_1}^{j_1} \succ \dots \succ \mathcal{D}_{i_r}^{ j_r} \succ a_1 \succ a_2 \succ \dots \succ a_{|M|} \succ b_1 \succ \dots \succ b_{|W|}\\
    \succ \mathcal{C}_1^1 \succ \mathcal{C}_2^1 \succ \dots \mathcal{C}_{|U|}^1 \succ \mathcal{C}_1^2 \succ \dots \succ\mathcal{C}_{|U|}^2 \succ \mathcal{C}_{1}^3 \succ \dots \succ \mathcal{C}_{|U|}^{|W|}.
  \end{align*}
  It is easy to verify that the preferences are indeed derived from this master poset.
  \end{proof}
  Notably, the preferences of every agent are not only derived from the master poset of agents, but also from a master list of 2-sets (where again every agent can declare an arbitrary set of 2-sets to be unacceptable).

  Having described the construction, we continue by showing that the corresponding reduction is indeed correct.

  \subsubsection{Proof of the forward direction}

  We first show that a perfect stable matching in the \textsc{Perfect-SMTI-ML} instance~$\mathcal{I}$ implies a perfect stable matching in the constructed \textsc{3-DSRI-ML} instance~$\mathcal{I}'$.

  \begin{lemma}\label{lforwardi}
	  If the \textsc{Perfect-SMTI-ML} instance $\mathcal{I}$ admits a perfect stable matching $M$, then the \textsc{3-DSRI-ML} instance $\mathcal{I}'$ admits a stable matching.
  \end{lemma}

  \begin{proof}
    We construct a stable matching~$M'$ as follows, starting with $M' = \emptyset$.
	  For every edge~$(m_i, w_j)\in M$, we add the 3-set $\{a_i, b_j, c_{i,j}\}$ to~$M'$.
    If $w_j$ is tied with woman~$w_{j-1}$, then we add the 3-sets~$\{c_{i,\ell}, d_{i,j}^5, d_{i,j}^8\}$, $\{d_{i,j}^2, d_{i,j}^3,d_{i,j}^7\}$, and $\{d_{i,j}^1, d_{i,j}^4, d_{i,j}^6\}$.
    If $w_j$ is tied with woman~$w_{j+1}$, then we add the 3-sets~$\{d_{i,\ell}^1, d_{i,\ell}^2,d_{i,\ell}^8\}$ and $\{d_{i,\ell}^3, d_{i,\ell}^4, d_{i,\ell}^5\}$.
    For each tie gadget~$T_i^{j}$ between~$m_i$, $w_j$ and~$w_{j+1}$ such that $(m_i, w_j)\notin M$ and $(m_i, w_{j + 1})\notin M$, we add the 3-sets~$\{c_{i,j}, d_{i,j}^5, d_{i,j}^8\}$, $\{d_{i,j}^2, d_{i,j}^3,d_{i,j}^7\}$, and $\{d_{i,j}^1, d_{i,j}^4, d_{i,j}^6\}$.
	  For each cut-off gadget~$I_i$ for man~$m_i$, we add the 3-set~$\{x_3, x_4, x_5\}$ to~$M'$.

    We claim that $M'$ is a stable matching.
    Since $M$ is perfect, every agent $a_i$ is matched to a 2-set it prefers to any 2-set of agents from its cut-off gadget.
    By \Cref{oig}, we get that no agent from cut-off gadget~$I_i $ except for $a_i$ can be part of a blocking 3-set.
    For every tie gadget~$T_i^{j}$, \Cref{otg} tells us that no blocking 3-set contains only agents of the form~$d_{i,j}^k$.
    All other acceptable 3-sets are of the form $\{a_i, b_j, c_{i,j}\}$ or $\{a_i, b_j, c_{i,j}'\}$.
    First, we consider 3-sets of the form $\{a_i, b_j , c_{i,j}\}$.
    Agent~$a_i$ prefers $\{b_j, c_{i,j}\}$ to~$M'(a_i) = \{b_\ell, c_{i, \ell}\}$ if and only if $m_i$ prefers~$w_j$ to $M(m_i)$ or $m_i$ ties $w_j$ and $w_\ell$ and $\ell = j + 1$.
    However, in the latter case (i.e., $m_i$~ties~$w_j$ and~$w_\ell$, and we have $\ell = j+ 1$), we have $\{c_{i,\ell}, d_{i,j}^5, d_{i,j}^8\}\in M'$, and thus $c_{i_\ell}$ does not prefer~$\{a_i, b_\ell\}$ to~$M' (c_{i_\ell})$.
    Agent $b_j$ prefers $\{a_i, c_{i, j}\}$ to~$M' (b_j)$ if and only if $w_j$ prefers $m_i$ to $M( w_j)$.
    Thus, by the stability of $M$, $\{a_i, b_j, c_{i, j}\}$ is not blocking.
    Next, we consider an acceptable 3-set~$\{a_i, b_j, c_{i,j}'\}$.
    Agent~$a_i$ prefers~$\{b_j, c_{i,j}'\}$ to $M' (a_i)$ if and only if $m_i$ prefers~$w_j $ to $M (m_i)$.
    Similarly, agent~$b_i$ prefers $\{a_i, c_{i, j}'\}$ to $M' (b_j)$ if and only if $w_j $ prefers~$m_i$ to $M (w_j)$.
    Therefore, the stability of $M$ implies that $\{a_i, b_j, c_{i,j}'\}$ is not blocking.
    Altogether, $M'$ is stable.
  \end{proof}

\subsubsection{Proof of the backward direction}

We now turn to the reverse direction, i.e., showing that a stable matching in the \textsc{3-DSRI-ML} instance implies a perfect stable matching in the \textsc{Perfect-SMTI-ML} instance.

\begin{lemma}\label{lbackwardi}
	If the \textsc{3-DSRI-ML} instance $\mathcal{I}'$ admits a stable matching $M'$, then the \textsc{Perfect-SMTI-ML} instance $\mathcal{I}$ admits a perfect stable matching.
\end{lemma}

\begin{proof}
  Let $M'$ be a stable matching in $\mathcal{I}'$.
  By \Cref{lig}, each agent $a_i$ has to be matched to a 2-set outside its cut-off gadget.
  Any such 2-set involves an agent $b_j$.
  Thus, this defines a perfect matching $M\coloneqq \{(m_i, w_j) : \exists v \text{ s.t.\ } \{a_i, b_j, v\}\in M'\}$.

  We claim that $M$ is stable.
  Assume that $M$ admits a blocking pair $(m_i, w_j)$.
  If $m_i$ does not tie $w_j$ with another woman, then $\{a_i, b_j, c_{i,j}\}$ is a blocking 3-set ($a_i$ and $b_j$ prefer this 3-set as $(m_i, w_j)$ is blocking, and $c_{i,j}$ as it is the only acceptable 3-set for $c_{i,j}$).
  If~$m_i$ ties $w_j$ with a woman~$w_\ell$, then $M$ cannot contain one of the edges $(m_i, w_j)$ or $(m_i, w_\ell)$ (as else $(m_i, w_j)$ was not blocking).
  Thus, $M'$ does not contain the 3-set $t\coloneqq \begin{cases}
\{a_i, b_j, c_{i,j}\} & \text{ if }  \ell = j - 1\\
\{a_i, b_{j}, c_{i,j}'\} & \text{ if } \ell = j + 1
  \end{cases}$.
  Since $c_{i,j}$ if $\ell = j -1$ or $c_{i,j}'$ if $\ell = j + 1$ is unmatched, and~$(m_i, w_j)$ is a blocking pair, we get that $t$ is a blocking 3-set, contradicting the stability of $M'$.
\end{proof}

The \NP-completeness of \textsc{3-DSRI-ML} now easily follows.

\begin{theorem}\label{thm:incomplete-w-h}
  \textsc{3-DSRI-ML} is \NP-complete, even if the master poset is a strict order.
\end{theorem}

\begin{proof}
	\textsc{3-DSRI-ML} is clearly in \NP, as we can check the stability of a matching in $O(n^3) $ time by checking for every set of three agents whether it is blocking, where $n$ is the number of agents.
	The reduction adds an agent $a_i$ together with a cut-off gadget or $b_j$ for every man~$m_i$ or woman~$w_j$.
	Furthermore, for every acceptable pair~$(m_i, w_j)$, we add an agent $c_{i,j}$ together with the acceptable 3-set~$\{a_i, b_j, c_{i,j}\}$ or a tie gadget.
	As tie and cut-off gadgets have constant size, the reduction can clearly be performed in linear time.
	The correctness of the reduction is proven in \Cref{lforwardi,lbackwardi}.
  \Cref{lem:consistent-ml} shows that the master poset is a strict order.  
\end{proof}

	Note that every acceptable 3-set contains exactly one agent from $A:= \{a_i : i \in [|U|]\} \cup \{d_{i,j}^8, d_{i,j}^4, d_{i,j}^7\} \cup \{x_4^{a_i} : i \in [|U|]\}$, one agent from $B:= \{b_j : j\in [|W|]\} \cup \{d^5_{i,j}, d^6_{i,j}, d^2_{i,j}\}\cup \{x_2^{a_i}, x_5^{a_i} : i \in [|U|]\}$, and one agent from $C:= \{c_{i,j}, c_{i,j'}, c_{i,j + 1}, d^1_{i,j}, d^3_{i,j}\}\cup \{x_3^{a_i}, x_6^{a_i} : i \in [|U|]\}$.
	Thus, \Cref{thm:incomplete-w-h} also shows \NP-completeness for the tripartite version of \textsc{3-DSRI-ML}.

	By ``cloning'' each agent corresponding to a man $d-3$ times (and for each ``acceptable 3-set'', adding the cloned men to this 3-set, and adding all $(d-1)$-subsets of the resulting $d$-set at their corresponding place in the preferences), one can derive \NP-completeness of \textsc{$d$-DSRI-ML} for any fixed $d \ge 3$.

\section{Conclusion}
\label{sec:conclusion}
Being a fundamental problem within the fields of stable matching~\cite{Manlove13} and
the analysis of hedonic games~\cite{AS16}, our work provides a seemingly first
systematic study on the parameterized complexity of the NP-hard
\textsc{Multidimensional Stable Roommates}. 
Focusing on the natural and well-motivated
concept of master lists with the goal to identify efficiently solvable
special cases, we reported partial success.
While we have one main algorithmically positive result, namely
fixed-parameter tractability for the parameter ``maximum number of
agents incomparable to a single agent'', all other (single) parameterizations
led to (often surprising) hardness results (see Table~\ref{tab:results}).

As to challenges for future research,
first, it remains open whether our fixed-parameter tractability
result mentioned above also transfers to the setting
of \textsc{Multidimensional Stable Marriage}.
Second, addressing the quest for identifying more islands of
tractability, the study of further, perhaps also combined parameters
is a worth-while goal.
One possible parameter here would be to consider the setting that there are few strictly ordered master posets, and every agent derives its preferences of one of the master lists.
Third, a natural open question is whether our W[1]-hardness results can be accompanied by XP-algorithms, or whether \tdsrp\ is paraNP-hard for the considered parameters.

\section*{Acknowledgements}
\acktext

\bibliographystyle{abbrv}
\bibliography{3DSM}

\end{document}